\documentclass[12pt]{article}
 \usepackage{amsmath, amsthm, amssymb, bbm, setspace,bigints}
 \usepackage[margin=1 in]{geometry}
\usepackage{caption, enumitem}
\usepackage{subcaption}
\usepackage[toc,page]{appendix}
\usepackage{graphicx, amsfonts, tikz, multirow}
\usepackage{natbib}
\usetikzlibrary{bayesnet}
\usepackage{pdfpages}
\usepackage{epstopdf}
\usepackage{booktabs}
\usepackage[ruled]{algorithm2e}
\usepackage[colorlinks,citecolor=blue,urlcolor=blue]{hyperref}
\usepackage{rotating}
\usepackage[utf8]{inputenc}
\usepackage{authblk}
\usepackage{float}

\def\Gauss{{ \mathrm{N} }}
\def\T{\mathrm{\scriptscriptstyle{T}}}

\def\E{E_{P_{f_0}}}
\def\|{\mid}

\newtheorem{theorem}{Theorem}[section]
\newtheorem{lemma}[theorem]{Lemma}
\newtheorem{assumption}[theorem]{Assumption}

\newtheorem{proposition}[theorem]{Proposition}

\DeclareMathOperator{\etr}{etr}
\date{}
\title{Nearest Neighbor Dirichlet Mixtures}

\author[1]{Shounak Chattopadhyay\thanks{shounak.chattopadhyay@duke.edu}}
\author[2]{Antik Chakraborty\thanks{antik015@purdue.edu}}
\author[1]{David B. Dunson \thanks{dunson@duke.edu}}
\affil[1]{Department of Statistical Science, Duke University}
\affil[2]{Department of Statistics, Purdue University}

\begin{document}
\maketitle

\begin{abstract}
There is a rich literature on Bayesian
methods for density estimation, which characterize the unknown density as a mixture of kernels.  Such methods have advantages in terms of providing uncertainty quantification in estimation, while being adaptive to a rich variety of densities.  However, relative to frequentist locally adaptive kernel methods, Bayesian approaches can be slow and unstable to implement in relying on Markov chain Monte Carlo algorithms.  To maintain most of the strengths of Bayesian approaches without the computational disadvantages, 
 we propose a class of nearest neighbor-Dirichlet mixtures. The approach starts by grouping the data into neighborhoods based on standard algorithms.  Within each neighborhood, the density is characterized via a Bayesian parametric model, such as a Gaussian with unknown parameters.  Assigning a Dirichlet prior to the weights on these local kernels, we obtain a  pseudo-posterior for the weights and kernel parameters.  A simple and embarrassingly parallel Monte Carlo algorithm is proposed to sample from the resulting pseudo-posterior for the unknown density. Desirable asymptotic properties are shown, and the methods are evaluated in simulation studies and applied to a motivating data set in the context of classification. 
\end{abstract}


\section{Introduction}\label{sec:intro}

Bayesian nonparametric methods provide a useful alternative to black box machine learning algorithms, having potential advantages in terms of characterizing uncertainty in inferences and predictions.  However, computation can be slow and unwieldy to implement.  Hence, it is important to develop simpler and faster Bayesian nonparametric approaches, and {\em hybrid} methods that borrow the best of both worlds.  For example, if one could use the Bayesian machinery for uncertainty quantification and reduction of mean square errors through shrinkage, while incorporating algorithmic aspects of machine learning approaches, one may be able to engineer a highly effective hybrid. The focus of this article is on proposing such an approach for density estimation, motivated by the successes and limitations of nearest neighbor algorithms and Bayesian mixture models. 

Nearest neighbor algorithms are popular due to a combination of simplicity and performance. Given a set of $n$ observations $\mathcal{X}^{(n)} = (X_1,\ldots, X_n)$ in $\mathbb{R}^p$, the density at $x$ is estimated as $\hat{f}_{\text{knn}}(x) = k/(nV_p R_{k}^{p})$, where $k$ is the number of neighbors of $x$ in $\mathcal{X}^{(n)}$, $R_k = R_k(x)$ is the distance of $x$ from its $k$th nearest neighbor in $\mathcal{X}^{(n)}$, and $V_p$ is the volume of the $p$-dimensional unit ball \citep{loftsgaarden1965nonparametric, mack1979multivariate}. Refer to \cite{biau2015lectures} for an overview of related estimators and corresponding theory.  

Nearest neighbor density estimators are a type of locally adaptive kernel density estimators. The literature on such methods identifies two broad classes: {\bf balloon estimators} and {\bf sample smoothing estimators}; see \cite{scott2015multivariate, terrell1992variable} for an overview. {\bf Balloon estimators} characterize the density at a query point $x$ using a bandwidth function $h(x)$; classical examples include the naive $k$-nearest neighbor density estimator \citep{loftsgaarden1965nonparametric} and its modification in \cite{mack1979multivariate}. More elaborate balloon estimators face challenges in terms of choice of $h(x)$ and obtaining density estimators that do not integrate to $1$.
{\bf Sample smoothing estimators} use $n$ different bandwidths $h(X_i)$, one for each sample point $X_i$, to estimate the density at a query point $x$ globally. 
By construction, sample smoothing estimators are {\em bona fide} density functions integrating to $1$. To fit either the balloon or the sample smoothing estimator, one may compute an initial pilot density estimator employing a constant bandwidth and then use this pilot to estimate the bandwidth function \citep{breiman1977variable, abramson1982bandwidth}. Another example of a locally adaptive density estimator is the local likelihood density estimator \citep{loader1996local, loader2006local, hjort1996locally}, which fits a polynomial model in the neighborhood of a query point $x$ to estimate the density at $x$, estimating the parameters of the local polynomial by maximizing a penalized local log-likelihood function.
The above methods produce a point estimate of the density without uncertainty quantification (UQ).

Alternatively, there is a Bayesian literature on locally adaptive kernel methods, which express the unknown density as:
\begin{align}
f(x) = \sum_{h=1}^m \pi_h \mathcal{K}(x; \theta_h),\quad 
\theta_h \sim P_0,\quad 
(\pi_h)_{h=1}^m \sim Q_0,
\label{eq:SB}
\end{align}
which is a mixture of $m$ components, with the $h$th having probability weight $\pi_h$ and kernel parameters $\theta_h$; by allowing the location and bandwidth to vary across components, local adaptivity is obtained.  A Bayesian specification is completed with prior $P_0$ for the kernel parameters and $Q_0$ for the weights.  In practice, it is common to rely on an over-fitted mixture model \citep{rousseau2011asymptotic}, 
which chooses $m$ as a pre-specified finite upper bound on the number of components, and lets 
\begin{eqnarray}
\label{eq:Dirichlet}
\pi = (\pi_1,\ldots,\pi_m)^{\T} \sim \mbox{Dirichlet}(\alpha,\ldots,\alpha).
\end{eqnarray}
Augmenting component indices $c_i \in \{1,\ldots,m\}$ for $i=1,\ldots,n$, a simple Gibbs sampler can be used for posterior computation, alternating between sampling (i) $c_i$ from a multinomial conditional posterior, for $i=1,\ldots,n$; (ii) $\theta_h \mid - \sim 
P_0(\theta_h)\prod_{i:c_i=h}\mathcal{K}(X_i;\theta_h)$; and (iii) $\pi \mid - \sim \mbox{Dirichlet}(\alpha + n_1,\ldots,\alpha+n_m),$ with $n_h = \sum_{i=1}^n \mathbb{I}(c_i=h)$ for $h=1,\ldots,m$.

Relative to frequentist locally adaptive methods, Bayesian approaches are appealing in automatically providing a characterization of uncertainty in estimation, while having excellent practical performance for a broad variety of density shapes and dimensions.  However, implementation typically relies on Markov chain Monte Carlo (MCMC), with the Gibbs sampler sketched above providing an example of a common algorithm used in practice. Unfortunately, current MCMC algorithms for posterior sampling in mixture models tend to face issues with
{\em slow mixing}, meaning the sampler can take a very large number of iterations to adequately explore different posterior modes and obtain sufficiently accurate posterior summaries.



MCMC inefficiency has motivated a literature on faster approaches, including 
sequential approximations \citep{wang2011fast,zhang2014sequential} and variational Bayes \citep{blei2006variational}.  These methods are order dependent, tend to converge to local modes, and/or lack theory support. \cite{newton1999recursive, newton2002nonparametric} instead rely on predictive recursion.  Such estimators are fast to compute and have theory support, but are also order dependent and do not provide a characterization of uncertainty.  Alternatively, one can use a Polya tree as a conjugate prior \citep{lavine1992some, lavine1994more}, and there is a rich literature on related multiscale and recursive partitioning approaches, such as the optional Polya tree \citep{wong2010optional}. However, Polya trees have disadvantages in terms of sensitivity to a base partition and a tendency to favor spiky/erratic densities.  These disadvantages are inherited by most of the computationally fast modifications. 

This article develops an alternative to current locally adaptive density estimators, obtaining the practical advantages of Bayesian approaches in terms of uncertainty quantification and a tendency to have relatively good performance for a wide variety of true densities, but without the computational disadvantage due to the use of MCMC.  This is accomplished with a {\em Nearest Neighbor-Dirichlet Mixture} (NN-DM) model.
The basic idea is to rely on fast nearest neighbor search algorithms to group the data into local neighborhoods, and then use these neighborhoods in defining a Bayesian mixture model-based approach. Section \ref{sec:methodology} 
outlines the NN-DM approach and describes implementation details for Gaussian kernels.
Section \ref{sec:theory} provides some theory support for NN-DM.  Section \ref{sec:simulation} contains simulation experiments comparing NN-DM with a rich variety of competitors in univariate and multivariate examples, including an assessment of UQ performance.  Section \ref{sec:application-pulsar} contains a real data application, and Section \ref{sec:discussion} a discussion.    


\section{Methodology}\label{sec:methodology}

\subsection{Nearest Neighbor Dirichlet Mixture Framework}\label{subsection:NNDPframework} 
Let $d(x_1,x_2)$ denote a distance metric between data points $x_1,x_2 \in \mathcal{X}$.  For $\mathcal{X} = \mathbb{R}^p$, the Euclidean distance is typically chosen. For each $i \in \{1, 2, \ldots, n\}$, let $X_{i[j]}$ denote the $j$th nearest neighbor to $X_i$ in the data $\mathcal{X}^{(n)} = (X_1,\ldots,X_n)$ such that 
$d(X_i,X_{i[1]}) \leq \ldots \leq d(X_i,X_{i[n]})$, with ties broken by increasing order of indices. By convention, we define $X_{i[1]} = X_i$.  The indices on the $k$ nearest neighbors to $X_i$ are denoted as $\mathcal{N}_i = \{ j: d(X_i,X_j) \leq d(X_i, X_{i[k]}) \}$. Denote the set of data points in the $i$th neighborhood by $\mathcal{S}_i = \{X_j : j \in \mathcal{N}_i\}$. In implementing the proposed method, we typically let the number of neighbors $k$ vary as a function of $n$. When necessary, we use the notation $k_n$ to express this dependence.  However, we routinely drop the $n$ subscript for notational simplicity.

Fixing $x \in \mathcal{X}$, we model the density of the data within the $i$th neighborhood using 
\begin{align}
f_i(x) = \mathcal{K}(x ; \theta_i),\quad \theta_i \sim P_0, \label{eq:NN-DP1}
\end{align}
where $\theta_i$ are parameters specific to neighborhood $i$ that are given a global prior distribution $P_0$. To combine the $f_i(x)$s into a single global $f(x)$, similarly to equations 
(\ref{eq:SB})-(\ref{eq:Dirichlet}), we let
\begin{eqnarray}
f(x) = \sum_{i=1}^n \pi_i f_i(x),\quad 
\pi=( \pi_i )_{i=1}^{n} \sim \mbox{Dirichlet}(\alpha,\ldots,\alpha),\quad 
\theta_i \sim P_0.
\label{eq:NN-DP}
\end{eqnarray}
The key difference relative to standard Bayesian mixture model \eqref{eq:SB}  is that in \eqref{eq:NN-DP} we include one component for each data sample and assume that only the data in the $k$-nearest neighborhood of sample $i$ will inform about $\theta_i$.  In contrast, \eqref{eq:SB} lacks any sample dependence, and we infer allocation of samples to mixture components in a posterior inference phase.

Given the restriction that only data in the $i$th neighborhood $\mathcal{S}_i$ inform about $\theta_i$, the
 pseudo-posterior density $\widetilde{\Pi}_{1}(\theta_i ; \mathcal{S}_i, P_0)$ of $\theta_i$ with data $\mathcal{S}_i$ and prior $P_0$ is
\begin{align}
 \widetilde{\Pi}_{1}(\theta_i ; \mathcal{S}_i, P_0) \propto P_0(\theta_i) \prod_{j \in \mathcal{N}_i} \mathcal{K}(X_j ; \theta_i), 
 \label{eq:postthti} 
\end{align}
where the right-hand side of \eqref{eq:postthti} is motivated from Bayes' theorem. This pseudo-posterior is in a simple analytic form if $P_0$ is conjugate to $\mathcal{K}(x;\theta)$. The prior $P_0$ can involve unknown parameters and borrows information across neighborhoods; this reduces the large variance problem common to nearest neighbor estimators. 
Since the neighborhoods are overlapping, proposing a pseudo-posterior update for $\pi$ under \eqref{eq:NN-DP} is not straightforward. However, one can define the number of effective members in the $i$th neighborhood $\mathcal{S}_i$ similar in spirit to the number of points in the $h$th cluster in mixture models of the form 
\eqref{eq:SB}. 
By convention, we define the point $X_i$ that generated its neighborhood $\mathcal{S}_i$ to be an effective member of that neighborhood. For any other data point $X_j$ to be a effective member of the neighborhood generated by $X_i$ for $j \neq i$, we require $X_j \in \mathcal{S}_i$ but $X_j \notin \mathcal{S}_u$ for all $u=1,\ldots,n$ such that $u \notin \{i, j\}$. That is, $X_j$ lies in the neighborhood generated by $X_i$ but does not lie in the neighborhood of any other $X_u$ for $u \notin \{i,j\}$. In Section \ref{sec:alpha-plus-one}, we show that the number of effective member points defined as above approaches 1 as $n\to \infty$. This motivates the following Dirichlet pseudo-posterior density $\widetilde{\Pi}_{2}(\pi ; \mathcal{X}^{(n)})$ for the neighborhood weights $\pi$:
\begin{equation}
 \widetilde{\Pi}_{2}(\pi ; \mathcal{X}^{(n)}) = 
 \mbox{Dirichlet}(\pi \mid \alpha + 1,\ldots, \alpha + 1), \label{eq:NN-DP2}
\end{equation}
where $\mbox{Dirichlet}(p \mid q_1, \ldots, q_d)$ denotes the density of the Dirichlet distribution evaluated at $p$ with parameters $(q_1, \ldots, q_d)$. We provide a justification for the pseudo-posterior update \eqref{eq:NN-DP2} in Section \ref{sec:alpha-plus-one}. This distribution is inspired from the  conditional posterior on the kernel weights in the Dirichlet mixture of equations (\ref{eq:SB})-(\ref{eq:Dirichlet}), but we use $n$ components and fix the effective number of samples allocated to each component at one.  

Based on equations (\ref{eq:NN-DP1})-(\ref{eq:NN-DP2}), our nearest neighbor-Dirichlet mixture produces a pseudo-posterior distribution for the unknown density $f(x)$ through simple  distributions for the parameters characterizing the density within each neighborhood and for the weights. To generate independent Monte Carlo samples from the pseudo-posterior for $f$, one can simply draw independent samples of $(\theta_i)_{i=1}^{n}$ and $\pi$ from (\ref{eq:postthti}) and (\ref{eq:NN-DP2}) respectively, and plug these samples into the expression for $f(x)$ in \eqref{eq:NN-DP}. The resulting mechanism can be described as
\begin{align}
\label{eq:NNDM-PP}
\begin{split}
    \theta_i & \overset{ind}{\sim} \widetilde{\Pi}_{1}(\cdot \, ; \mathcal{S}_i, P_0) \quad \mbox{for $i=1,\ldots,n$,} \\
    \pi & \sim \widetilde{\Pi}_{2}(\cdot \,; \mathcal{X}^{(n)}) \\
    f(x) & = \sum_{i=1}^{n} \pi_i \mathcal{K}(x; \theta_i).
\end{split}
\end{align} 
In \eqref{eq:NNDM-PP}, we denote the induced pseudo-posterior distribution on $f$ by $f \sim \widetilde{\Pi}$. Although this is not exactly a coherent fully Bayesian posterior distribution, we claim that it can be used as a practical alternative to such a posterior in practice.  This claim is backed up by theoretical arguments, simulation studies, and a real data application in the sequel.

\subsection{Illustration with Gaussian Kernels}\label{sec:multivariate_gaussian}
Suppose we have independent and identically distributed (iid) observations $\mathcal{X}^{(n)}$ from the density $f$, where $X_i \in \mathbb{R}^p$ for $i = 1, \ldots, n$ and $f$ is an unknown density function with respect to the Lebesgue measure on $\mathbb{R}^p$ for $p \geq 1$. Let $\mathbb{R}^{p \times p}_{+}$ denote the set of all real-valued $p \times p$ positive definite matrices. Fix $x \in \mathbb{R}^p$. 
We proceed by setting $\mathcal{K}(x; \theta)$ to be the multivariate Gaussian density $\phi_p(x; \eta, \Sigma)$, given by $$\phi_p(x; \eta, \Sigma) = (2\pi)^{-p/2} |\Sigma| ^{-1/2} \exp{\{-(x - \eta)^\T \Sigma^{-1}(x - \eta)/2 \}},$$ where $\theta = (\eta, \Sigma)$, $\eta \in \mathbb{R}^p$ and $\Sigma \in \mathbb{R}^{p \times p}_{+}$. We first compute the neighborhoods $\mathcal{N}_i$ corresponding to $X_i$ as in Section \ref{subsection:NNDPframework} and place a normal-inverse Wishart (NIW) prior on $\theta_i = (\eta_i, \Sigma_i)$, given by $(\eta_i, \Sigma_i) \sim \mbox{NIW}_p(\mu_0, \nu_0, \gamma_0, \Psi_0)$ independently for $i = 1, \ldots,n$. That is, we let $$\eta_i \mid \Sigma_i \sim \Gauss\left(\mu_0, \dfrac{\Sigma_i}{ \nu_0} \right), \quad \Sigma_i \sim \mbox{IW}_p(\gamma_0, \Psi_0),$$ with $\mu_0 \in \mathbb{R}^p$, $\nu_0>0,\,\gamma_0> p-1$ and  $\Psi_0 \in \mathbb{R}_{+}^{p \times p}$; for details about parametrization see Section \ref{app:invwishart} of the Appendix. 


Monte Carlo samples from the pseudo-posterior of $f(x)$ can be obtained using Algorithm \ref{algo:NNDP_multivariate}. The corresponding steps for the univariate case are provided in Section \ref{app:univariate_gaussian} of the Appendix.
\begin{algorithm}
\begin{itemize}
\item \textbf{Step 1: }For $i=1,\ldots,n$, compute the neighborhood $\mathcal{N}_i$ for data point $X_i \in \mathbb{R}^{p}$ according to distance $d(\cdot, \cdot)$ with $(k-1)$ nearest neighbors in $\mathcal{X}^{-i} = \mathcal{X}^{(n)} \setminus \{X_i\}$.
\item \textbf{Step 2: }Update the parameters for neighborhood $\mathcal{N}_i$  to $(\mu_i, \nu_n, \gamma_n, \Psi_i)$, where $\nu_n = \nu_0 + k$, $\gamma_n  = \gamma_0 + k$, $$\mu_i  = \dfrac{1}{\nu_n} \left(\nu_0\mu_0 + k \bar{X}_{i}\right), \quad \bar{X}_{i} = \dfrac{1}{k} \sum_{j \in \mathcal{N}_i} X_j, \mbox{ and }$$ $$\Psi_i  =  \Psi_0 + \sum_{j \in \mathcal{N}_i} (X_j - \bar{X}_{i})(X_j - \bar{X}_{i})^\T  + \dfrac{k \nu_0} {\nu_n} (\bar{X}_{i} - \mu_0)(\bar{X}_{i} - \mu_0)^\T .$$
\item \textbf{Step 3: }To compute the $t$-th Monte Carlo sample $f^{(t)}(x)$ of $f(x)$, sample Dirichlet weights $\pi^{(t)} \sim \mbox{Dirichlet}(\alpha+1, \ldots, \alpha+1)$ and neighborhood specific parameters $(\eta_{i}^{(t)}, \Sigma_{i}^{(t)}) \sim  \mbox{NIW}_p(\mu_i,\,\nu_n, \,\gamma_n, \Psi_i)$ independently for $i = 1, \ldots, n$, and set 
\begin{equation*}
         f^{(t)}(x) = \sum_{i=1}^n \pi_{i}^{(t)} \phi_p \left(x ; \eta_{i}^{(t)}, \Sigma_{i}^{(t)} \right).
\end{equation*}
\end{itemize}
\caption{Nearest neighbor-Dirichlet mixture algorithm to obtain Monte Carlo samples from the pseudo-posterior of $f(x)$ with Gaussian kernel and normal-inverse Wishart prior.}
\label{algo:NNDP_multivariate}
\end{algorithm}
Although the pseudo-posterior distribution of $f(x)$ lacks an analytic form, we can obtain a simple form for its pseudo-posterior mean by integrating over the pseudo-posterior distribution of $(\theta_i)_{i=1}^{n}$ and $\pi$. Recall the definitions of $\mu_i$ and $\Psi_i$ from Step 2 of Algorithm \ref{algo:NNDP_multivariate} and define $\Lambda_i = \{\nu_{n}(\gamma_{n}-p+1)\}^{-1} (\nu_{n}+1) \, \Psi_i$. Then the pseudo-posterior mean of $f(x)$ is given by
 \begin{equation}\label{eq:multivariate_posterior_mean}
   \hat{f}_{n}(x) = \dfrac{1}{n} \sum_{i=1}^{n} t_{\gamma_{n}-p+1}(x ; \mu_i, \Lambda_i),
\end{equation}
where $t_{\gamma}(x ; \mu, \Lambda)$ for $x \in \mathbb{R}^p$ is the $p$-dimensional Student's t-density with degrees of freedom $\gamma > 0$, location $\mu \in \mathbb{R}^{p}$ and scale matrix $\Lambda \in \mathbb{R}_{+}^{p \times p}$. 
We proceed with using Gaussian kernels and NIW conjugate priors when implementing the NN-DM for the remainder of the paper.

\subsection{Hyperparameter Choice}\label{sec:CV}
The hyperparameters in the prior for the neighborhood-specific parameters need to be chosen carefully -- we found results to be sensitive to $\gamma_0$ and $\Psi_0$. If non-informative values are chosen for these key hyperparameters, we tend to inherit typical problems of nearest neighbor estimators including lack of smoothness and high variance. Suppose $\Sigma \sim \mbox{IW}_{p}(\gamma_{0}, \Psi_{0})$ and for $i, j = 1,\ldots,p$, let $\Sigma_{ij} \text{ and } \Psi_{0,\, ij}$ denote the $i, j$th entry of $\Sigma$ and $\Psi_0$, respectively. Then $\Sigma_{jj} \sim \mbox{IG}(\gamma_{*}/2, \Psi_{0,\, jj}/2)$ where $\gamma_{*} = \gamma_{0} - p + 1$. For $p=1$, the $\mbox{IW}_{p}(\gamma_0, \Psi_0)$ density simplifies to an $\mbox{IG}(\gamma_0/2, \gamma_0 \delta_0^2/2)$ density with $\delta_0^2 = \Psi_0/\gamma_0$. Thus borrowing from the univariate case, we set $\Psi_{0,\, jj} = \gamma_{*}\delta_{0}^2$ and $\Psi_{0,\, ij} = 0$ for all $i \neq j$, which implies that $\Psi_{0} = (\gamma_* \delta_{0}^2) \, \mathbbm{I}_p$ and we use leave-one-out cross-validation to select the optimum $\delta_{0}^2$. With $p$ dimensional data, we recommend fixing $\gamma_{0} = p$ which implies a multivariate Cauchy prior predictive density. We choose the leave-one-out log-likelihood as the criterion function for cross-validation, which is closely related to minimizing the Kullback-Leibler divergence between the true and estimated density \citep{hall1987kullback,bowman1984alternative}. The explicit expression for the pseudo-posterior mean in \eqref{eq:multivariate_posterior_mean} makes cross-validation computationally efficient. The description of a fast implementation is provided in Section \ref{app:CVdetails} of the Appendix.

The proposed method has substantially faster runtime if one uses a default choice of hyperparameters. 
In particular, we found the default values $\mu_0 = 0_p, \nu_0 = 0.001, \gamma_0 = p,$ and $\Psi_0 = \mathbbm{I}_p$ to work well across a number of simulation cases, especially when the true density is smooth.
Although using cross-validation to estimate $\Psi_0$ can lead to improved performance when the underlying density is spiky, cross-validation provides little to no gains for smooth true densities.  Furthermore, with low sample size and increasing number of dimensions, we found this improvement to diminish rapidly. In order to obtain desirable uncertainty quantification in simulations and applications, we found small values of $\alpha$ to work well. As a default value, we recommend using $\alpha = 0.001$ for small samples and moderate dimensions.

The other key tuning parameter for NN-DM is the number of nearest neighbors $k = k_n$. The pseudo-posterior mean in \eqref{eq:multivariate_posterior_mean} reduces to a single $t_{\gamma_n - p +1}$ kernel if $k_n = n$. In contrast, $k_n = 1$ provides a sample smoothing kernel density estimate with a specific bandwidth function \citep{terrell1992variable}. Therefore, the choice of $k$ can impact the smoothness of the density estimate.
To assess the sensitivity of the NN-DM estimate to the choice of $k$, we investigate how the out-of-sample log-likelihood of a test set changes with respect to $k$ in Section \ref{sec:k-cv-results}. These simulations suggest that the proposed method is quite robust to the exact choice of $k$. In practice with finite samples and small dimensions, we recommend a default choice of 
 $k_n = \lfloor n^{1/3} \rfloor + 1$ and $k_n = 10$ for univariate and multivariate cases, respectively.  These values led to good performance across a wide variety of simulation cases as described in  
 Section \ref{sec:simulation}.

\section{Theory}\label{sec:theory}

\subsection{Asymptotic Properties}\label{sec:mean_and_variance}

There is a rich literature on asymptotic properties of the posterior measure for an unknown density under Bayesian models, providing a frequentist justification for Bayesian density estimation; refer, for example to 
\cite{ghosal1999posterior}, \cite{ghosal2007posterior}. Unfortunately, the tools developed in this literature rely critically on the mathematical properties of fully Bayes posteriors, providing theoretical guarantees for a computationally intractable exact posterior distribution under a Bayesian model.  Our focus is instead on providing frequentist asymptotic guarantees for our computationally efficient NN-DM approach, with this task made much more complex by the dependence across neighborhoods induced by the use of a nearest neighbor procedure.

We first focus on proving pointwise consistency of the pseudo-posterior of $f(x)$ induced by \eqref{eq:NNDM-PP} for each $x \in [0,1]^p$, using Gaussian kernels as in Section \ref{sec:multivariate_gaussian}. We separately study the mean and variance of the NN-DM pseudo-posterior distribution, first showing that the pseudo-posterior mean in \eqref{eq:multivariate_posterior_mean} is pointwise consistent and then that the pseudo-posterior variance vanishes asymptotically.
 The key idea behind our proof is to show that the pseudo-posterior mean is asymptotically close to a kernel density estimator with suitably chosen bandwidth for fixed $p$ and $k_n \to \infty$ at a desired rate. The proof then follows from standard arguments leading to consistency of kernel density estimators. 
The NN-DM pseudo-posterior mean mimics a kernel density estimator only in the asymptotic regime; in finite sample simulation studies (refer to Section \ref{sec:simulation}), NN-DM has much better performance. The detailed proofs of all results in this section are in the Appendix.

Consider independent and identically distributed data $\mathcal{X}^{(n)}$ from a fixed unknown density $f_{0}$ with respect to the Lebesgue measure on $\mathbb{R}^p$ equipped with the Euclidean metric, inducing the measure $P_{f_{0}}$ on $\mathcal{B}(\mathbb{R}^p)$. We use $\widetilde{E}\{f(x)\}$, $\widetilde{\mbox{var}}\{f(x)\},$ and $\widetilde{\mbox{pr}}\{f(x) \in B \}$ to denote the mean of $f(x)$, variance of $f(x)$, and probability of the event $\{f(x) \in B \}$ for $B \in \mathcal{B}(\mathbb{R}^p)$, respectively, under the pseudo-posterior distribution of $f(x)$ implied by \eqref{eq:NNDM-PP}. We make the following regularity assumptions on $f_0$: 

\begin{assumption}[Compact support] \label{assump:a1}
 $f_0$ is supported on $[0,1]^p$.
\end{assumption}
\begin{assumption}[Bounded gradient]\label{assump:a2}
$f_{0}$ is continuous on $[0,1]^p$ with $||\nabla f_{0}(x)||_2 \, \leq L$ for all $x \in [0,1]^p$ and some finite $L > 0$.
\end{assumption}
\begin{assumption}[Bounded sup-norm]\label{assump:a3}
$||\log(f_0)||_{\infty} < \infty.$
\end{assumption}
Our asymptotic analysis relies on analyzing the behaviour of the pseudo-posterior updates within each nearest neighborhood. We leverage on key results from \cite{biau2015lectures, evans2002asymptotic} which are based on the assumption that the true density has compact support as in Assumption \ref{assump:a1}. Assumption \ref{assump:a2} ensures that the kernel density estimator has finite expectation. Versions of this assumption are common in the kernel density literature; for example, refer to \cite{tsybakov2004introduction}. Assumptions \ref{assump:a1} and \ref{assump:a3} imply the existence of $0 < a_1, a_2 < \infty$ such that $0 < a_1 < f_0(x) < a_2 < \infty$ for all $x \in [0,1]^p$, which is referred to as a positive density condition by \cite{evans2008law, evans2002asymptotic}. This is used to establish consistency of the proposed method and justify the choice of the pseudo-posterior distribution of the weights. 
 These assumptions are standard in the literature studying frequentist asymptotic properties of nearest neighbor and Bayesian density estimators.
 

For $i=1,\ldots,n$, recall the definitions of $\mu_i$ and $\Lambda_i$ from  \eqref{eq:multivariate_posterior_mean}: 
$$\mu_{i} = \frac{\nu_0}{\nu_n}\mu_0 + \frac{k_n}{\nu_n}\bar{X}_i, \quad \Lambda_{i} = \frac{\nu_{n}+1}{\nu_{n}(\gamma_n - p + 1)}\Psi_{i},$$
where $\nu_n = \nu_0 + k_n, \gamma_n = \gamma_0 + k_n$, and $\bar{X}_i, \Psi_i$ are as in Algorithm \ref{algo:NNDP_multivariate}. Define the bandwidth matrix 
\begin{equation}
    \label{eq:h-n-sq}
    H_n = h_n^2 \mathbbm{I}_p, \quad \text{where \,} h_n^2 = \frac{(\nu_n+1)(\gamma_0-p+1)}{\nu_n (\gamma_n - p + 1)} \delta_0^2.
\end{equation}
We have suppressed the dependence of $\mu_i$ and $\Lambda_i$ on $n$ for notational convenience. It is immediate that $h_n^2 \to 0$ if $k_n \to \infty$ as $n \to \infty$. Fix $x \in [0,1]^p$. To prove consistency of the pseudo-posterior mean, we first show that $\hat{f}_n(x)$ and $f_K(x) = (1/n) \sum_{i=1}^{n} t_{\gamma_n - p + 1}(x; X_i, H_n)$ are asymptotically close, that is we show that $E_{P_{f_0}}(\, |\hat{f}_n(x) - f_K(x)| \,) \to 0$ as $n \to \infty$. To obtain this result, we approximate $\mu_i$ by $X_i$ and $\Lambda_i$ by $H_n$ using successive applications of the mean value theorem. Finally, we exploit the convergence of $f_K(x)$ to the true value $f_0(x)$ to obtain the consistency of $\hat{f}_n(x)$. The proof of convergence of $f_K(x)$ to $f_0(x)$ is provided in Section \ref{app:kdeconsistency} of the Appendix. The precise statement regarding the consistency of the pseudo-posterior mean is given in the following theorem. Let $a\wedge b$ denote the minimum of $a$ and $b$.



\begin{theorem}
Fix $x \in [0,1]^p$. Let $k_n = o(n^{i_0})$ with $i_0 = \{2/(p^2+p+2)\} \wedge \{4/(p+2)^2\}$ such that $k_n \to \infty$ as $n \to \infty$, and $\nu_0 = o\{n^{-2/p} k_n^{(2/p)+1}\}$. Then, $\hat{f}_{n}(x) \to f_{0}(x)$ in $P_{f_0}$-probability as $n \to \infty$.
\label{theorem:post_mean}
\end{theorem}

We now look at the pseudo-posterior variance of $f(x)$. We let 
\begin{equation}
    \label{eq:post-var-RnDn}
    R_{n} = \dfrac{\Gamma\{(\gamma_n - p +2)/2 \}}{\Gamma\{(\gamma_n - p +1)/2 \}} \left[\dfrac{\nu_n+2}{4 \pi \nu_n (\gamma_n - p + 2)}\right]^{p/2} \text{ and }  D_{n} = \frac{(\gamma_n - p + 1) (\nu_n + 2)}{2(\gamma_n - p + 2) (\nu_n + 1)}.
\end{equation}
For $i=1,\ldots,n$, let $B_{i} = D_n \Lambda_i$ and define
\begin{equation}
    \label{eq:post-var-fhatvar}
    \widehat{f}_{var}(x) = \frac{1}{n} \sum_{i=1}^{n} t_{\gamma_n - p + 2}(x; \mu_i, B_i).
\end{equation} 
As $n \to \infty$, we have $D_n \to 1/2$. Analogous steps to the ones used in the proof of Theorem \ref{theorem:post_mean} can be used to imply that $\widehat{f}_{var}(x) \to f_0(x)$ in $P_{f_0}$-probability. Also, as $n \to \infty$, $k_n^{(p-1)/2} R_{n} = \mathcal{O}(1)$ using Stirling's approximation. We now provide an upper bound on the pseudo-posterior variance of $f(x)$ which shows convergence of the pseudo-posterior variance to $0$.
\begin{theorem}
Let $H_n$ be the bandwidth matrix defined in \eqref{eq:h-n-sq}. Let $R_n, D_n$ be as in \eqref{eq:post-var-RnDn} and $\widehat{f}_{var}$ be as in \eqref{eq:post-var-fhatvar}. Under Assumptions \ref{assump:a1}-\ref{assump:a3} with $x, k_n,$ and $\nu_0$  as in Theorem \ref{theorem:post_mean}, we have
\begin{equation}
\widetilde{\mbox{{\em var}}} \{f(x) \} \leq \frac{R_n D_n^{-p/2} \widehat{f}_{var}(x)}{|H_n|^{1/2}} \left\{\frac{1}{n(\alpha+1) + 1} + \frac{1}{n} \right\}.  
\end{equation}
This implies $\widetilde{\mbox{{\em var}}} \{f(x) \} \to 0$ in $P_{f_0}$-probability as $n \to \infty$.
\label{theorem:post_var}
\end{theorem}
Refer to Sections \ref{app:multproof} and \ref{app:proof_post_var} in the Appendix for proofs of Theorems 4 and 5, respectively. Pointwise pseudo-posterior consistency follows from Theorems \ref{theorem:post_mean} and \ref{theorem:post_var}, as shown below. 
\begin{theorem}\label{theorem:post_conc}
Let $f_0$ satisfy Assumptions \ref{assump:a1}-\ref{assump:a3} with $x, k_n$ and $\nu_0$ as in Theorem \ref{theorem:post_mean}. Fix $\epsilon > 0$ and define the $\epsilon$-ball around $f_{0}(x)$ by $U_{\epsilon} = \{y_* :\, |y_* - f_{0}(x)| \, \leq \epsilon\}$. Let $\widetilde{\text{{\em pr}}}\{f(x) \in U_{\epsilon}^{c} \}$ denote the probability of the set $U_\epsilon^c$ under the pseudo-posterior distribution of $f(x)$ as induced by \eqref{eq:NNDM-PP}. Then 
$\widetilde{\text{{\em pr}}}\{f(x) \in U_{\epsilon}^{c} \} \to 0$ in $P_{f_0}$-probability as $n \to \infty$.
\end{theorem}
\begin{proof}
Fix $\epsilon > 0$ and consider the $\epsilon$-ball $U_\epsilon = \{y_* :\, |y_* - f_{0}(x)| \, \leq \epsilon\}$. 
Then by Chebychev's inequality, we have $\widetilde{\text{{\em pr}}}\{f(x) \in U_{\epsilon}^{c} \} \leq [\{\hat{f}_n(x) - f_{0}(x)\}^2 + \widetilde{\mbox{var}}\{f(x)\}]/\epsilon^2 \xrightarrow{} 0$ in $P_{f_0}$-probability as $n \to \infty$, using Theorems \ref{theorem:post_mean} and \ref{theorem:post_var}.
\end{proof}

We next focus on the limiting distribution of $f(x)$ for the univariate case. From Section \ref{app:univariate_gaussian} of the Appendix, the pseudo-posterior distribution of $(\eta_i, \sigma_i^2)$ for $i=1,\ldots,n$ is given by NIG$(\mu_i, \nu_n,$ $\gamma_n/2, \gamma_n \delta_i^2/2)$, where $\mu_i, \nu_n, \gamma_n$ are as before and $$\gamma_n \delta_i^2 = \gamma_0 \delta_0^2 + \sum_{j \in \mathcal{N}_i} (X_j - \bar{X}_i)^2 + \frac{k_n \nu_0}{\nu_n} (\bar{X}_i - \mu_0)^2.$$  We establish in Theorem \ref{theorem:asymptotic-normality} that the limiting distribution of $f(x)$ is a Gaussian distribution with appropriate centering and scaling. This allows interpretation of $100(1-\beta)\%$ pseudo-credible intervals as $100(1-\beta)\%$ frequentist confidence intervals on average for large $n$. 
\begin{theorem}
\label{theorem:asymptotic-normality}
Fix $x \in [0,1]$. Suppose $f_0$ satisfies Assumptions \ref{assump:a1}-\ref{assump:a3} and also satisfies $|f_0^{(4)}(x)| \leq C_0$ for all $x \in [0,1]$ for some finite $C_0 > 0$. Let $k_n$ satisfy $k_n = o(n^{2/7})$ such that $ n^{-2/9} k_n \to \infty$, $h_n$ be as in \eqref{eq:h-n-sq} satisfying $h_n \to 0$, and $\alpha = \alpha_n \to \infty$, as $n \to \infty$. For $t \in \mathbb{R}$, define $$G_n(t) = \widetilde{\text{{\em pr}}}\left[(nh_n)^{1/2}\left\{f(x) - \left(f_0(x) + \frac{h_n^2 f_0^{(2)}(x)}{2} \right)\right\} \leq t \right].$$ Then, we have
$$\lim_{n \to \infty} E_{P_{f_0}}\{G_n(t)\} = \Phi\left(t \, ; 0, \frac{f_0(x)}{2\pi^{1/2}}\right),$$
where $\Phi(t \,; 0, \sigma^2)$ denotes the cumulative distribution function of the $N(0, \sigma^2)$ density.
\end{theorem}
For a proof of Theorem \ref{theorem:asymptotic-normality}, we refer the reader to Section \ref{app:surrogate-AN} of the Appendix. 

\subsection{Pseudo-Posterior Distribution of Weights}
\label{sec:alpha-plus-one}

We investigate the rationale behind the pseudo-posterior update \eqref{eq:NN-DP2}  of the weight $\pi$, which has a symmetric prior distribution $\pi \sim \mbox{Dirichlet}(\alpha, \ldots, \alpha)$ as motivated in Section \ref{sec:intro}. As discussed in Section \ref{sec:intro}, the conditional update for the weights $\pi$ in a finite Bayesian mixture model with $m$ components given the cluster allocation indices $\{c_1, \ldots, c_n\}$ is obtained by $\text{Dirichlet}(\alpha + n_1, \ldots, \alpha + n_m)$, where $\alpha$ is the prior concentration parameter and $n_h = \sum_{i=1}^{n} \mathbb{I}(c_i = h)$ is the number of data points allocated to the $h$th cluster. 
This is not true in our case as the $k_n$-nearest neighborhoods have considerable overlap between them. 
Instead, we consider the number of effective member data points in each of these neighborhoods.

Define the $k_n$-nearest neighborhood of $X_i$ to be the set $\mathcal{S}_i = \{X_j : d(X_i, X_j) \leq d(X_i, X_{i[k_n]})\}$ where $X_{i[k_n]}$ is the $k_n$-th nearest neighbor of $X_i$ in the data $\mathcal{X}^{(n)}$, following the notation in Section \ref{subsection:NNDPframework}. We assume $d(\cdot, \cdot)$ is the Euclidean metric from here on, and let $R_i = d(X_i, X_{i[k_n]}) = ||X_i - X_{i[k_n]}||_2$ denote the distance of $X_i$ from its $k_n$-th nearest neighbor in $\mathcal{X}^{(n)}$.  

Let $N_i$ denote the number of effective members in $\mathcal{S}_i$ as defined in Section \ref{subsection:NNDPframework}. Then, we can express $N_i$ as
\begin{equation}
    \label{eq:member-number}
    N_i = 1+ \sum_{j \neq i}\mathbb{I} \left[X_j \in \mathcal{S}_i \,, \underset{u \notin \{i, j\}}{\bigcap} \{X_j \notin \mathcal{S}_u\} \right],
\end{equation}
where $\mathbb{I}(A)$ is the indicator function of the set $A$. Under  $X_1, \ldots, X_n \overset{iid}{\sim} f_0$, we have
\begin{equation}
   E_{P_{f_0}}(N_1) = 1 + (n-1) P_{f_0}\left[X_2 \in \mathcal{S}_1 \, , \bigcap_{u = 3}^{n} \{X_2 \notin \mathcal{S}_u\}\right],
\end{equation}
by symmetry. Furthermore, $N_i$ are identically distributed for $i=1,\ldots,n$. We now state a result which provides a motivation for our choice of the pseudo-posterior update of $\pi$. For two sequences of real numbers $(a_n)$ and $(b_n)$, we write $a_n \sim b_n$ if $|a_n / b_n| \to c_0$ as $n \to \infty$ for some constant $c_0 > 0.$

\begin{theorem}
\label{theorem:alpha-plus-one}
Suppose $X_1, \ldots, X_n \overset{iid}{\sim} f_0$ with $f_0$ satisfying Assumptions \ref{assump:a1}-\ref{assump:a3}. Furthermore, suppose that $k_n \sim n^{i_0 - \epsilon}$ for some $\epsilon \in (0,i_0)$, where $i_0$ is as defined in Theorem \ref{theorem:post_mean}. Then, 
\begin{equation}
    \underset{n \to \infty}{\lim} \, n \, P_{f_0}\left[X_2 \in \mathcal{S}_1, \bigcap_{u=3}^{n} \{X_2 \notin \mathcal{S}_u\}\right] = 0. 
\end{equation}
\end{theorem}
Proof of Theorem \ref{theorem:alpha-plus-one} is in Section \ref{app:alpha-plus-one-proof} of the Appendix. The above theorem suggests we asymptotically have only one effective member per neighborhood $\mathcal{S}_i$, namely the point $X_i$ that itself generated this neighborhood. This result motivates our choice of the pseudo-posterior update of the weight vector $\pi$. 
We illustrate uncertainty quantification of the proposed method in finite samples in Section \ref{sec:coverage} with this choice of pseudo-posterior update of the weight vector $\pi$.

\section{Simulation Experiments} \label{sec:simulation}

\subsection{Preliminaries} \label{subsection:simulation_preliminaries}

In this section, we compare the performance of the proposed density estimator with several other standard density estimators through several numerical experiments. We evaluate estimation performance based on the expected $\mathcal{L}_1$ distance \citep{devroye1985nonparametric}. For the pair $(f_0, \hat{f})$, where $f_0$ is the true data generating density and $\hat{f}$ is an estimator, the expected $\mathcal{L}_1$ distance is defined as $\mathcal{L}_{1}(f_0, \hat{f}) = E_{P_{f_0}} \{\int  |f_0(x) - \hat{f}(x)| \, dx \}$. We compute an estimate $\hat{\mathcal{L}}_1(f_0, \hat{f})$ of $\mathcal{L}_1(f_0, \hat{f})$ in two steps. First, we sample $n$ training points $X_1, \ldots, X_n \sim f_0$ and obtain $\hat{f}$ based on this sample, and then further sample $n_t$ independent test points $X_{n+1}, \ldots, X_{n + n_t} \sim f_0$ and compute $$\hat{L} = \frac{1}{n_t} \sum_{i = 1}^{n_t} \left| \frac{\hat{f}(X_{n+i})}{f_0(X_{n+i})} - 1 \right|.$$ In the second step, to approximate the expectation with respect to $P_{f_0}$, the first step is repeated $R$ times. Letting $\hat{L}_r$ denote the estimate for the $r$th replicate, we compute the final estimate as $\hat{\mathcal{L}}_1(f_0, \hat{f}) = (1/R) \sum_{r=1}^R \hat{L}_r$. Then, it follows that $\hat{\mathcal{L}}_1(f_0, \hat{f}) \to \mathcal{L}_{1}(f_0, \hat{f})$ as $n_{t}, R \xrightarrow{} \infty$, by the law of large numbers. In our experiments, we set $n_t = 500$ and $R = 20$. We let $0_p$ and $\mathbbm{1}_p$ denote the vector with all entries equal to $0$ and the vector with all entries equal to $1$ in $\mathbb{R}^p$, respectively, for $p\geq 1$.

All simulations were carried out using the R programming language \citep{rsoftware}. For Dirichlet process mixture models, we collect $2,000$ Markov chain Monte Carlo (MCMC) samples after discarding a burn-in of $3,000$ samples using the \texttt{dirichletprocess} package \citep{dirichletprocesspackage}. The default implementation of the Dirichlet process mixture model in $p$ dimensions in the \texttt{dirichletprocess} package uses multivariate Gaussian kernels and has the base measure as $\text{NIW}_p(0_p, p, p, \mathbbm{I}_p)$ with the Dirichlet concentration parameter having the $\text{Gamma}(2,4)$ prior \citep{west1992hyperparameter}. For the nearest neighbor-Dirichlet mixture, $1,000$ Monte Carlo samples are taken. For the kernel density estimator, we select the bandwidth by the default plug-in method \texttt{hpi} for univariate cases and \texttt{Hpi} for multivariate cases \citep{sheather1991reliable, wand1994multivariate} using the package \texttt{ks} \citep{kspackage}. We additionally consider the k-nearest neighbor estimator studied in \cite{mack1979multivariate}, setting the number of neighbors $k = n^{1/2}$, and the variational Bayes (VB) approximation to Dirichlet process mixture models \citep{blei2006variational}. We also compare with the optional Polya tree (OPT) \citep{wong2010optional} using the package \texttt{PTT}. For univariate cases, we consider the recursive predictive density estimator (RD) from \cite{hahn2018recursive}, Polya tree mixtures (PTM) using the package \texttt{DPpackage} \citep{DPpackage}, and the sample smoothing kernel density estimator (A-KDE) using the package \texttt{quantreg}. Lastly, we also compare with the local likelihood density estimator (LLDE) using the package \texttt{locfit} for both univariate and multivariate cases.
Dirichlet process mixture model hyperparameter values are kept the same in both the MCMC and variational Bayes implementations, with the number of components of the variational family set to 10 for all cases.
We denote the nearest neighbor-Dirichlet mixture, Dirichlet process mixture (DPM) implemented with MCMC, kernel density estimator, variational Bayes approximation to the DPM, and k-nearest neighbor density estimator by NN-DM, DP-MC, KDE, DP-VB, and KNN, respectively, in tables and figures.


\subsection{Univariate Cases}\label{sec:univ_simulations}

We set $n = 200, 500$ with $k_{n} = \lfloor n^{1/3} \rfloor + 1$.
We consider 10 choices of $f_0$ from the \texttt{R} package \texttt{benchden} \citep{benchden}; the specific choices are Cauchy (CA), claw (CW), double exponential (DE), Gaussian (GS), inverse exponential (IE), lognormal (LN), logistic (LO), skewed bimodal (SB), symmetric Pareto (SP), and sawtooth (ST) with default choices of the corresponding parameters. The prior hyperparameter choices for the proposed method are $\mu_0 = 0, \nu_0 = 0.001, \gamma_0 = 1$; 
$\delta_0^2$ is chosen via the cross-validation method of Section \ref{sec:CV}. Detailed numerical results are deferred to Table \ref{tab:univariate_tab} in the Appendix. Instead, in Figure \ref{fig:univariate_results}, we provide a visual summary of the performance of each method under consideration by forming a box plot of the estimated $\mathcal{L}_1$ errors of the methods across all the data generating densities. Methods with lower median as indicated by the solid line of the box plot, and smaller overall spread are preferable as they provide higher accuracy and also maintain such accuracy across a collection of true density cases. Results of KNN are omitted in Figure \ref{fig:univariate_results} due to much higher values compared to other methods. For the KDE and RD estimator, the plot and the table exclude the results for the heavy-tailed densities CA, IE, and SP due to very high $\mathcal{L}_1$ errors. 
\begin{figure}
\begin{subfigure}{.55\textwidth} 
\includegraphics[height = 6cm, width = 7.3cm]{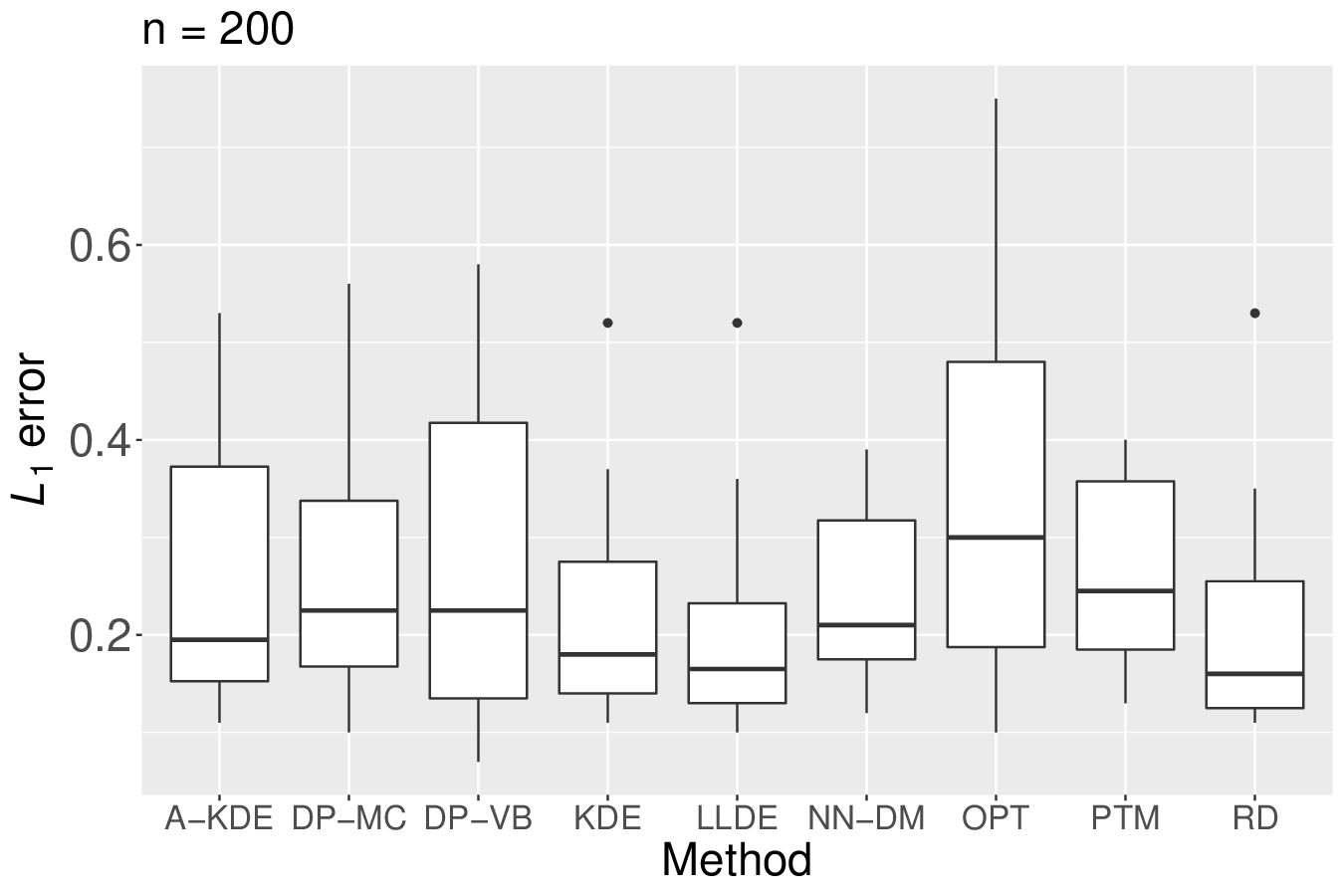}
\end{subfigure}
\hspace{-0.3in}
\begin{subfigure}{.35\textwidth}
\includegraphics[height = 6cm, width = 7.3cm]{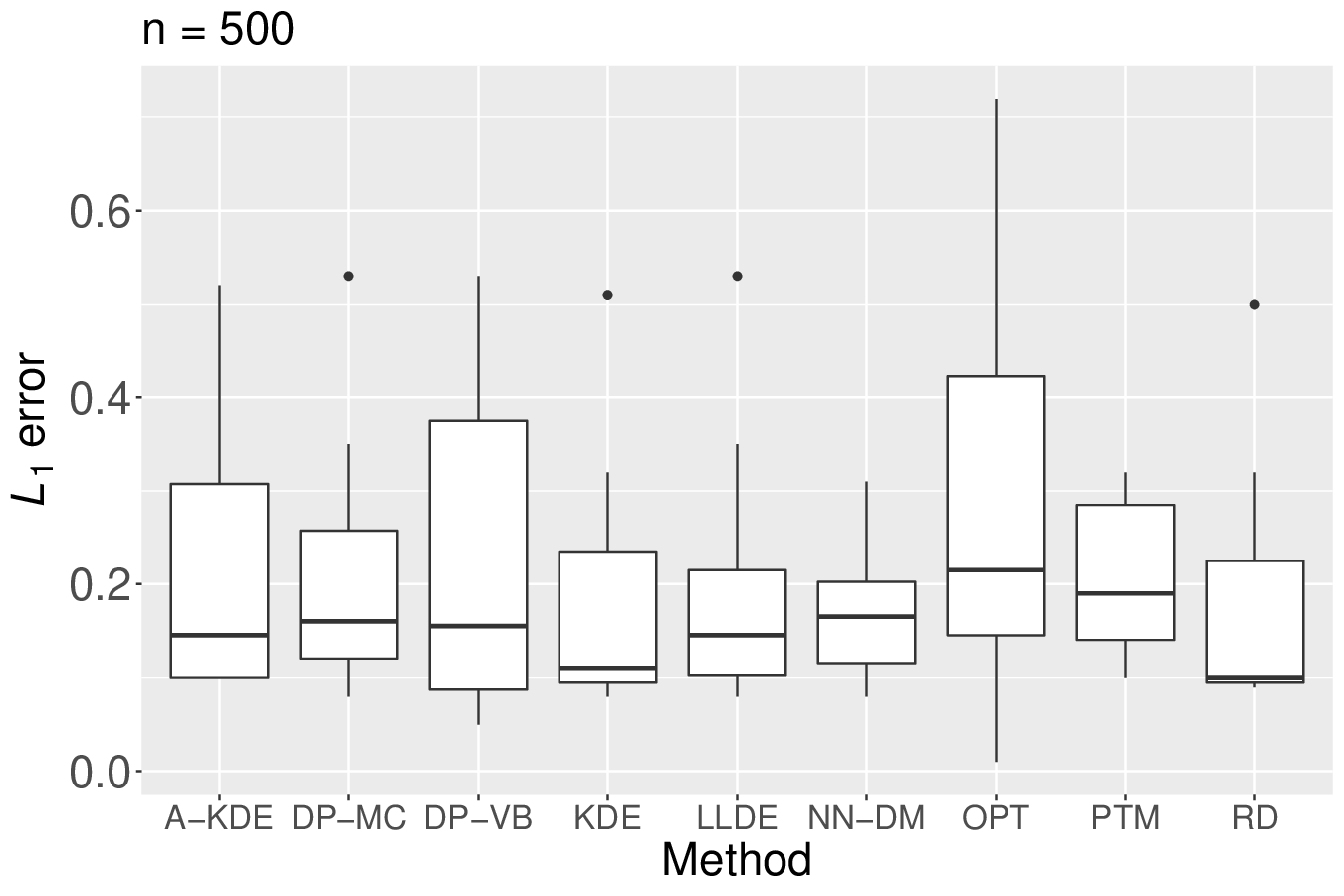}
\end{subfigure}
\caption{Box plots of $\hat{\mathcal{L}}_1(f_0, \hat{f})$ for the 10 different choices of the true density $f_0$ and different estimators $\hat{f}$ for univariate data. The box plots for KDE and RD exclude the heavy-tailed cases CA, IE, and SP.}
\label{fig:univariate_results}
\end{figure}

Overall, a major advantage of the proposed method is its versatility among the considered methods. The Bayesian nonparametric methods DP-MC, DP-VB, PTM, OPT, and RD are often close to NN-DM in terms of their performance when the true densities are smooth and do not display locally spiky behavior. However, the NN-DM performs better than other methods in densities where such local behavior is present and performs very close to the best estimator for either the smooth heavy-tailed or thin-tailed densities. 
The KDE and RD perform well when data are generated from a smooth underlying density. However, there are some cases where the error for KDE and RD is very high. For instance, when $n = 500$ and $f_0$ is the standard Cauchy (CA) density, the estimated $\mathcal{L}_1$ error for the KDE is 38501.85 and the algorithm for the RD estimate did not converge. Both the KDE and RD also perform poorly in very spiky multi-modal densities such as the ST. Compared to the LLDE and the A-KDE, the NN-DM displays similar performance in heavy-tailed and smooth densities when $n=200$, with the NN-DM performing better for the spiky densities. However, when $n=500$, the NN-DM shows significant improvements over the LLDE and the A-KDE for spiky densities such as the CW and the ST.  

In Figure \ref{fig:plotuniv}, we show the performance of the NN-DM estimator $\hat{f}_n$
(with hyperparameters chosen as described earlier) relative to the posterior mean under a DP-MC with default or hand-tuned hyperparameters, when $500$ data points are generated from the sawtooth (ST) density. The Dirichlet process mixture with default hyperparameters is unable to detect the multiple spikes, merging adjacent modes to form larger clusters, perhaps due to inadequate mixing of the Markov chain Monte Carlo sampler or to the Gaussian kernels used in the mixture. As a result, we had to hand-tune the hyperparameters for the Dirichlet process mixture to obtain comparable performance with the NN-DM (without hand-tuning). We obtained the best results when changing the hyperparameters of the base measure of the DP-MC to $\text{NIG}(0,0.01, 1, 1)$ while keeping the prior on $\alpha$ the same as before. This illustrates the deficiency of the DP-MC in estimating densities with spiky local behavior unless we hand-tune the hyperparameters, which requires knowledge of the true density.
We also compare the performance of the two methods with a smoother test density in Figure \ref{fig:plotsbimodal}, where the data are generated from a skewed bimodal (SB) distribution. Both the estimates are comparable, but the nearest neighbor-Dirichlet mixture provides better uncertainty quantification. Similar results are obtained for $n=1000$, and hence are omitted.
\begin{figure}
\begin{subfigure}{.35\textwidth} 
\includegraphics[height = 5cm, width = 7cm]{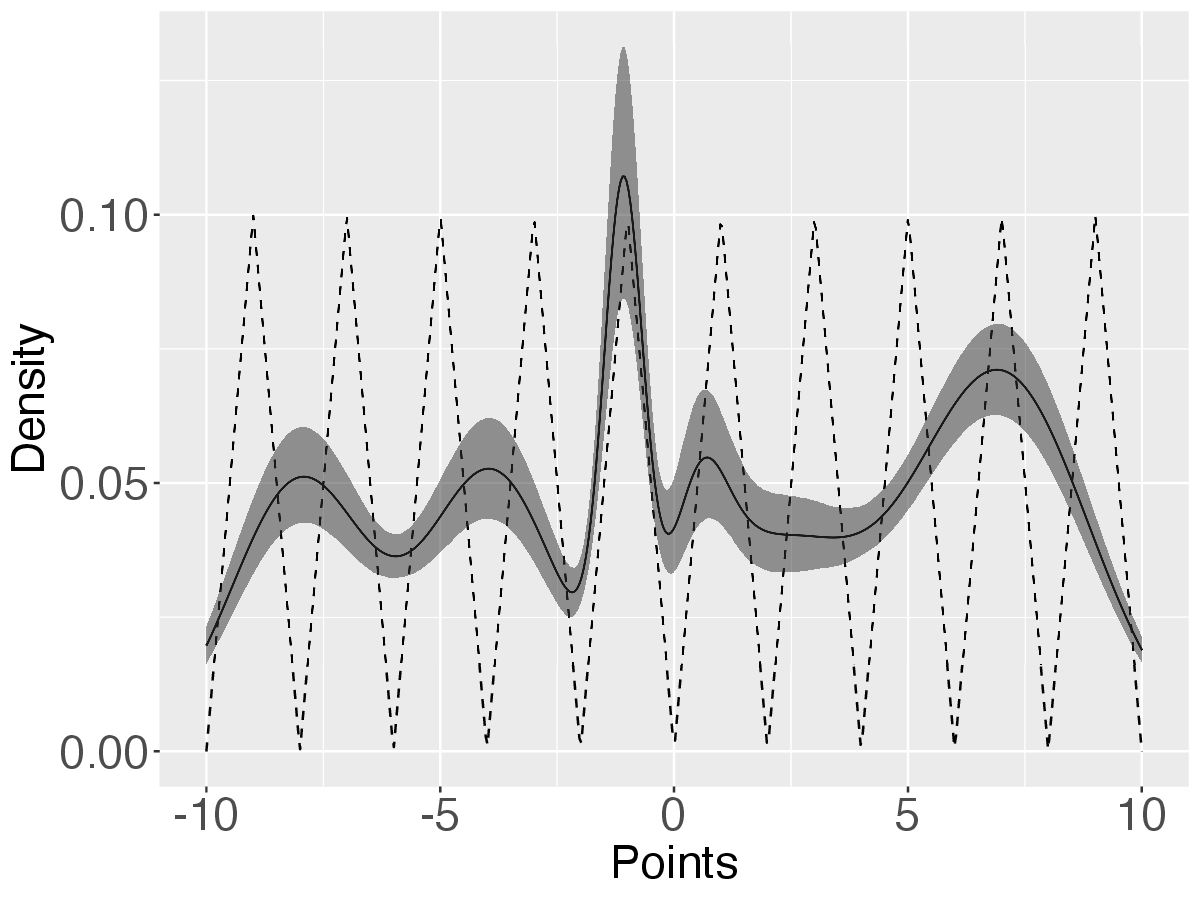}
\end{subfigure}
\hspace{0.85in}
\begin{subfigure}{.35\textwidth}
\includegraphics[height = 5cm, width = 7cm]{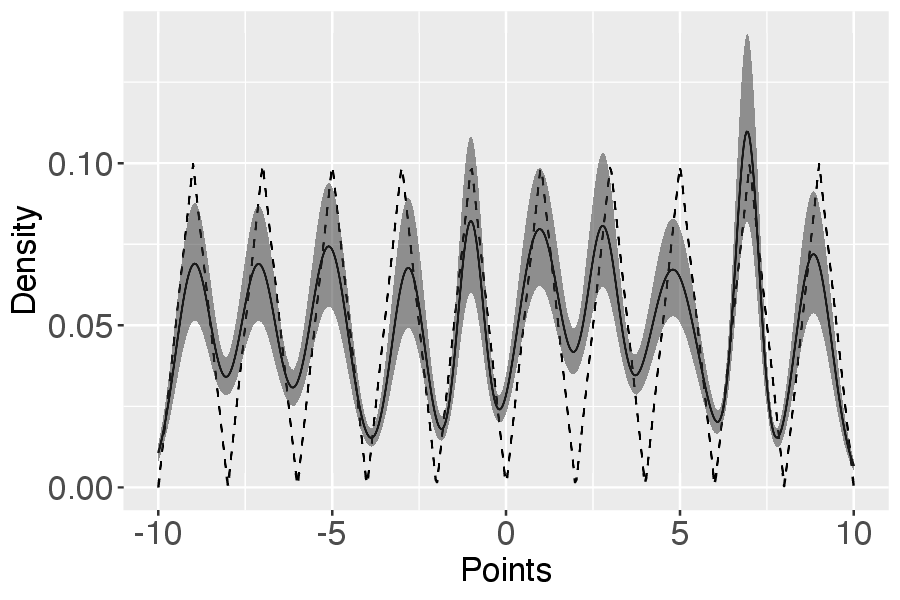}
\end{subfigure}
\begin{center}
    \begin{subfigure}{.35\textwidth}
\hspace{-1cm}\includegraphics[height = 5cm, width = 7cm]{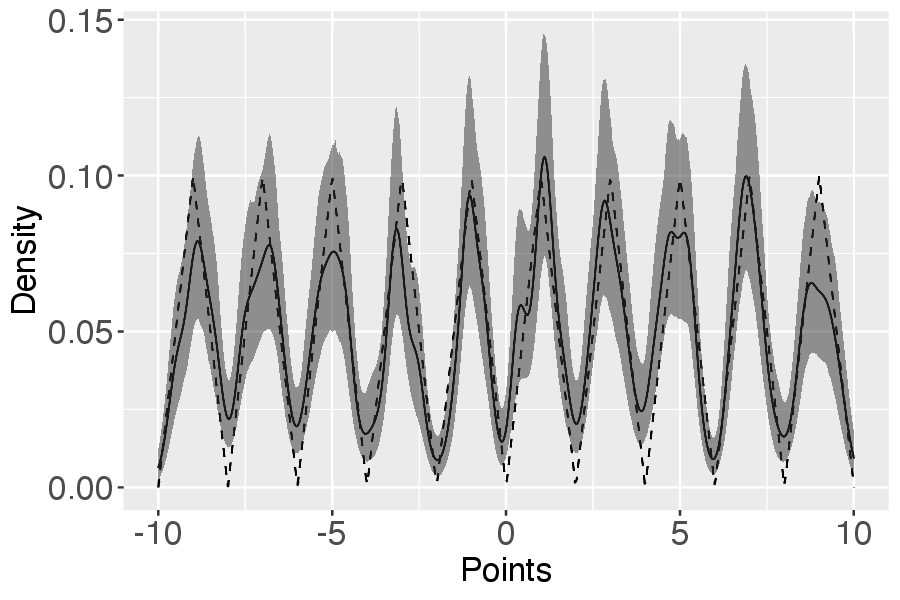}
\end{subfigure}
\end{center}
\caption{Plot comparing density estimates for the NN-DM and DP-MC for  $n=500$ samples generated from the 
	 sawtooth (ST) density. Shaded regions correspond to $95\%$ (pseudo) posterior credible intervals. The true density is displayed using dotted lines. The top panel shows the performance of DP-MC with default hyperparameters on the left and with hand-tuned hyperparameters on the right. The bottom panel shows the performance of the NN-DM.}
\label{fig:plotuniv}
\end{figure}

\begin{figure}
\begin{subfigure}{.35\textwidth}
\includegraphics[height = 5cm, width = 7cm]{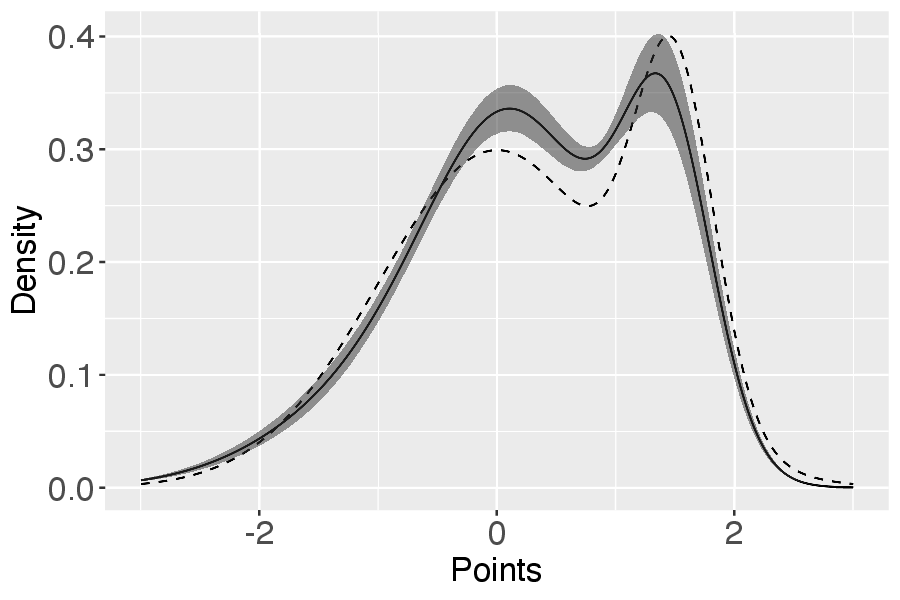}
\end{subfigure}
\hspace{0.85in}
\begin{subfigure}{.35\textwidth}
\includegraphics[height = 5cm, width = 7cm]{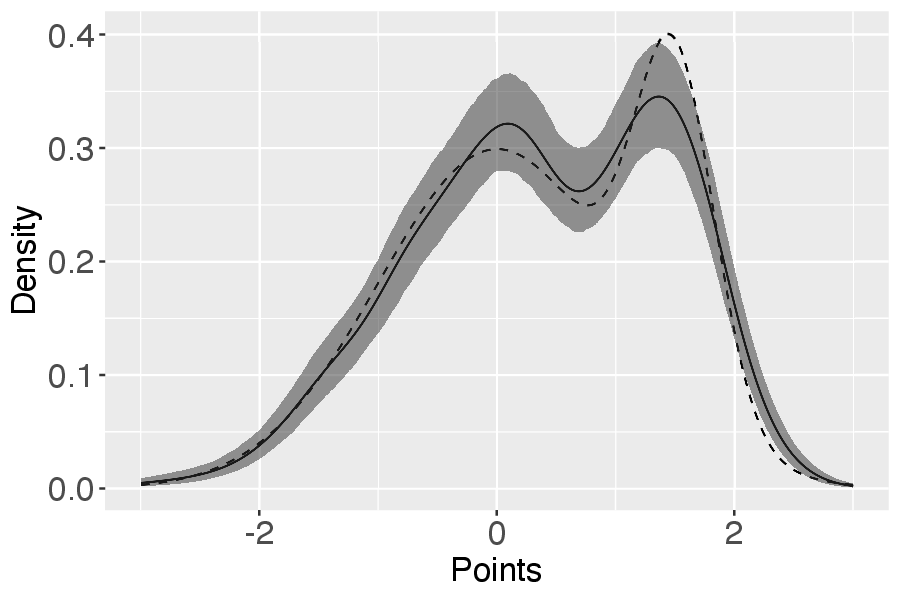}
\end{subfigure}
\caption{Similar to Figure \ref{fig:plotuniv}, with data of sample size $n=500$ generated from the skewed bimodal (SB) density. Left panel shows the DP-MC fit and the right panel shows the NN-DM fit. }
\label{fig:plotsbimodal}
\end{figure} 
\subsection{Multivariate Cases}\label{sec:mvt_simulations}

For the multivariate cases, we consider $n = 200$ and $1000$. The number of neighbors is set to $k = 10$ and the dimension $p$ is chosen from $\{2, 3, 4, 6\}$. Recall the definition of $\phi_p(x ; \mu, \Sigma)$ from Section \ref{sec:multivariate_gaussian} and let $\Phi(x)$ be the cumulative distribution function of the standard Gaussian density. Let $S_0 = \rho \, \mathbbm{1}_{p} \mathbbm{1}_{p}^{\T} + (1-\rho) \, \mathbbm{I}_p$ with $\rho = 0.8$. Let $x = (x_1, \ldots, x_p)^{\T}$.
We consider the following cases. \\ 
    (1) {\em Mixture of Gaussians (MG)}: $f_0(x) = 0.4 \, \phi_{p}(x ; m_1, S_0) + 0.6 \, \phi_{p}(x ; m_2, S_0)$, where $m_1 = -2\times \mathbbm{1}_{p}, m_2 = 2 \times \mathbbm{1}_p$.  \\
    (2) {\em Skew normal (SN)}: $f_0(x) = 2 \phi_{p}(x ; m_0, S_0) \Phi\{s_0^{\T }W^{-1}(x - m_0)\}$ \citep{azzalini2005skew}, where $W$ is the diagonal matrix with diagonal entries $W_{ii}^2 = S_{0,\, ii}$ for $i=1,\ldots,p$. We choose $m_0 = 0_{p}$ and the skewness parameter vector $s_0 = 0.5 \times \mathbbm{1}_{p}$. \\
    (3) {\em Multivariate t-distribution (T)}: 
     $f_0(x) = t_{d_0}(x ; m_*, S_0)$ is the density of the $p$-dimensional multivariate Student's t-distribution. We set $d_0 = 10$ and $m_* = \mathbbm{1}_{p}$. \\  
    (4) {\em Mixture of multivariate skew t-distributions (MST)}: 
    $f_0(x) = 0.25 \, t_{d_0}(x ; m_1, S_0, s_0) + 0.75 \, t_{d_0}(x ; m_2, S_0, s_0)$. Here, $t_{d}(\cdot \, ; \mu, S, s)$ is the skew t-density \citep{azzalini2005skew} with parameters $d, \mu, S, s$, with $d_0, s_0$ defined as before and $m_1, m_2$ the same as in the first case. \\ 
    (5) {\em Multivariate Cauchy (MVC)}: $f_0(x) \propto \{1 + (x - \mu_*)^\T S_0^{-1}(x - \mu_*)\}$ where $\mu_* = 0_p$. \\
    (6) {\em Multivariate Gamma (MVG)}: $f_0(x) \propto c_{\Phi} (F_1(x_1), \ldots, F_p(x_p) \mid S_0)\,\prod_{j=1}^{p} f_j(x_j;\gamma_{j1}, \gamma_{j2})$ where $f_j$ and $F_j$ denote the density and distribution function of the univariate gamma distribution with shape parameter $\gamma_{j1}$ and rate parameter $\gamma_{j2}$, respectively, for $j=1,\ldots,p$ and $c_{\Phi}(\cdot \mid \Gamma)$ is as described in \cite{xue2000multivariate}. This is a Gaussian copula based construction of the  multivariate gamma distribution. 
    We set $\gamma_{j1}=\gamma_{j2}=1$ for $j=1,\ldots,p$.
    
    The hyperparameters for the nearest neighbor-Dirichlet mixture are chosen as $\mu_0 = 0_{p}, \nu_0 = 0.001, \gamma_0 = p$, and $\Psi_0 = \{(\gamma_{0}-p+1)\delta_{0}^2\}\mathbbm{I}_{p} = \delta_{0}^2 \, \mathbbm{I}_p$, where the optimal $\delta_0^2$ is chosen via cross-validation as described in Section \ref{sec:CV}. Default hyperparameters as described in Section \ref{subsection:simulation_preliminaries} are chosen for the MCMC and VB implementations of the DPM.

\begin{figure}
\centering
\begin{subfigure}{.4\textwidth} 
  \centering
\includegraphics[height = 4cm, width = 6cm]{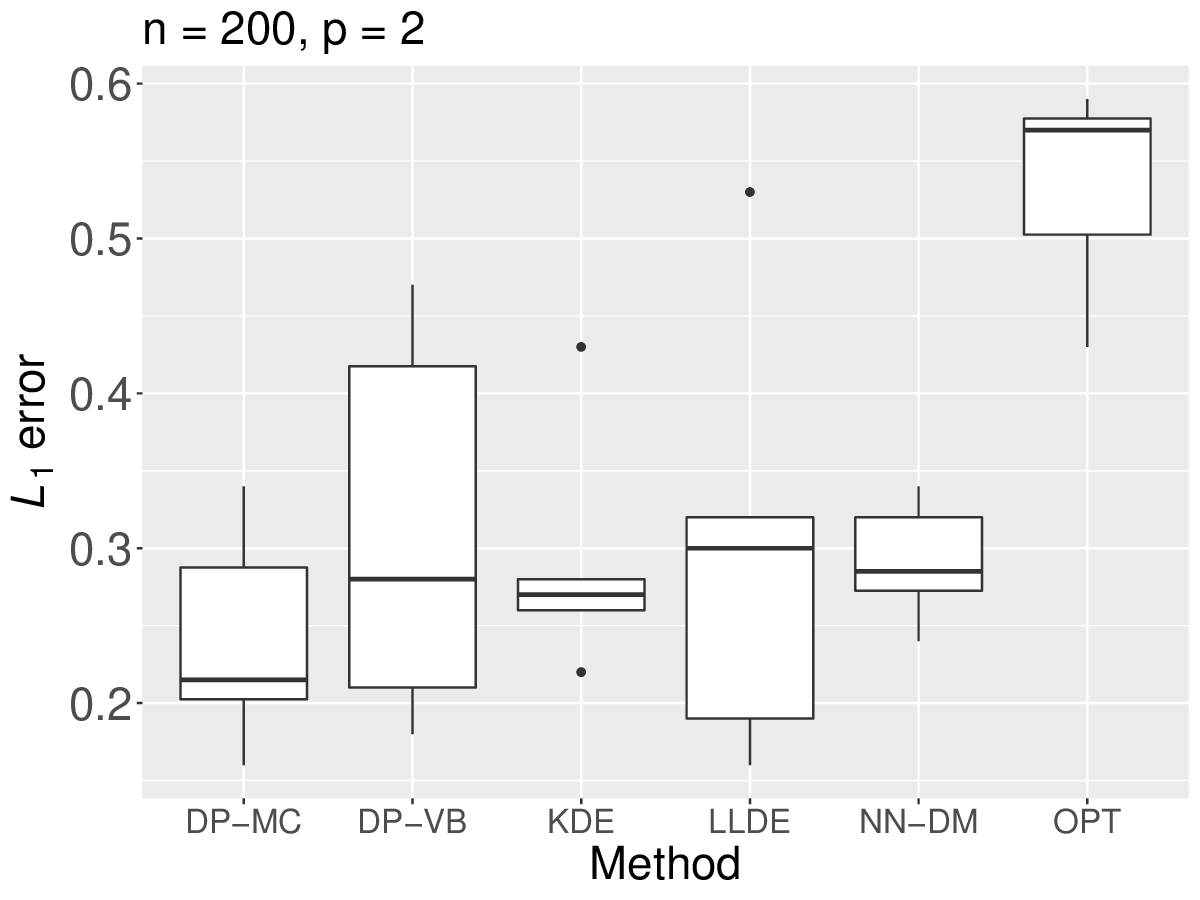}
\end{subfigure}
\begin{subfigure}{.4\textwidth} 
  \centering
\includegraphics[height = 4cm, width = 6cm]{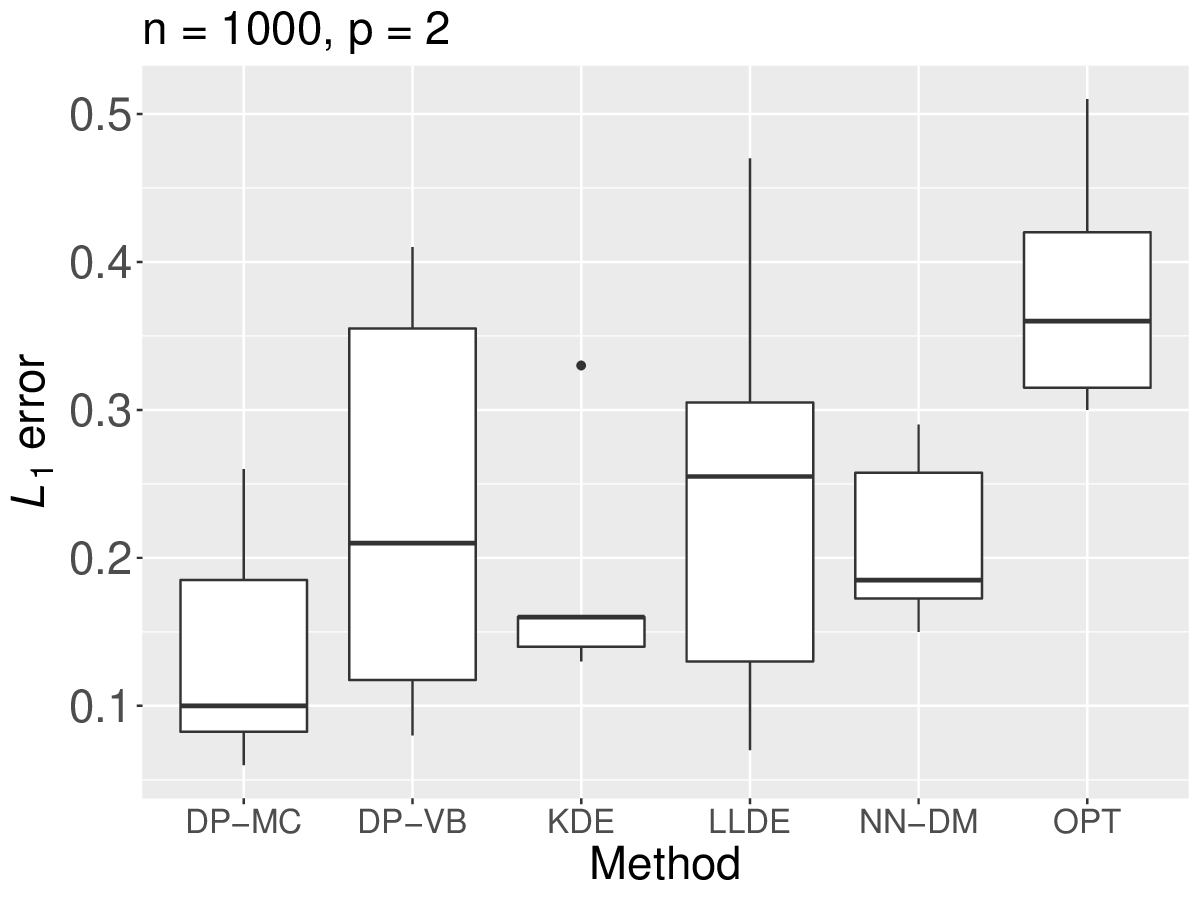}
\end{subfigure}
\begin{subfigure}{.4\textwidth} 
\includegraphics[height = 4cm, width = 6cm]{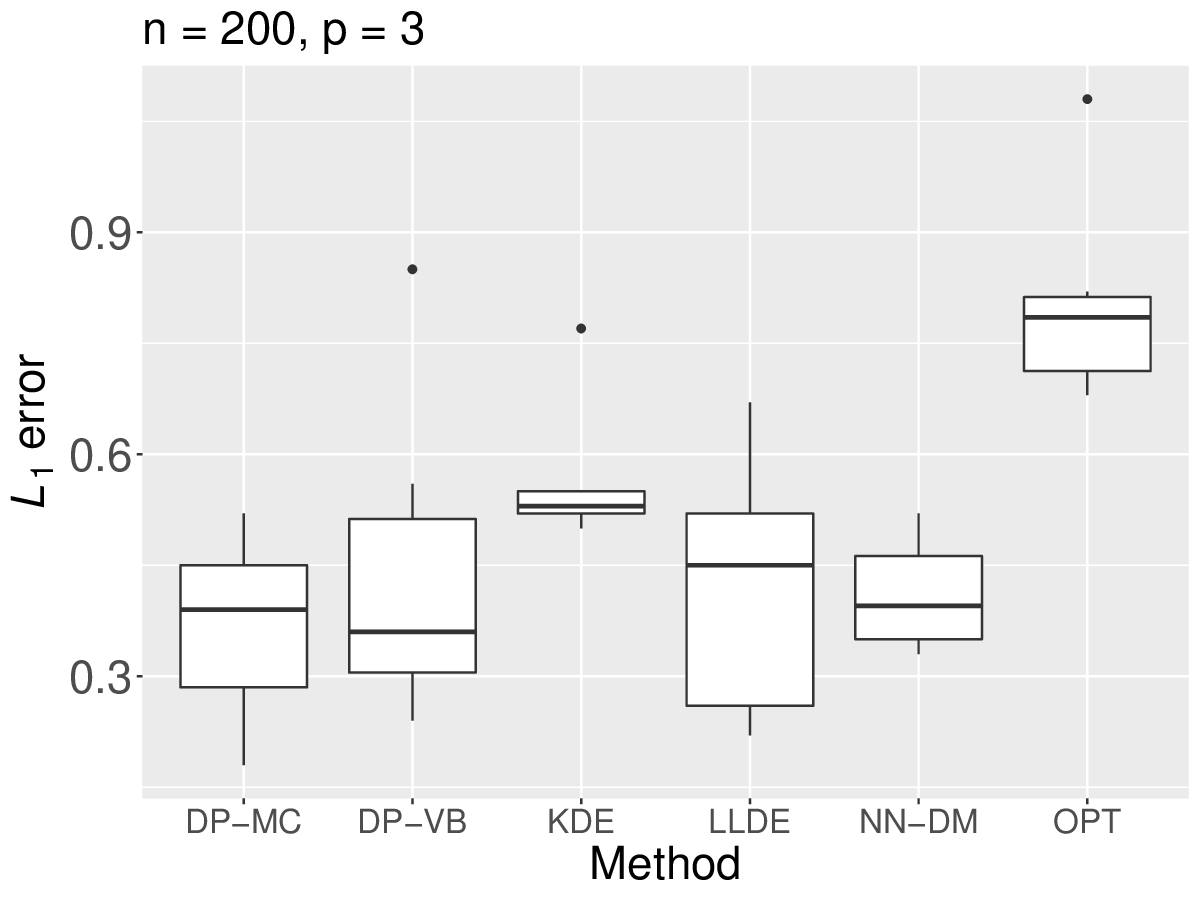}
\end{subfigure}
\begin{subfigure}{.4\textwidth} 
  \centering
\includegraphics[height = 4cm, width = 6cm]{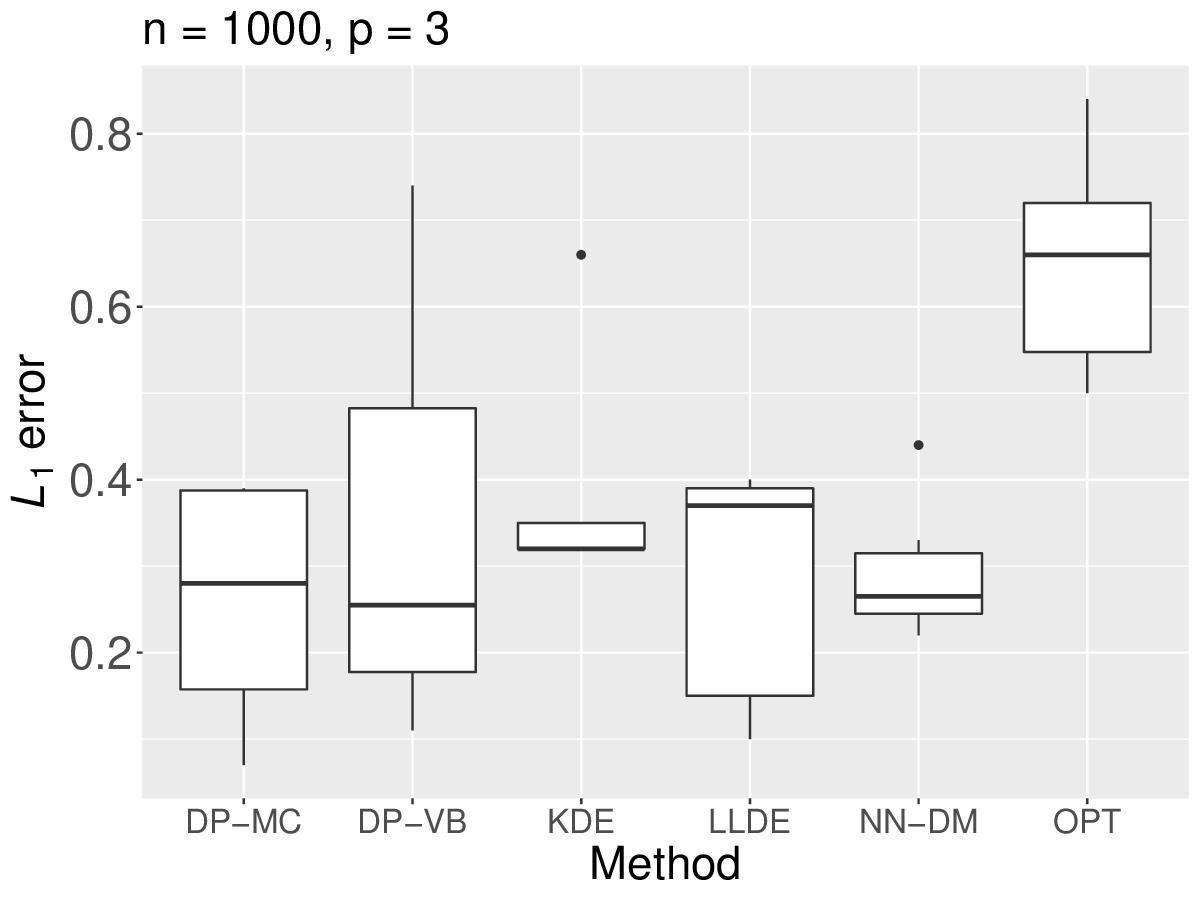}
\end{subfigure}
\begin{subfigure}{.4\textwidth} 
  \centering
\includegraphics[height = 4cm, width = 6cm]{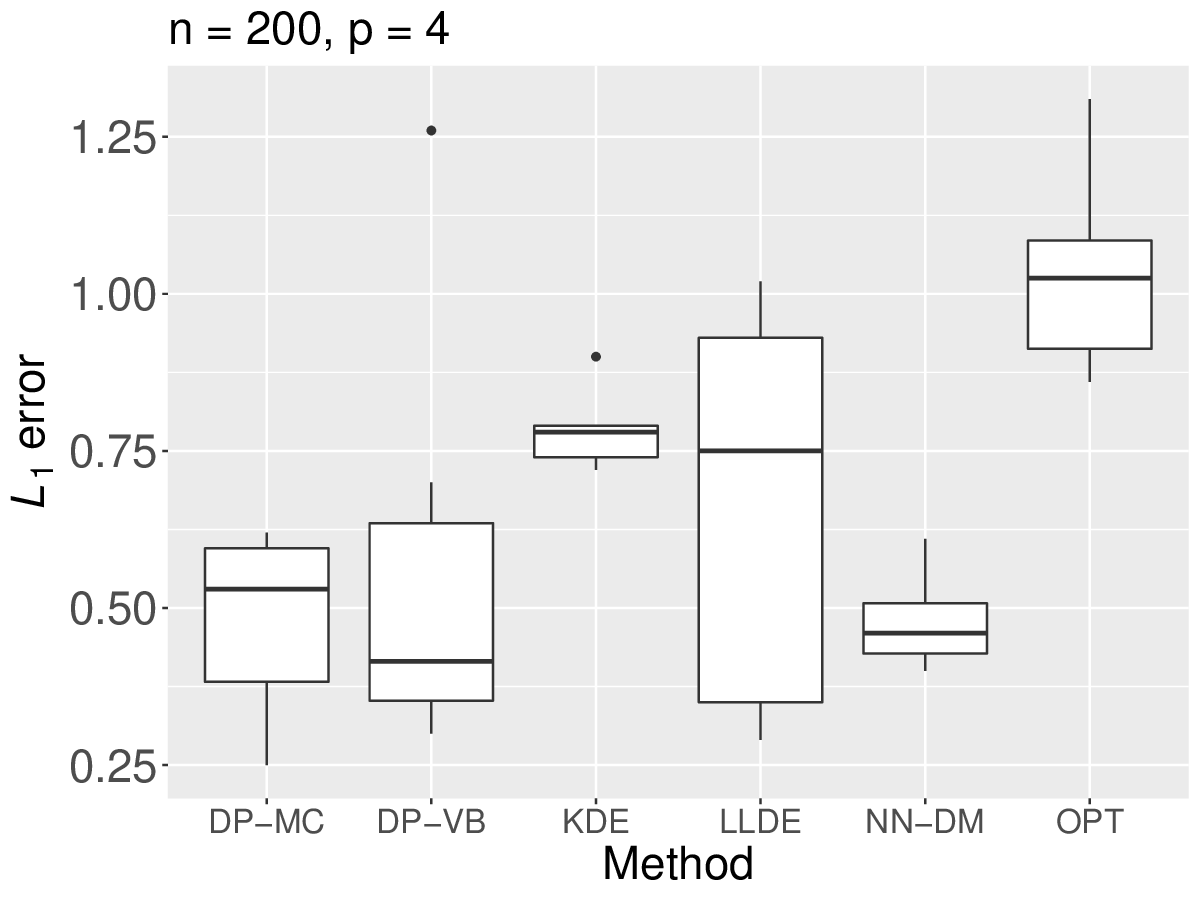}
\end{subfigure}
\begin{subfigure}{.4\textwidth} 
  \centering
\includegraphics[height = 4cm, width = 6cm]{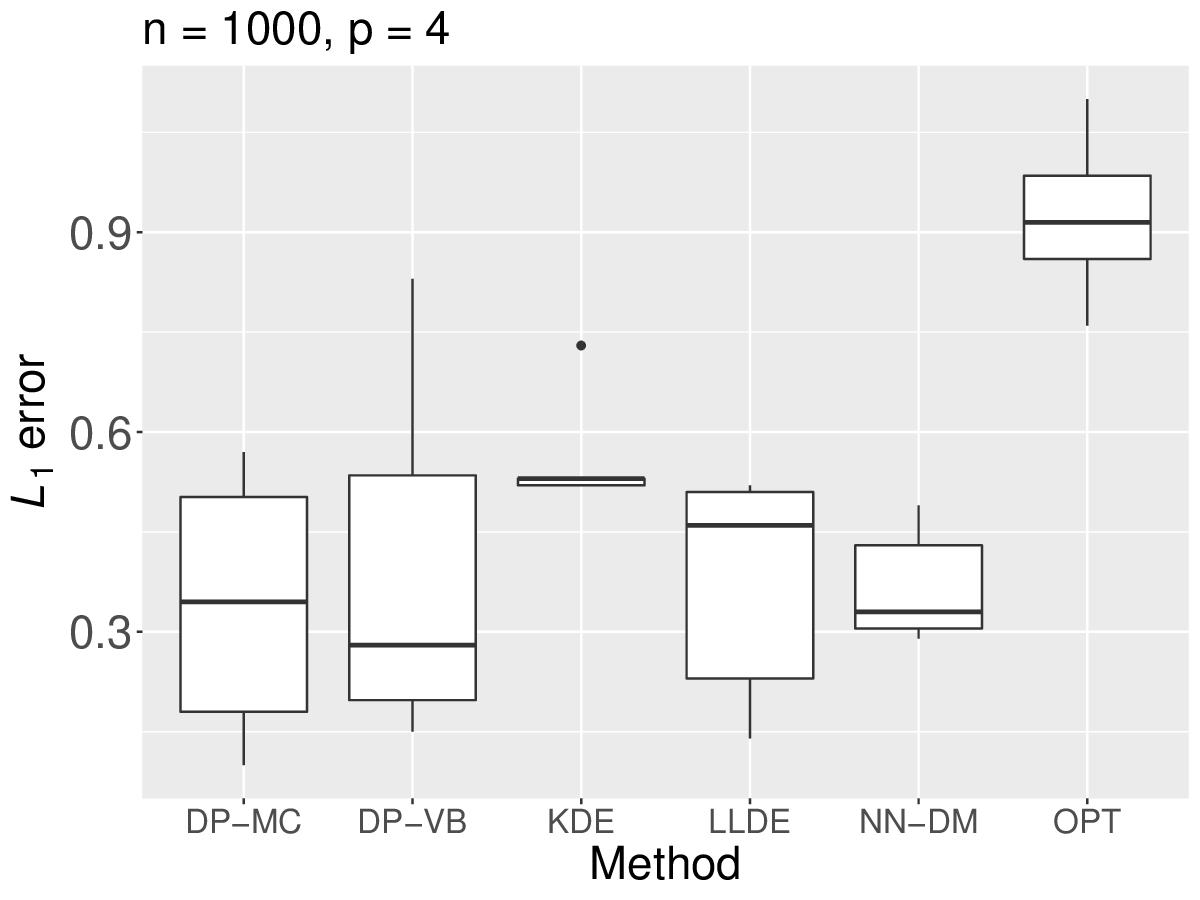}
\end{subfigure}
\begin{subfigure}{.4\textwidth} 
  \centering
\includegraphics[height = 4cm, width = 6cm]{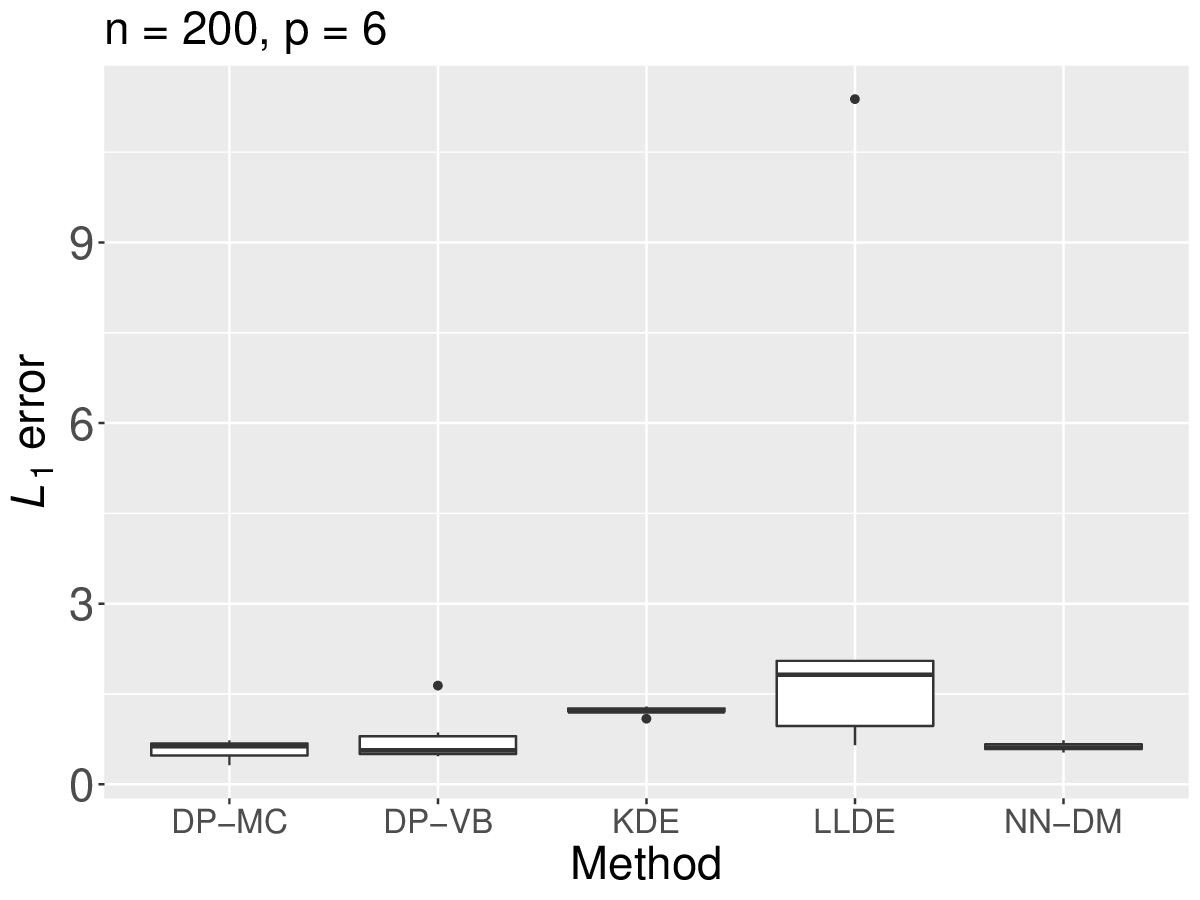}
\end{subfigure}
\begin{subfigure}{.4\textwidth} 
  \centering
\includegraphics[height = 4cm, width = 6cm]{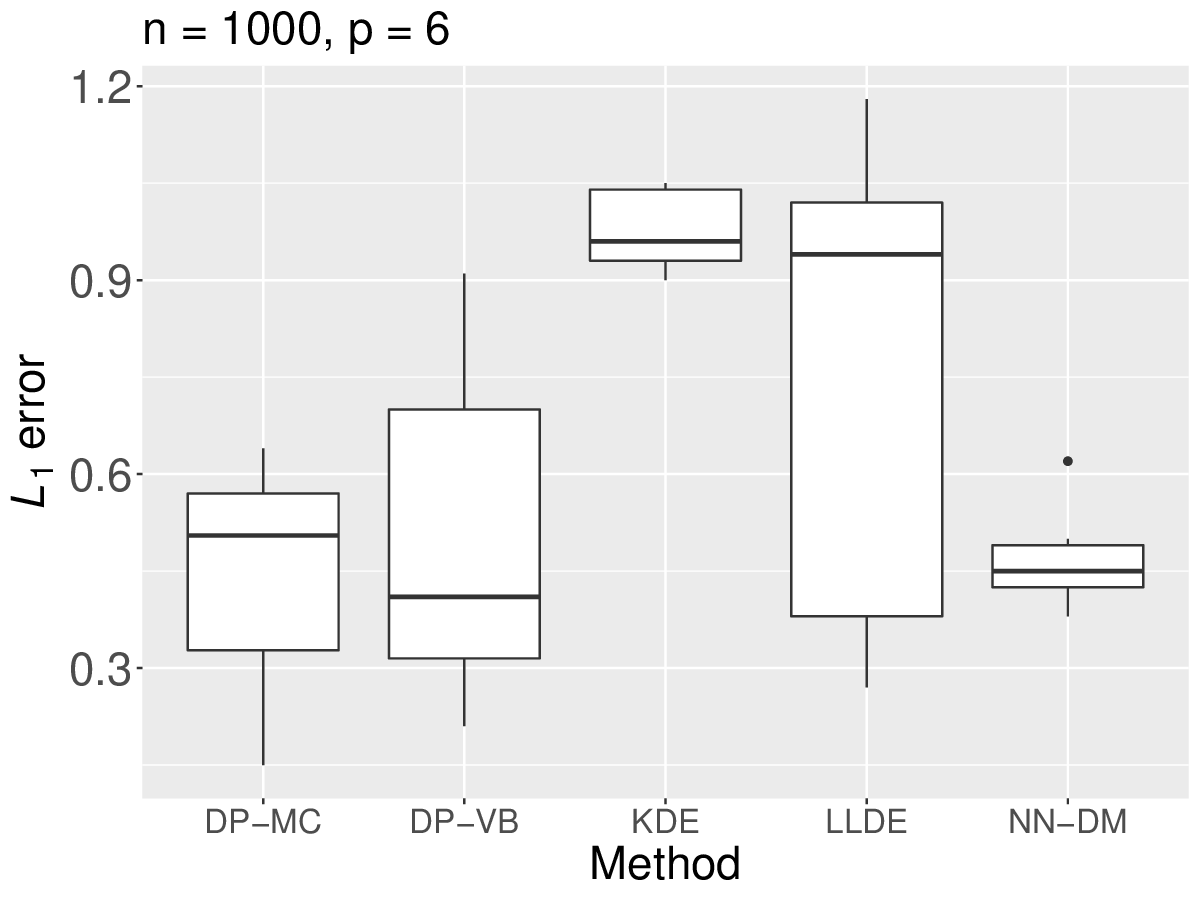}
\end{subfigure}

\caption{Box plots of $\hat{\mathcal{L}}_1(f_0, \hat{f})$ for the 6 different choices of the true density $f_0$ and different estimators $\hat{f}$ for multivariate data. The box plots for KDE and LLDE exclude the MVC density. The box plots for $p=6$ exclude results from OPT.}
\label{fig:multivariate_results}
\end{figure}

Similar to the univariate case, we defer the numerical results to Table \ref{tab:multivariate_tab} in the Appendix and in Figure \ref{fig:multivariate_results} display a visual summary consisting of box plot of estimated $\mathcal{L}_1$ errors over the densities considered. The proposed method is very robust against a wide selection of true distributions, with its $\mathcal{L}_1$ error scaling nicely with the dimension. The KDE shows a noticeably sharp decline in performance - when the dimension is changed from $2$ to $6$, the average increase in $\mathcal{L}_{1}$ error is by factors of about $5$ and $7$ for sample sizes $200$ and $1000$, respectively. 
This is possibly due to lack of adaptive density estimation in higher dimensions using a single bandwidth matrix, since data in $\mathbb{R}^p$ become increasingly sparse with increasing $p$. As in the univariate case, we had to exclude the MVC density for the KDE due to the algorithm not converging. 
The performances of NN-DM, DP-MC, and DP-VB are quite competitive across densities, with NN-DM faring better than the DP-VB when estimating densities such as the MVC and the MVG. Furthermore, the NN-DM is hit the least significantly by the curse of dimensionality out of the three. This is particularly prominent for the DP-MC when the true density is either MG or MST with $n = 200$ and $p = 6$, and for the DP-VB when the true density is MVC. It is also important to keep in mind that the NN-DM provides similar results compared to the DP-MC while being at least an order of magnitude faster, as illustrated in Section \ref{sec:runtime}.
The performance of the OPT  is hit quite significantly as the number of dimensions increases, along with the algorithm not converging for $p=6$. The LLDE  provides competitive results with the NN-DM in lower dimensions. However, in higher dimensions, the LLDE often does not converge, indicating lack of stability of the algorithm. We reported the average of the replicates for which the algorithm did converge. The results suggest that the performance of the LLDE is also affected quite drastically with increasing dimensions. 
When compared across all data generating cases considering the variation in densities, dimensions and sample sizes, the proposed method is seen to be more versatile than its competitors.

\subsection{Accuracy of Uncertainty Quantification}
\label{sec:coverage}

In this section, we assess frequentist coverage of $95\%$ pseudo-posterior credible intervals for the NN-DM and compare with coverage based on the $95 \%$ posterior credible intervals obtained from DP-MC and DP-VB. \cite{ghosal2017fundamentals} recommend investigating the frequentist coverage of Bayesian credible intervals. We do not include frequentist coverage for Polya tree mixtures (PTMs) and the optional Polya tree (OPT) due to the lack of available code. We consider the cases $p \in \{1,2\}$ in our experiments with sample size $n = 500$. For each choice of density $f_0$, we fix $n_t = 200$ test points $\mathcal{X}_t = \{X_{t1}, \ldots, X_{tn_{t}}\}$ generated from the density $f_0$. With these fixed test points, we generate $n = 500$ data points in our sample for $R_{cov} = 200$ times and check the coverage of posterior/pseudo-posterior credible intervals obtained from the three methods. We implement the DP-MC with base measure $\text{NIW}_p(0_p, 0.01, p, \mathbbm{I}_p)$ and a $\text{Gamma}(2,4)$ prior on the concentration parameter as in \cite{west1992hyperparameter}. These choices of hyperparameters were seen to give better frequentist coverage results than using the default values used in Sections \ref{sec:univ_simulations} and \ref{sec:mvt_simulations}. Same choices of hyperparameters are maintained for DP-VB. For the NN-DM, we take $k = 8$ in the univariate case and $k = 5$ in the bivariate case, $\alpha = 0.001$,
and other hyperparameters chosen as before. We report the average coverage probability and average length of the (pseudo) credible intervals across all the points in the test data $\mathcal{X}_t$ in Tables \ref{tab:univariate_coverage} and \ref{tab:mvt_coverage} for the univariate and bivariate cases, respectively.

For univariate densities, both the DP-MC and DP-VB display severe under-coverage. In most of the cases, the DP-VB and NN-DM have similar width of (pseudo) credible intervals but the DP-VB displays dramatically lower coverage than the NN-DM. The under-coverage displayed by the DP-MC may be due to MCMC mixing issues. The NN-DM shows near nominal coverage in the smooth Gaussian (GS) and lognormal (LN) densities, while also attaining near nominal coverage in the skewed bimodal (SB), claw (CW), and sawtooth (ST) densities which are multi-modal. The shortcomings of DP-MC and DP-VB are especially noticeable when dealing with spiky densities such as the claw or sawtooth. For bivariate cases considered in Table \ref{tab:mvt_coverage}
we see a similar trend; the NN-DM method provides uniformly better uncertainty quantification across all the densities considered. It is clear that in terms of frequentist uncertainty quantification, the NN-DM displays vastly superior coverage to the DP-MC and the DP-VB without inflating the interval width.
%
\begin{center}
\begin{table}
\huge
\centering
\scalebox{0.50}{
\begin{tabular}{cccccc}\toprule
Method & CA & CW & DE & GS & IE \\
\cmidrule{1-6}\\								
\multirow{3}{*}{} 
NN-DM & 0.75 (0.05) & 0.89 (0.21) & 0.75 (0.06) & 0.92 (0.08) & 0.81 (0.11)\\
 DP-MC & 0.48 (0.02) &	0.06 (0.01) & 0.35 (0.02) & 0.37 (0.01) & 0.39 (0.04) 								 \\
 DP-VB & 0.33 (0.05) & 0.18 (0.07) & 0.28 (0.07) & 0.79 (0.05) & 0.14 (0.04) \\
 \hline \\
Method  & LN & LO & SB & SP & ST\\
 \hline\\
 NN-DM & 0.92 (0.17)& 0.81 (0.03) & 0.88 (0.10) & 0.72 (0.01) & 0.91 (0.05) \\
 DP-MC & 0.31 (0.05) & 0.55 (0.03) & 0.46 (0.03) & 0.46 (0.01) & 0.64 (0.03) \\
 DP-VB & 0.19 (0.15) & 0.10 (0.03) & 0.40 (0.10) & 0.20 (0.01) & 0.07 (0.01) \\
\cmidrule{1-6}
\end{tabular}
}
\caption{Comparison of the frequentist coverage of $95\%$ (pseudo) posterior credible intervals of the nearest neighbor-Dirichlet mixture and the MCMC and variational implementations of the Dirichlet process mixture for univariate data. Average length of the intervals are also provided for each case within parentheses. Number of replications and sample size are $R_{cov} = 200$ and $n_{cov} = 500$, respectively.}
\label{tab:univariate_coverage}
\end{table}
\end{center}
\begin{center}
\begin{table}
\huge
\centering
\scalebox{0.50}{
\begin{tabular}{ccccccc}\toprule
Method & MG & MST & MVC & MVG & SN & T \\
\cmidrule{1-7}\\
  NN-DM 					 & 0.92 (0.04) & 0.88 (0.03) & 0.69 (0.03) & 0.80 (0.31) & 0.92 (0.06)& 0.88 (0.03) \\			
 DP-MC 	& 0.53 (0.01) & 0.56 (0.01) & 0.47 (0.01) & 0.41 (0.16) & 0.39 (0.02) & 0.55 (0.01) \\								
  DP-VB & 0.56 (0.03) & 0.58 (0.03) & 0.18 (0.02) & 0.55 (0.26) & 0.49 (0.05) & 0.57 (0.02)\\
  \cmidrule{1-7}
\end{tabular}
}
\caption{Comparison of the frequentist coverage of $95\%$ (pseudo) posterior credible intervals of the nearest neighbor-Dirichlet mixture and the MCMC and variational implementations of the Dirichlet process mixture for bivariate data. Average length of the intervals are also provided for each case within parentheses. Number of replications and sample size are $R_{cov} = 200$ and $n_{cov} = 500$, respectively.}
\label{tab:mvt_coverage}
\end{table}
\end{center}
\subsection{Comparison for high dimensional data}
\label{sec:high-dim}
In addition to the above experiments, we performed a simulation experiment for high-dimensional data. Specifically, we set $n=1000$, $p = 50$, and consider the same set of true densities in Section \ref{sec:mvt_simulations}. We compared results from the proposed NN-DM method and the DP-VB. Due to severe computational time, we did not consider the DP-MC in this scenario. We also tried optional Polya trees \citep{wong2010optional} using the \texttt{PTT} package; however, the current implementation of the method breaks down in this high-dimensional setup. Due to numerical instability in estimating the $\mathcal{L}_1$ error in higher dimensions, we evaluate the methods in terms of their out-of-sample log-likelihood (OOSLL) instead \citep{gneiting2007strictly}, on a test set of 500 data points. We report the average OOSLL over 30 replications in Table \ref{tab:large_p_table}. The results indicate that both methods perform very similarly in terms of out-of-sample fit to the data,
with the NN-DM outperforming the DP-VB when the true density is MVC. We also observed that for this experiment, the NN-DM methods with default choice of hyperparameters and with cross-validated choice of $\Psi_0$ have almost identical performance. For the NN-DM, we set $k = 12$ after carrying out a sensitivity analysis on $k$ by considering $k = 5, 7, 10, 15,$ and $20$. The best results for the NN-DM were obtained for $k \in \{7, 10, 12\}$ with negligible difference in out-of-sample log-likelihoods between these three choices, with $k=12$ performing the best.
\begin{table}[]
    \centering
    \begin{tabular}{ccccccc}
    \hline
        Method & MG & SN & T & MST & MVC & MVG \\
        \hline\\
        NN-DM & $-1.75$ & $-1.74$ & $-1.84$ & $ -1.84$ &  $-1.31$ & $ -1.36$ \\
        DP-VB & $ -1.75$ &  $ -1.74$ &  $ -1.86$ & $ -1.84$ &  $-1.34$ &  $ -1.36$ \\
        \hline
    \end{tabular}
    \caption{Out-of-sample log-likelihood ($ \times 10^4$) of NN-DM and DP-VB on a test set of 500 points for 6 different multivariate densities considered in Section \ref{sec:mvt_simulations}, for $n = 1000$ and $p = 50$. 
    Greater out-of-sample log-likelihood is better.}
\label{tab:large_p_table}
\end{table}

\subsection{Runtime Comparison} \label{sec:runtime}


With $n$ data points in $p$ dimensions, the initial nearest neighbor allocation into $n$ neighborhoods can be carried out in $\mathcal{O}(n\log n)$ steps \citep{knn_nlogn, ma2019true}. Once the neighborhoods are determined with $k_n$ points in each neighborhood, obtaining the neighborhood specific empirical means and covariance matrices has $\mathcal{O}(nk_n p+nk_n p^2) = \mathcal{O}(nk_n p^2)$ complexity. Obtaining the pseudo-posterior mean \eqref{eq:multivariate_posterior_mean} then requires inversion of $n$ such $p \times p$ matrices to evaluate the multivariate t-density, with a runtime of $\mathcal{O}(np^3)$. Therefore, the total runtime to obtain the pseudo-posterior mean is of the order $\mathcal{O}(nk_n p^2 + np^3)$. When we are interested in uncertainty quantification, we require Monte Carlo samples of the NN-DM, which are independently drawn from its pseudo-posterior. This involves sampling the Dirichlet weights, the neighborhood specific unknown mean and covariance matrix parameters of the Gaussian kernel, and evaluating a Gaussian density for each neighborhood, as outlined in Algorithm \ref{algo:NNDP_multivariate}. To obtain $M$ Monte Carlo samples, the combined complexity of this step is thus $\mathcal{O}(Mn + Mnp^3) = \mathcal{O}(Mnp^3)$. Overall the runtime complexity to obtain NN-DM samples is therefore $\mathcal{O}(Mnp^3 + nk_n p^2 + np^3)$. For high dimensional scenarios, this runtime can be greatly improved by using a low rank matrix factorization of both the neighborhood specific empirical covariance matrices and the sampled covariance matrix parameters to make matrix inversion more efficient \citep{golub2012matrix}. 
We now provide a detailed simulation study of runtimes of the proposed method, with all the simulations carried out on an M1 MacBook Pro with $16$ GB of RAM.

We first focus on some runtime experiments comparing NN-DM and DP-MC. In our experiments, we focus on $p=1$ and $p=4$. The runtime for NN-DM consists of the time to estimate $\delta_0^2$ by cross-validation as in Section \ref{sec:CV} and then drawing samples from its pseudo-posterior.  For both dimensions, the sample size is varied from $n = 200$ to $n = 1500$ in increments of 100. Data are generated from the standard Gaussian density (GS) for $p=1$ and from a mixture of skew t-distributions with the parameters as described for the case MST in Section \ref{sec:mvt_simulations} for $p=4$. For $p=1$, we evaluate the two methods at $500$ test points, while for $p=4$ we evaluate the methods at $200$ test points. The hyperparameters are kept the same as in Sections \ref{sec:univ_simulations} and \ref{sec:mvt_simulations}. We took $1000$ Monte Carlo samples for the NN-DM and $2500$ MCMC samples for the DP-MC with a burn-in of $1500$ samples. We provide a figure summarising the results in Figure \ref{fig:time_mvt}. In the top panel of Figure \ref{fig:time_mvt}, we plot the average of the logarithm (base $10$) of the run times of each approach for $10$ independent replications. Corresponding $\mathcal{L}_1$ errors of the two methods is  included in the bottom panel of Figure \ref{fig:time_mvt}.

In Figure \ref{fig:time_mvt}, the NN-DM is at least an order of magnitude faster than DP-MC. The time saved becomes more pronounced in the multivariate case, where for sample size $1500$ the NN-DM is $\sim 50$ times faster.
The gain in computing time does not come at the cost of accuracy as can be seen from the right panel; the proposed method maintains the same order of $\mathcal{L}_1$ error as the DP-MC in the univariate case and often outperforms the DP-MC in the multivariate case. We did not implement the Monte Carlo sampler for the proposed algorithm in parallel, but such a modification would substantially improve runtime. Bypassing cross-validation and choosing default hyperparameters instead as outlined in Section \ref{sec:CV}, NN-DM took $3.3$ seconds and $16.4$ seconds when $p=1$ and $p=4$, respectively, with sample size $n = 1500$. In the same scenario, DP-MC took $99.4$ seconds and $1618.1$ seconds for $p=1$ and $p=4$, respectively. Thus the NN-DM with default hyperparameters is about $30$ times faster when $p=1$ and almost $100$ times faster when $p=4$. 

We also compare the runtime of the proposed method with three recent implementations of the DPM, namely the packages \texttt{bnpy} \citep{hughes2014bnpy}, \texttt{DPMMSubClusters} \citep{dinari2019distributed}, and \texttt{vdpgm} \citep{NIPS2006_2bd235c3} available for download at \url{https://kenichikurihara.com/variational-dirichlet-process-gaussian-mixture-model/}. These three packages implement variational approximations of the DPM posterior with different modifications.  
We also include the DP-MC and OPT for comparison. All the runtime results are comparable only up to machine and coding language differences. Amongst the competitor package implementations, the \texttt{NNDM} and \texttt{dirichletprocess} packages are the only ones providing (pseudo) posterior samples of the density estimate at a test point. We consider the average runtime of $R = 10$ replicates to fit a training data set of iid $N(0,1)$ entries with $n=1500$ and $p=4$. For the NN-DM, DP-MC, and OPT, we consider $1000$ (pseudo) posterior samples. Table \ref{tab:runtimecompare} provides the runtimes for the different packages considered. Overall, the fastest implementation is observed for the \texttt{PTT} package. The next fastest implementations are the \texttt{NNDM} without cross-validation (CV), \texttt{DPMMSubClusters}, and \texttt{bnpy}. The runtime for \texttt{NNDM} with CV  closely follows the previous implementations, with both \texttt{NNDM} with and without CV providing (pseudo) posterior samples. The major improvement in runtime for NN-DM is mainly due to the fact that neighborhood allocations are fixed here which is not the case for DP-MC. 

\begin{figure}
\centering
\begin{subfigure}{.4\textwidth} 
  \centering
\includegraphics[height = 4cm, width = 6cm]{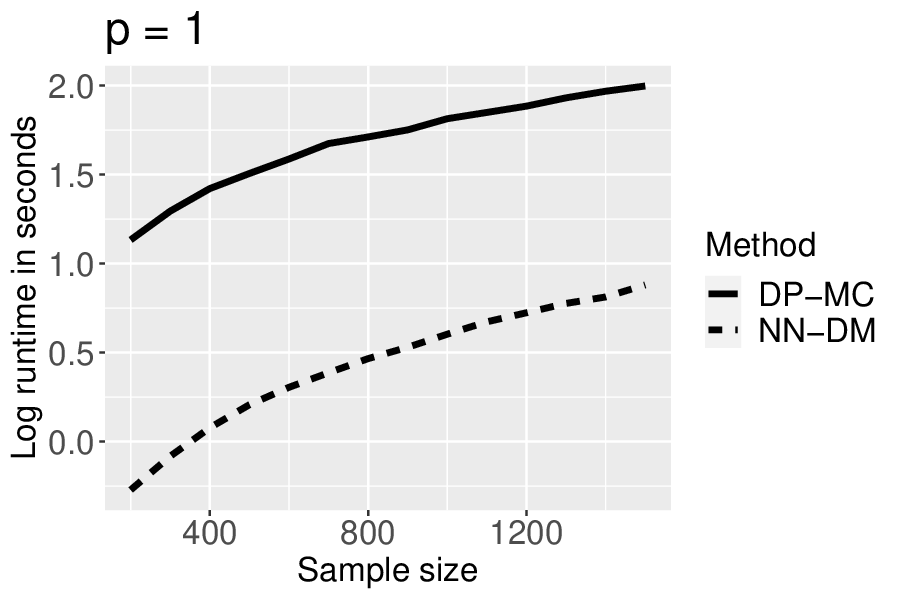}
\end{subfigure}
\begin{subfigure}{.4\textwidth} 
  \centering
\includegraphics[height = 4cm, width = 6cm]{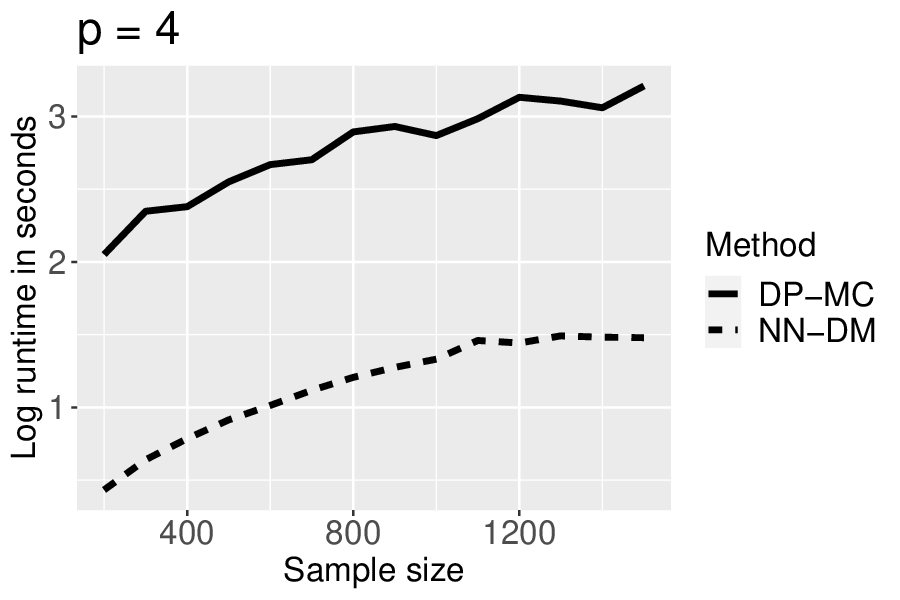}
\end{subfigure}
\begin{subfigure}{.4\textwidth} 
\includegraphics[height = 4cm, width = 6cm]{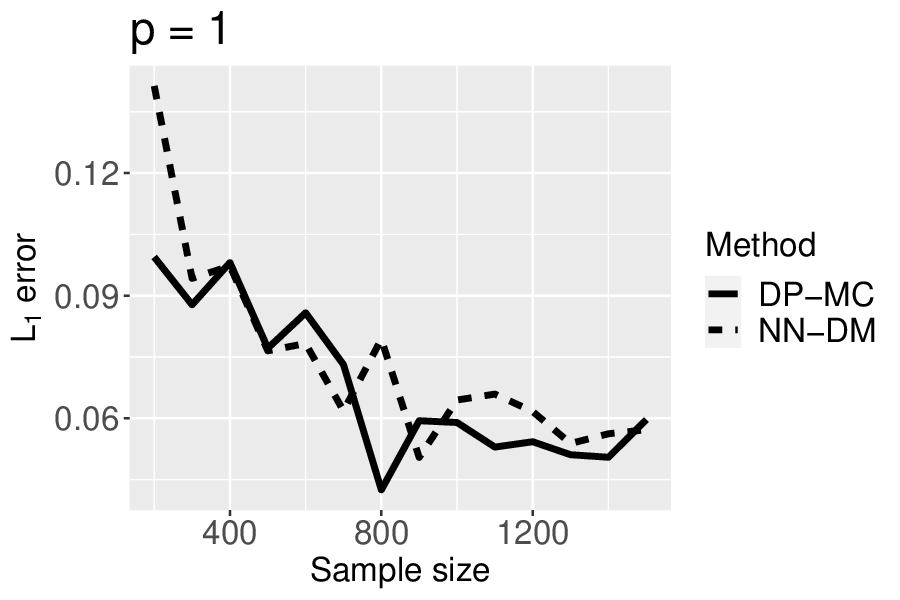}
\end{subfigure}
\begin{subfigure}{.4\textwidth} 
  \centering
\includegraphics[height = 4cm, width = 6cm]{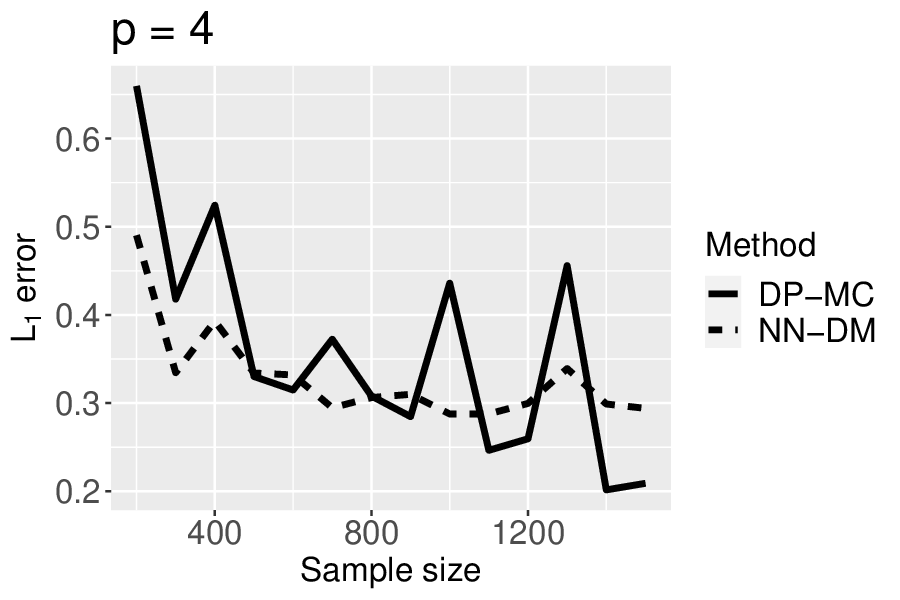}
\end{subfigure}
\caption{Runtime comparison of DP-MC and NN-DM in univariate case and for 4-dimensional data. Top panel shows runtimes in $\log_{10}$ scale whereas bottom panel shows corresponding $\mathcal{L}_1$ error. Sample size $n$ is varied from $200$ to $1500$ in increments of $100$.}
\label{fig:time_mvt}
\end{figure}

\begin{table}[h]
\begin{tabular}{|l|l|l|}
\hline
\textbf{Package (Language)}            & \textbf{Average Runtime (s)} & \textbf{Provides Samples?} \\ \hline
\texttt{bnpy} (Python)                 & 5.79                & No                      \\ \hline
\texttt{DPMMSubClusters} (Julia)       & 4.33                & No                      \\ \hline
\texttt{vdpgm} (MATLAB)                & 58.38               & No                      \\ \hline
\texttt{NNDM} (RCpp and R, with CV)    & 18.96               & Yes                     \\ \hline
\texttt{NNDM} (Rcpp and R, without CV) & 3.52                & Yes                     \\ \hline
\texttt{dirichletprocess} (R)          & 1068.48             & Yes                     \\ \hline
\texttt{PTT} (Rcpp and R)              & 0.59                & No                     \\ \hline
\end{tabular}
\caption{Table comparing the average runtimes of different packages for $n=1500$, $p = 4$. NN-DM runtimes are provided both with and without cross-validation (CV), using the package \texttt{NNDM} developed by the authors.}
\label{tab:runtimecompare}
\end{table}

\subsection{Sensitivity to the choice of $k$}
\label{sec:k-cv-results}

In this subsection, we investigate the role of $k_n = k$ in finite samples for the proposed method. We consider $n=200$ samples from the SP density in the univariate case and the MG density in the bivariate case. In each case, we fix a test set of $n_t = 500$ points, and evaluate the out-of-sample log-likelihood (OOSLL) of the test points for $20$ different integer values of $k$ ranging from $2$ to $50$. Finally, we report results averaged from $10$ independent replicates of this setup. We note that for each considered value of $k$, the parameter $\delta_0^2$ was estimated using leave-one-out cross-validation. Figure \ref{fig:k-cv-results} shows how the OOSLL averaged over replicates changes as a function of $k$ for each density considered. The original OOSLL values of the test data points were scaled by the number of test points $n_t = 500$ for better representability.

For the univariate SP density, the optimal value of $k$ which maximizes the average OOSLL is $\widehat{k} = 9$. This is close to the choice of $k = 6$ as taken in Section \ref{sec:univ_simulations}. For the bivariate MG density, we observe that the choice of $k$ maximizing the OOSLL is $\widehat{k} = 12$, which is also close to the choice of $k = 10$ as taken in Section \ref{sec:mvt_simulations}. For both the univariate and the bivariate case, the out-of-sample log-likelihood of the test set shows little variation with changing $k$. This indicates that the estimates obtained from the proposed method are quite robust to the particular choice of $k$. 

\begin{figure}
\begin{subfigure}{.30\textwidth}
\includegraphics[height = 5cm, width = 7cm]{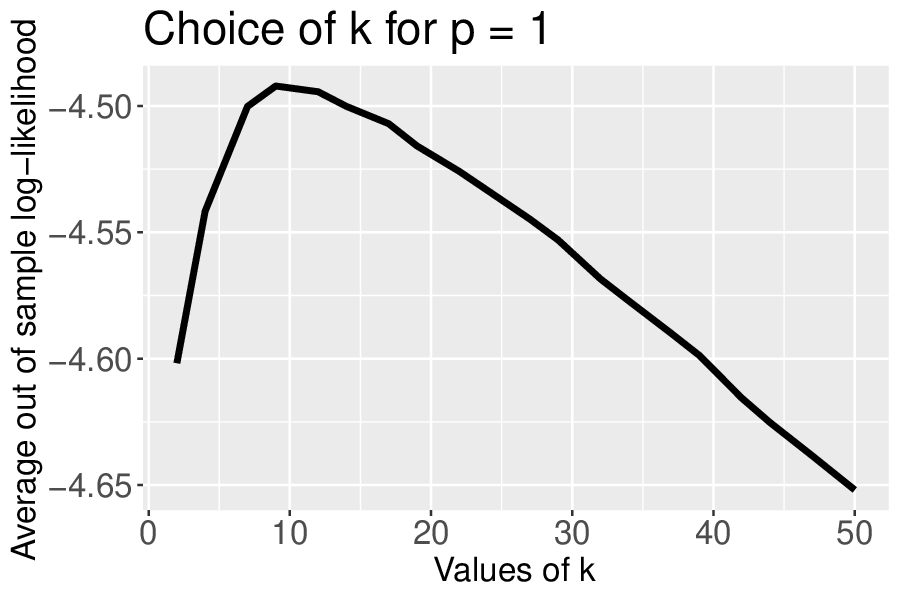}
\end{subfigure}
\hspace{1.2in}
\begin{subfigure}{.30\textwidth}
\includegraphics[height = 5cm, width = 7cm]{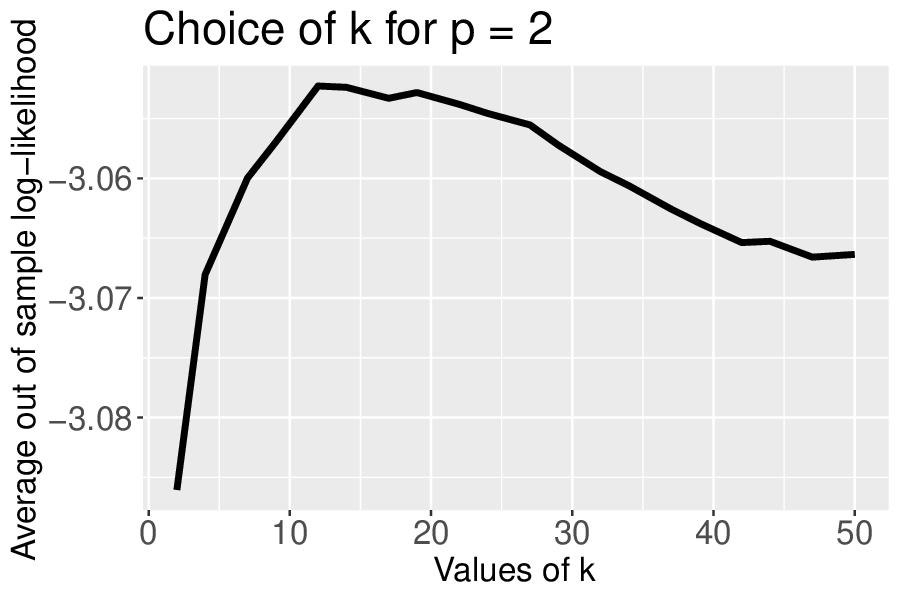}
\end{subfigure}
\caption{Average out-of-sample log-likelihood of $500$ test points for the NN-DM as a function of $k$ for one-dimensional and two-dimensional data. Number of samples and number of replications are $n = 200$ and $R = 10$, respectively.}
\label{fig:k-cv-results}
\end{figure}

\section{Application} 
\label{sec:application-pulsar}

We apply the proposed density estimator to binary classification. Consider data $\mathcal{D} = \{(X_i, Y_i) : i = 1, \ldots, n\}$, where $X_i \in \mathbb{R}^p$ are $p$-dimensional feature vectors and $Y_i \in \{0, 1\} $ are binary class labels. To predict the probability that $y_0 = 1$ for a test point $x_0$, we use Bayes rule:
\begin{equation}
    \mbox{pr}(y_0 = 1 \mid x_0) = \dfrac{\tilde{f}_1(x_0) \, \mbox{pr}(y_0 = 1)}{ \tilde{f}_0(x_0) \, \mbox{pr}(y_0 = 0) + \tilde{f}_1(x_0) \, \mbox{pr}(y_0 = 1)},
    \label{eq:bayes_classifier}
\end{equation}
where $\tilde{f}_j(x_0)$ is the feature density at $x_0$ in class $j$ and $\mbox{pr}(y_0 = j)$ is the marginal probability of class $j$, for $j=0,1$. Based on $n_t$ test data, we let $\widehat{\mbox{pr}}(y_0=1) = (1/n_t)\sum_{i=1}^{n_t} Y_i$, with $\widehat{\mbox{pr}}(y_0=0) = 1 - \widehat{\mbox{pr}}(y_0=1)$. We use either the NN-DM pseudo-posterior mean $\hat{f}_n(\cdot)$, the DP-MC posterior mean $\hat{f}_{\text{DP}}(\cdot)$, or the DP-VB posterior mean $\hat{f}_{\text{VB}}(\cdot)$ for estimating the within class densities, and compare their classification performances in terms of sensitivity, specificity, and probabilistic calibration. We omit the KDE as to the best of our knowledge, no routine \texttt{R} implementation is available for data having more than $6$ dimensions. 

The high time resolution universe survey data \citep{keith2010high} contain information on sampled pulsar stars. Pulsar stars are a type of neutron stars and their radio emissions are detectable from the Earth. These stars have gained considerable interest from the scientific community due to their several applications \citep{lorimer2012handbook}. The data are publicly available from the University of California at Irvine machine learning repository. Stars are classified into pulsar and non-pulsar groups according to 8 attributes \citep{lyon2016pulsars}. There are a total of $17898$ instances of stars, among which $1639$ are classified as pulsar stars. 

We create a test data set of $200$ stars, among which $23$ are pulsar stars. The training size is then varied from $300$ to $1800$ in increments of $300$, each time adding $300$ training points by randomly sampling from the entire data leaving out the initial test set. In Figure \ref{fig:htru_data1}, we plot the sensitivity and specificity of the three methods in consideration.
All the methods exhibit similar sensitivity across various training sizes; the DP-MC has marginally better specificity for training sizes $1200$ and $1500$, while the NN-DM has better specificity for training sizes $300$ and $600$. Both the NN-DM and the DP-MC exhibit higher specificity and sensitivity than the DP-VB across all training sample sizes considered.

\begin{figure}
\begin{subfigure}{.37\textwidth} 
\includegraphics[height = 5.5cm, width = 7cm]{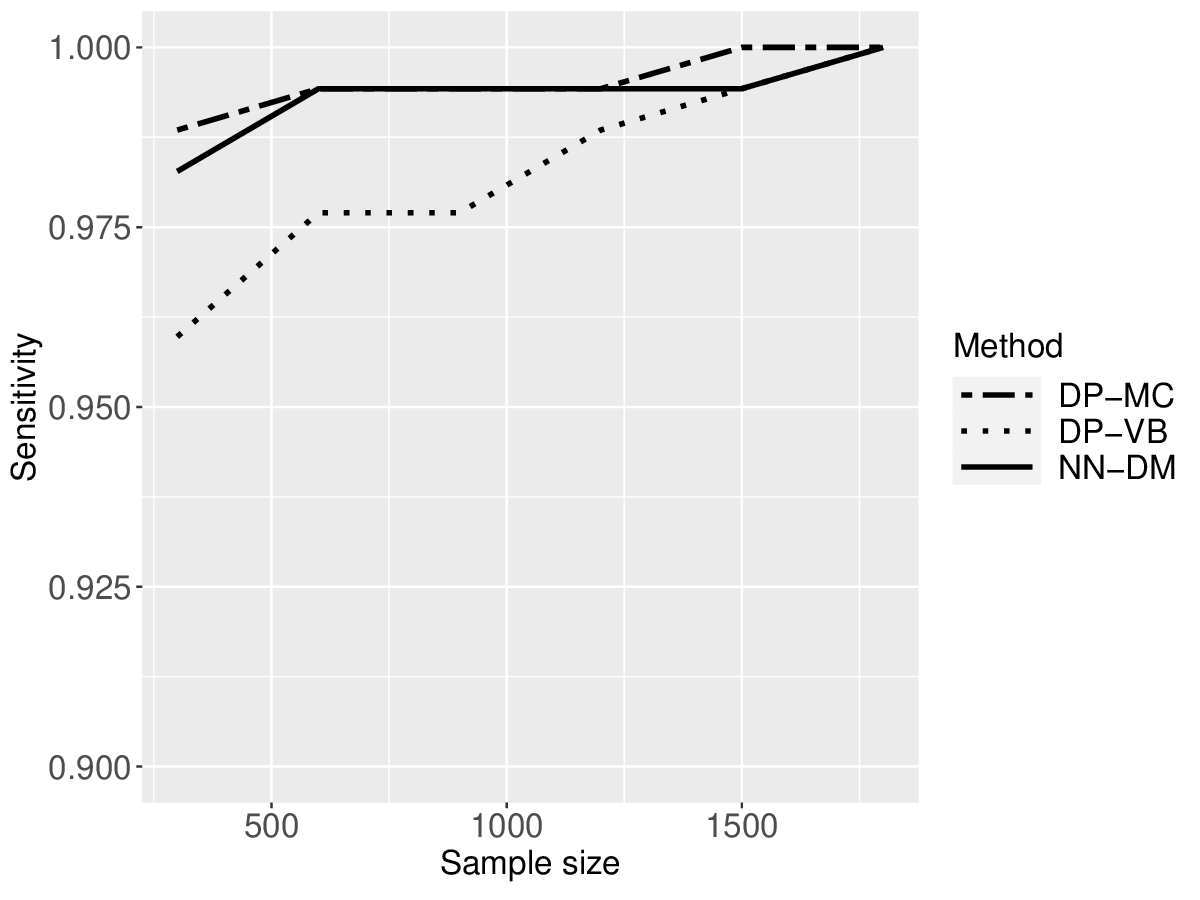} 
\end{subfigure}
\hspace{0.85in}
\begin{subfigure}{.37\textwidth}
\includegraphics[height = 5.5cm, width = 7cm]{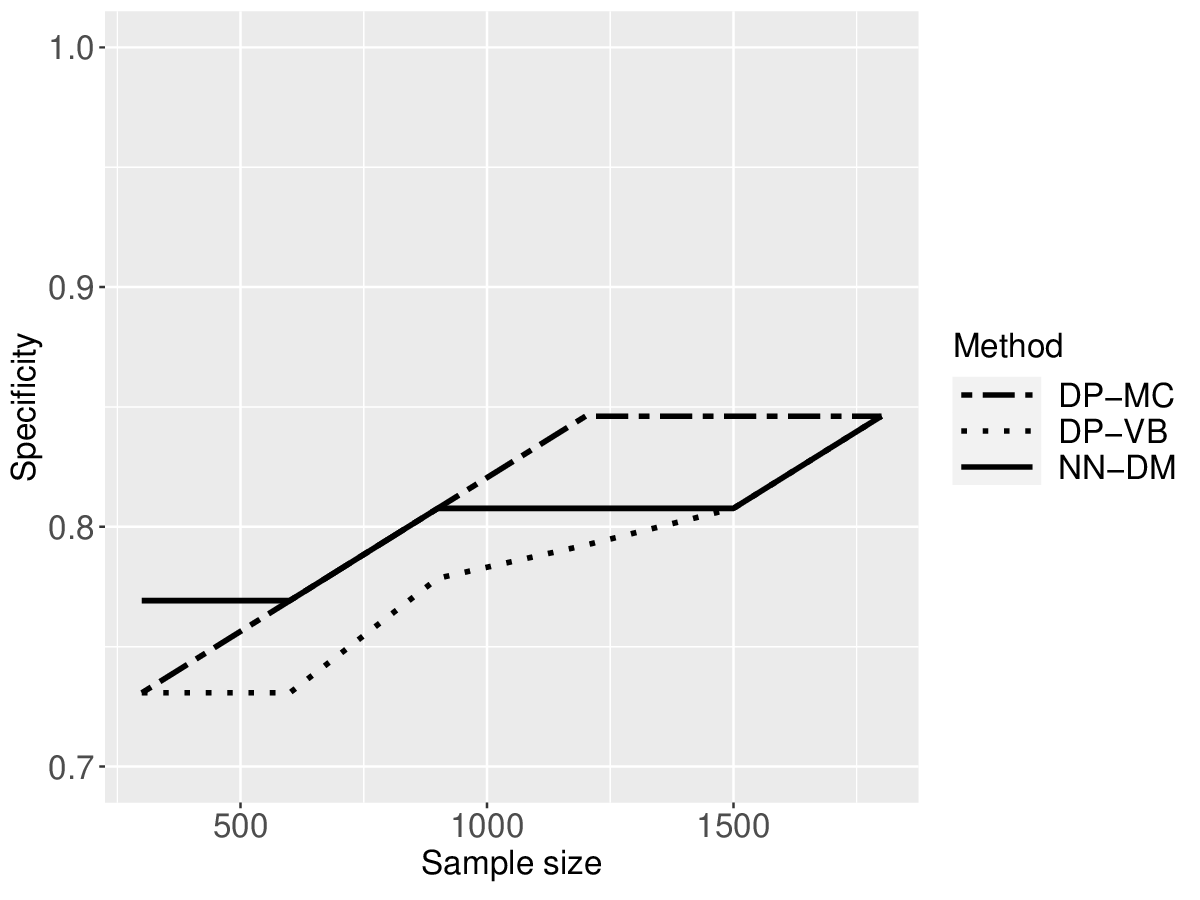}
\end{subfigure}
\caption{Sensitivity and specificity of the NN-DM, DP-MC, and DP-VB for the high time resolution universe survey data.}  
\label{fig:htru_data1}
\end{figure}

We also compare the methods using the Brier score, a proper scoring rule \citep{gneiting2007strictly} for probabilistic classification. Suppose for $n_t$ test points and the $i$th Monte Carlo sample, $p_{i}$ denotes the sampled $n_t \times 1$ probability vector for a generic method. We compute the normalized Brier score for the $i$th sample as $(1/n_t)\,||p_i - Y_t||_{2}^2$, where $Y_t$ is the vector of class labels in the test set. Then with $T$ samples of $p_{i}$, $i=1,\ldots,T$, we compute the mean Brier score for the three methods considered. The mean Brier score for each training size is shown in the right panel of Figure \ref{fig:htru_data2}, which naturally shows a declining trend with increasing training size.  There is little to choose between the three classifiers in terms of mean Brier score; the proposed method fairs equally well in terms of calibration of estimated test set probabilities with the MCMC implementation of the Dirichlet process. In the left panel of Figure \ref{fig:htru_data2}, the receiver operating characteristic curve of the methods is shown for $1800$ training samples. The area under the curve (AUC) for the NN-DM, the DP-MC and the DP-VB are $0.96$, $0.95$ and $0.96$, respectively. For $1800$ training samples, the computation time for the proposed method is about $13$ minutes while for the DP-MC it is approximately $5$ hours.

\begin{figure}
\begin{subfigure}{.4\textwidth}
\includegraphics[height = 5.5cm, width = 7cm]{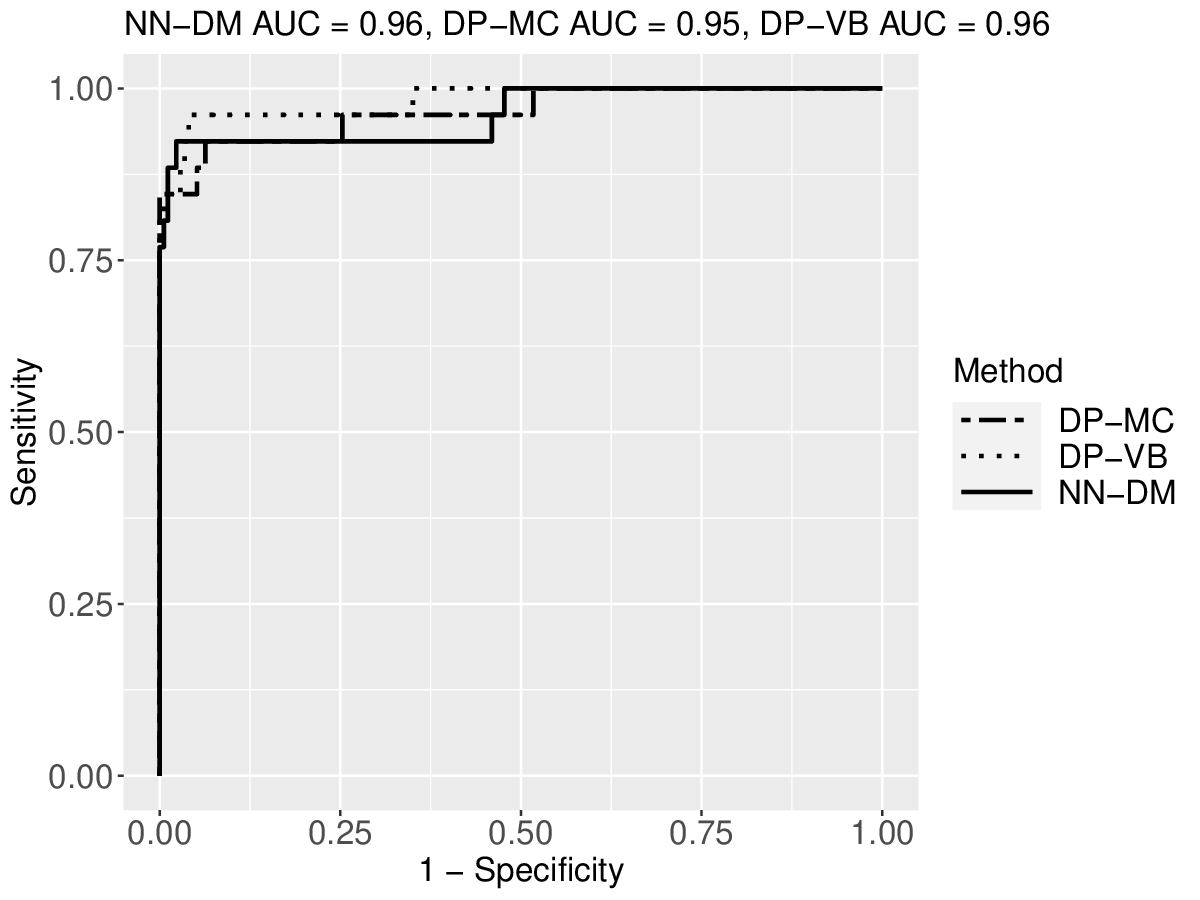}
\end{subfigure}
\hspace{0.8in}
\begin{subfigure}{.4\textwidth}
\includegraphics[height = 5.5cm, width =7cm]{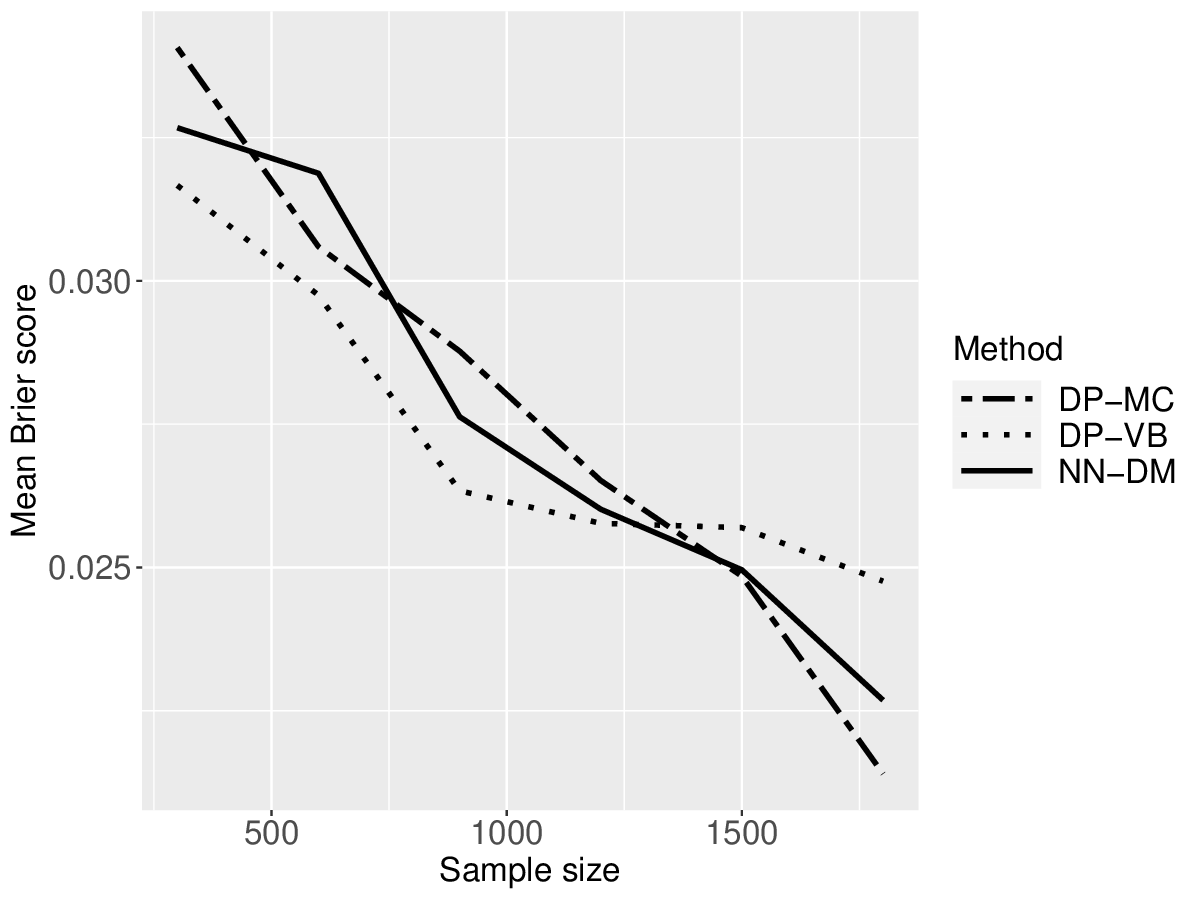}
\end{subfigure}
\caption{Left plot shows the receiver operating characteristic curve of the NN-DM, DP-MC, and DP-VB with $1800$ training samples. Area under the curve is abbreviated as AUC. Right plot shows normalized Brier scores for the methods with varying training sample size.}
\label{fig:htru_data2}
\end{figure} 

Hence, the proposed method is much faster, even without exploiting parallel computation. We also fitted the proposed method using the training set of all $17698$ points; DP-MC was too slow in this case. The sensitivity and specificity of the proposed method increased to $0.99$ and $0.91$, respectively. We additionally evaluated the methods in terms of the out-of-sample log-likelihood. The results are displayed in Figure \ref{fig:pulsar_oosll}. While the methods perform comparably in terms of their classification performance, NN-DM achieves a better fit overall, especially for the significantly less prevalent pulsar star type.

\begin{figure}
\begin{subfigure}{.35\textwidth}
\includegraphics[height = 5.5cm, width = 7cm]{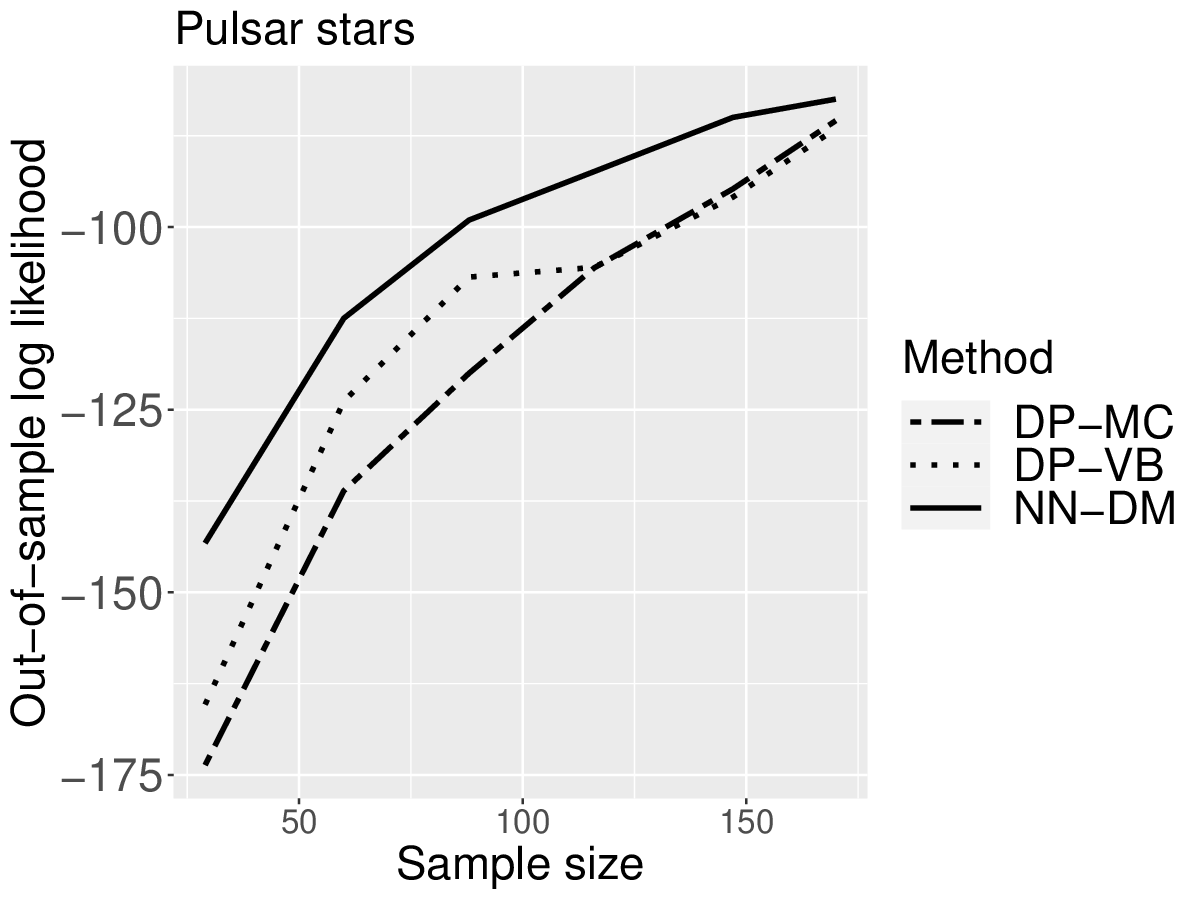}
\end{subfigure}
\hspace{0.85in}
\begin{subfigure}{.35\textwidth}
\includegraphics[height = 5.5cm, width =7cm]{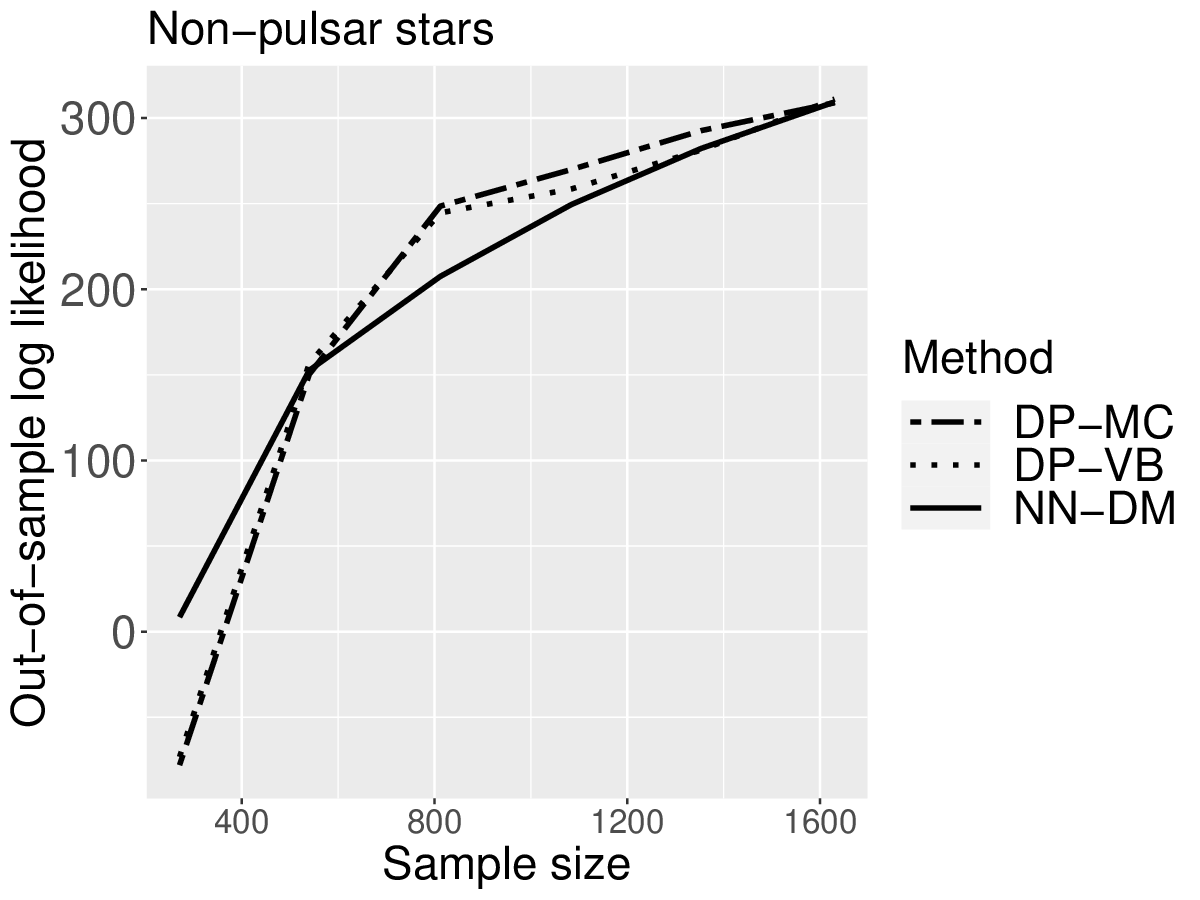}
\end{subfigure}
\caption{Left and right plots show the out-of-sample log-likelihoods of NN-DM, DP-MC, and DP-VB for the two different star types.}
\label{fig:pulsar_oosll}
\end{figure} 

\section{Discussion}
\label{sec:discussion}

The proposed nearest neighbor-Dirichlet mixture provides a useful alternative to 
Bayesian density estimation based on Dirichlet mixtures
with much faster computational speed and stability in avoiding MCMC. MCMC can have very poor performance in mixture models and other multimodal cases, due to difficulty in mixing, and hence can lead to posterior inferences that are unreliable. There is a recent literature attempting to scale up MCMC-based analyses in model-based clustering contexts including for Dirichlet process mixtures; refer, for example to \cite{song2020distributed} and \cite{ni2020consensus}.  However, these approaches are complex to implement and are primarily focused on the problem of clustering, while we are instead focused on flexible modeling of unknown densities.

The main conceptual disadvantage of the proposed approach is the lack of a coherent Bayesian posterior updating rule. However, we have shown that nonetheless the resulting pseudo-posterior can have appealing behavior in terms of frequentist asymptotic properties, finite sample performance, and accuracy in uncertainty quantification. In addition, it is important to keep in mind that Bayesian kernel mixtures have key disadvantages that are difficult to remove within a fully coherent Bayesian modeling framework.  These include a strong sensitivity to the choice of kernel and prior on the weights on these kernels; refer, for example to \cite{miller2019robust}.

There are several important next steps. The first is to develop fast and robust algorithms for using the nearest neighbor-Dirichlet mixture not just for density estimation but also as a component of more complex hierarchical models. For example, one may want to model the residual density in regression nonparametrically or treat a random effects distribution as unknown. In such settings, one can potentially update other parameters within a Bayesian model using Markov chain Monte Carlo, while using algorithms related to those proposed in this article to update the nonparametric part conditionally on these other parameters.  

\section{Acknowledgements}
\texttt{R} package \texttt{NNDM}  available at  
 \href{https://github.com/shounakchattopadhyay/NN-DM}{\texttt{https://github.com/shounakchattopadhyay/NN-DM}} was used for the numerical experiments.
 This research was partially supported by grants R01ES027498 and R01ES028804 of the United States National Institutes of Health and grant N00014-16-1-2147 of the Office of Naval Research.


\newpage

\appendix

\section*{Appendix}




\section{Prerequisites}\label{app:lemma_setup}

We first introduce some notation with accompanying technical details which will be used hereafter. We denote the Frobenius norm and determinant of $A \in \mathbb{R}^{p \times p}$ by $||A||_{F} \, = \{\mbox{tr}(A^{\T}A)\}^{1/2}$ and $|A|$, respectively. For $v \in \mathbb{R}^p$, one has $|| v v^{\T} ||_{F} \, = \, ||v||_{2}^2$ where $||a||_{2}\, = (a^{\T}a)^{1/2}$ is the Euclidean norm of $a$. For two symmetric matrices $A, B \in \mathbb{R}^{p \times p}$, we say that $A \geq B$ if $A - B$ is positive semi-definite, that is $x^{\T}(A-B)x \geq 0$ for all $x \in \mathbb{R}^p, x \neq 0_{p}$ where $0_{p} = (0,\ldots,0)^{\T}$. For a real symmetric matrix $A_{*}$, let the eigenvalues of $A_{*}$ be $e_{1}(A_*),\ldots,e_{p}(A_*)$, arranged such that $e_{1}(A_*)\geq \ldots \geq e_{p}(A_*)$. If $A \geq B$, then it follows by the min-max theorem \citep{teschl2009mathematical} that for each $j=1,\ldots,p$, we have $e_{j}(A) \geq e_{j}(B)$. In particular, we have $| A | \, \geq \, | B |$ and $|| A ||_{F}\, \geq \, || B ||_{F}$.

Now consider a true data generating density $X_1, \ldots, X_n \overset{iid}{\sim} f_0$  satisfying Assumptions \ref{assump:a1}-\ref{assump:a3} as in Section \ref{sec:mean_and_variance}. Let $\mathcal{X}^{(n)} = (X_1,\ldots,X_n)$ and suppose $f_0$ induces the measure $P_{f_0}$ on the Borel $\sigma$-field on $\mathbb{R}^p$, denoted by $\mathcal{B}(\mathbb{R}^p)$. We form the $k$-nearest neighborhood of $X_i$ using the Euclidean norm for $i=1,\ldots,n$. We also let $k$ depend on $n$ and express this dependence as $k_n$ when required. However, we routinely drop this dependence for notational simplicity. For $X_i$, let $Q_i$ be its $k$th nearest neighbor in $\mathcal{X}^{(n)}$ (for $k=1$, $Q_i = X_i$) and let $R_i$ be the distance between $X_i$ and $Q_i$, given by $R_i =  || X_i - Q_i ||_2$. Define the ball 
$$B_i = \{y \in [0,1]^p: 0 <||y - X_i||_2< R_i\}$$ 
and the probability 
$$G(X_i, R_i) = \int_{B_i} f_0(u)\, du$$
of the ball $B_i$ under $P_{f_0}$. Let $Y^{(i)}_1 = X_i$ and  $Y^{(i)}_{2}, \ldots, Y^{(i)}_{k-1}$ denote the rest of the interior points in $B_i$. Let the mean $\bar{X}_i$ and covariance matrix $S_i$ of the $i$th neighborhood be $$\bar{X}_i = \dfrac{1}{k_n}\left\{\sum_{j=1}^{k-1} Y^{(i)}_{j} + Q_i\right\},$$ $$S_{i} = \dfrac{1}{k_n} \left\{\sum_{j = 1}^{k-1} (Y^{(i)}_j - \bar{X}_i)(Y^{(i)}_j - \bar{X}_i)^{\T} + (Q_i - \bar{X}_i)(Q_i - \bar{X}_i)^{\T} \right\}.$$
We observe that $(Y_2^{(i)}, \ldots, Y_{k-1}^{(i)}, Q_i)$ is identically distributed for $i=1,\ldots,n$ since $X_1, \ldots, X_n$ are independent and identically distributed. Thus we only consider the case $i=1$ from here on. For sake of brevity, denote $Y_{u}^{(1)}$ by $Y_u$ for $u=2,\ldots,k-1$ and $Q_1$ by $Q$.  

Conditional on $X_1 = x_1 \in [0,1]^p$ and $R_1 = r_1 > 0$, following \cite{mack1979multivariate}, the conditional joint density of $Y_2, \ldots, Y_{k-1}$ and $Q$ is
\begin{equation*}
p(y_2, \ldots, y_{k-1}, q \mid x_1, r_1) = \left\lbrace\prod_{j=2}^{k-1} \dfrac{f_0(y_j)}{G(x_1, r_1)}\mathbb{I}\left(y_j \in B_1 \right)\right\rbrace \dfrac{f_0(q)}{G^{'}(x_1, r_1)} \mathbb{I}\left(|| q - x_1 || \, = r_1\right),   
\end{equation*} 
where $G^{'}(x_1, r_1) = \partial G(x_1, r_1) / \partial r_1$ and $\mathbb{I}(A)$ denotes the indicator function of the event $A \in \mathcal{B}(\mathbb{R}^p)$. Thus conditional on $X_1$ and $R_1$, the random variables $Y_2, \ldots,  Y_{k-1}$ are independent and identically distributed, and independent of $Q$. 

Let the function $\rho(x_1, r_1) =  r_1^{\kappa_1}$ where $\kappa_1$ is a non-negative integer. This function can be identified with $\phi(\cdot)$ in equation (11) of \cite{mack1979multivariate}. In the results that follow, we will require the expectation of $\rho(x_1, r_1)$ under $P_{f_0}$ for different choices of $\kappa_1$. To that end, we shall repeatedly make use of the equation (12) from \citet{mack1979multivariate} adapted to our setting:
\begin{eqnarray}\label{eq:mack_result}
\lefteqn{ E_{P_{f_0}} \{R_1^{\kappa_1}\mid X_1 = x_1\} = } \nonumber \\
&& \frac{(n-1)!}{(k-2)!(n-k)!} \int_{0}^{1} \left\lbrace \left( \frac{t}{C_p f_0(x_1)}\right)^{\kappa_1/p} + o(t^{\kappa_1/p})\right\rbrace t^{k-2} (1-t)^{n-k} dt. 
\end{eqnarray} 

Finally, we let $\widetilde{E}$ and $ \widetilde{\mbox{var}}$ denote the expectation and variance, respectively, of the NN-DM estimator $f(x)$ under the pseudo-posterior density $\widetilde{\Pi}$, described in \eqref{eq:NNDM-PP}. Conditioning notation under $\widetilde{\Pi}$ is as usual; for example, the conditional expectation $$\widetilde{E}\{f(x) \mid \pi_1, \ldots, \pi_n\} = \sum_{i=1}^{n} \pi_i \widetilde{E}\{\phi_p(x; \eta_i, \Sigma_i)\},$$ where the expectation $\widetilde{E}\{\phi_p(x; \eta_i, \Sigma_i)\}$ is with respect to the pseudo-posterior density of $(\eta_i, \Sigma_i)$ as described in Section \ref{sec:multivariate_gaussian}.

\section{Proof of Theorem \ref{theorem:post_mean}} \label{app:multproof}


Suppose $X_1,\ldots,X_n$ are independent and identically distributed random variables generated from the density $f_0$ supported on $[0,1]^p$ satisfying Assumptions \ref{assump:a1}-\ref{assump:a3}. For $i=1,\ldots,n$, recall the definitions of $\mu_i$ and $\Lambda_i$ from \eqref{eq:multivariate_posterior_mean}: 
$$\mu_{i} = \frac{\nu_0}{\nu_n}\mu_0 + \frac{k}{\nu_n}\bar{X}_i, \quad \Lambda_{i} = \frac{\nu_{n}+1}{\nu_{n}(\gamma_n - p + 1)}\Psi_{i}.$$
We want to show that $\hat{f}_{n}(x) = (1/n) \sum_{i=1}^{n} t_{\gamma_n - p + 1}(x; \mu_i, \Lambda_i) \to f_{0}(x)$ in $P_{f_0}$-probability as $n \to \infty$ for any $x \in [0,1]^{p}$, where $\hat{f}_{n}(x)$ is as described in \eqref{eq:multivariate_posterior_mean}. We first prove two propositions involving successive mean value theorem type approximations to $\hat{f}_{n}(x)$, which will imply the final result. We now state the two propositions, with accompanying proofs, before stating the final theorem.

\begin{proposition}\label{prop:mult1}
Fix $x \in [0,1]^p$. Let $f_{A}(x) = (1/n)\sum_{i=1}^{n}t_{\gamma_n - p + 1}(x;X_i, \Lambda_i)$. Also, let $k = o(n^{i_1})$ with $i_1 = 2/(p^2+p+2)$ and $\nu_0 = o(n^{-1/p} k^{(1/p)+1})$. Then, we have $\E(\,| \hat{f}_{n}(x) - f_{A}(x) | \,) \to 0$ as $n \to \infty$.
\end{proposition}
\begin{proof}
Since the $(\Lambda_{i})_{i=1}^{n}$ are identically distributed and $(\mu_{i})_{i=1}^{n}$ are identically distributed, we have $\E(\,| \hat{f}_{n}(x) - f_{A}(x) | \,) \leq \E\{\, | t_{\gamma_n-p+1}(x;\mu_1,\Lambda_1) - t_{\gamma_n-p+1}(x;X_1,\Lambda_1)| \, \}$. The multivariate mean value theorem now implies that 
\begin{equation}\label{eq:prop1_breakup}
\E(\, | \hat{f}_{n}(x) - f_{A}(x) |\,) \leq \E \left\{\, | \Lambda_{1}|^{-1/2}\, || \nabla t_{\gamma_n - p +1}(\xi; 0_p, I_{p}) ||_{2} \, || \Lambda_{1}^{-1/2}(X_1 - \mu_1) ||_2  \right\},
\end{equation}
where $\nabla t_{\gamma_n - p +1}(\xi; 0_p, I_{p}) = [\partial t_{\gamma_n - p+1}(x; 0_p, \mathbbm{I}_p)/ \partial x]_{\xi}$ for some $\xi$ in the convex hull of $\Lambda_{1}^{-1/2}(x-X_1)$ and $\Lambda_{1}^{-1/2}(x-\mu_1)$. 

Using standard results and the min-max theorem, we have $$|| \Lambda_{1}^{-1/2}(X_1 - \mu_1) ||_{2} \, \leq \, ||\Lambda_{1}^{-1/2}||_{F} \, ||X_{1} - \mu_1 ||_{2}.$$ If we let $H_{n} = H = \{\nu_{n}(\gamma_n-p+1)\}^{-1}(\nu_n+1)\Psi_0 = h^2 I_{p}$ where $h^{2} = h_n^2 = \{\nu_{n}(\gamma_n-p+1)\}^{-1}\{(\nu_n+1)(\gamma_0-p+1)\}\, \delta_{0}^2$ following the choice of $\Psi_0$ from Section \ref{sec:CV}, then it is clear that $\Lambda_{1} \geq H$. Therefore, we have $|| \Lambda_{1}^{-1/2}(X_1 - \mu_1) ||_{2} \, \leq \, || H^{-1/2} ||_{F} \, ||X_{1}-\mu_1 ||_{2}.$ Straightforward calculations show that $||H^{-1/2} \|\|_{F} \, = h^{-1} p^{1/2}$ and $|| X_{1}-\mu_1 ||_{2} \, \leq R_{1} + \{\nu_{n}^{-1}(1 + || \mu_0 ||_{2}\,)\nu_{0}\}$ where $R_{1} = \, ||X_{1} - X_{1[k]} ||_{2}$. Using Theorem $2.4$ from \cite{biau2015lectures} for $p \geq 2$ and  \eqref{eq:mack_result} for $p=1$, one gets
\begin{equation}\label{eq:rsq_bound}
    E_{P_{f_0}}(R_1^2) \leq d_p^2 \left(\frac{k}{n}\right)^{2/p},
\end{equation}
for an appropriate constant $d_p > 0$. Thus, we have $\E(R_{1}) \leq \{\E(R_{1}^2)\}^{1/2} \leq d_{p} (k/n)^{1/p}$ for sufficiently large $n$. This implies that
\begin{equation}\label{eq:mult_meanapprox}
E(||X_1 - \mu_1||_2) \leq d_p \left(\frac{k}{n}\right)^{1/p} + o\left(\frac{k}{n}\right)^{1/p}.
\end{equation}
We also have $| \Lambda_{1} | ^{-1/2} \, \leq \, | H |^{-1/2} = h^{-p}$. Finally, simple calculations yield that $$||\nabla t_{\gamma_n - p +1}(\xi; 0_p, I_{p})||_{2} \leq L_{1,n,p}$$ where $L_{1,n,p} > 0$ satisfies $L_{1,n,p} \to (2\pi)^{-p/2} e^{-1/2}$ as $n \to \infty$. 
Plugging all these back in \eqref{eq:prop1_breakup}, we obtain a finite constant $L_{2,n,p} > 0$ such that
\begin{equation}
\E(\, | \hat{f}_{n}(x) - f_{A}(x) |\,) \leq L_{2,n,p}(n^{-i_1}k)^{(p^2+p+2)/(2p)} + o\{(n^{-i_1}k)^{(p^2+p+2)/(2p)}\},
\end{equation}
which goes to $0$ as $n \to \infty$, completing the proof.
\end{proof}
We now provide the second mean value theorem type approximation which approximates the random bandwidth matrix $\Lambda_{i}$ in $f_{A}(x)$ by $H = H_n$ for each $i=1,\ldots,n$.
\begin{proposition}\label{prop:mult2}
Fix $x \in [0,1]^p$. Let $f_{K}(x) = (1/n)\sum_{i=1}^{n}t_{\gamma_n - p + 1}(x;X_i, H)$. Also, let $k = o(n^{i_2})$ with $i_2 = 4/(p+2)^2$ and $\nu_0 = o\{n^{-2/p} k^{(2/p)+1}\}$. Then, we have $\E(\,|{f}_{A}(x) - f_{K}(x)| \,) \to 0$ as $n \to \infty$.
\end{proposition}
\begin{proof}
Using the identically distributed properties of $(\Lambda_{i})_{i=1}^{n}$ and $(X_{i})_{i=1}^{n}$, we obtain $\E(\,|{f}_{A}(x) - f_{K}(x)| \,) \leq \E(\, | t_{\gamma_n-p+1}(x;X_1,\Lambda_1) - t_{\gamma_n-p+1}(x;X_1,H) | \,)$. Using the multivariate mean value theorem, we obtain that 
\begin{equation}\label{eq:prop2_breakup}
\E\left(\, |t_{\gamma_n-p+1}(x;X_1,\Lambda_1) - t_{\gamma_n-p+1}(x;X_1,H)| \,\right) \leq \E(\, ||M_1||_{F} \, || \Lambda_1 - H ||_{F} \,),
\end{equation}
where $M_1 = [\partial\{t_{\gamma_n-p+1}(x;X_1, \Sigma)\}/\partial \Sigma]_{\Sigma_0}$ for some $\Sigma_0$, with $\Sigma_0$ in the convex hull of $\Lambda_{1}$ and $H$. Since $\Lambda_1 \geq H$, we immediately have $\Sigma_0 \geq H$ as well. Using the definitions of $\Lambda_1$ and $H$, we have $$|| \Lambda_{1} - H ||_{F} \, \leq \frac{(\nu_n + 1)}{\nu_n (\gamma_n - p + 1)} \, \left\{ \left|\left| \sum_{j \in \mathcal{N}_1} (X_j - \bar{X}_1)(X_j - \bar{X}_1)^{\T} \right|\right|_{F} \, + \frac{k\nu_0}{\nu_n} \left|\left| \bar{X}_1 \bar{X}_1^{\T} \right|\right|_{F} \, \right\}.$$ Since $||\sum_{j \in \mathcal{N}_1} (X_j - \bar{X}_1)(X_j - \bar{X}_1)^{\T}||_F \leq \sum_{j \in \mathcal{N}_1} ||(X_j - \bar{X}_1)(X_j - \bar{X}_1)^{\T}||_F = \sum_{j \in \mathcal{N}_1} ||X_j - \bar{X}_1||_{2}^2 \leq \sum_{j \in \mathcal{N}_1} R_1^2 = kR_1^2$, we get for sufficiently large $n$ the following: 
\begin{align}\label{eq:mult_bwapprox}
   \E (\, || \Lambda_{1} - H ||_{F} \,) & \leq \E(R_{1}^2) + o\left(\frac{k}{n}\right)^{2/p},\\
   & \leq d_p^2 \left(\frac{k}{n}\right)^{2/p} + o\left(\frac{k}{n}\right)^{2/p},
   \label{eq:mult_bwapprox_part2}
\end{align}
using \eqref{eq:rsq_bound} and $\nu_0 = o\left\{n^{-2/p}k^{(2/p)+1}\right\}$.
Taking partial derivatives of $\mbox{log}\{t_{\gamma_n-p+1}(x;X_1,\Sigma)\}$ with respect to $\Sigma$ evaluated at $\Sigma_0$ and taking Frobenius norm of both sides, we obtain $$||t_{\gamma_n-p+1}^{-1}(x;X_1,\Sigma_0) \, M_1 ||_{F}\, \leq h^{-2}(\gamma_n+1)$$ for sufficiently large $n$. We now observe that $$t_{\gamma_n-p +1}(x;X_1, \Sigma_0) \leq c_{p, \gamma_n - p + 1} | \Sigma_0|^{-1/2}\, \leq c_{p, \gamma_n - p + 1} | H |^{-1/2}\, = h^{-p} c_{p, \gamma_n - p + 1},$$ where $c_{p, \beta} = (\pi \beta)^{-p/2} \{\Gamma(\beta/2)\}^{-1} \Gamma\{(\beta+p)/2\}$ for $p \geq 1, \beta > 0$. Note that $c_{p, \beta} \to (2\pi)^{-p/2}$ as $\beta \to \infty$ for any $p \geq 1$. This immediately implies that $|| M_1 ||_{F} \, \leq h^{-(p+2)} c_{p, \gamma_n-p+1}(\gamma_n+1)$ for sufficiently large $n$. Plugging all these back in equation \eqref{eq:prop2_breakup}, we obtain for sufficiently large $n$, a finite $L_{3,n,p}>0$ such that
\begin{equation}
\E(\, |{f}_{A}(x) - f_{K}(x)| \,) \leq L_{3,n,p}(n^{-i_2}k)^{(p+2)^2/(2p)} + o\{(n^{-i_2}k)^{(p+2)^2/(2p)}\},
\end{equation}
which goes to $0$ as $n \to \infty$, proving the proposition.
\end{proof}

We now prove Theorem \ref{theorem:post_mean}.
\\
\begin{proof}[Theorem 4]
$\E(\, |\hat{f}_{n}(x) - f_{K}(x)| \,) \leq \E(\, | \hat{f}_{n}(x) - f_{A}(x) | \,) + \E(\, | f_{A}(x) - f_{K}(x) | \,)$ by the triangle inequality. Using Propositions \ref{prop:mult1} and \ref{prop:mult2}, we obtain that $E_{P_{f_0}}(\, |\hat{f}_n(x) - f_K(x)| \,) \to 0$ as $n \to \infty$. From Section \ref{app:kdeconsistency} of the Appendix, we obtain $f_{K}(x) \to f_{0}(x)$ in $P_{f_0}$-probability. This immediately implies that given the conditions on $k, \nu_0,$ and for any $x \in [0,1]^p$, we have $\hat{f}_n(x) \to f_0(x)$ in $P_{f_0}$-probability.
\end{proof}
\section{Proof of Theorem \ref{theorem:post_var}}
\label{app:proof_post_var}
\begin{proof}
Fix $x \in [0,1]^p$. For $i=1,\ldots,n$, let $z_{i} = \phi_p(x \, ; \eta_{i}, \Sigma_{i})$ and suppose $z^{(n)} = (z_{1},\ldots,z_{n})^{\T}$. Then, we have $f(x) = \sum_{i=1}^{n} \pi_{i} z_{i} = z^{(n)\T} \pi^{(n)}$ where $\pi^{(n)} = (\pi_{1},\ldots,\pi_{n})^{\T}$. We begin with the identity
\begin{equation}\label{eq:post_var_breakdown}
\widetilde{\mbox{var}}\{f(x) \} = \widetilde{\mbox{var}}[\widetilde{E}\{f(x) \mid z^{(n)}\} ] + \widetilde{E}[\widetilde{\mbox{var}}\{f(x) \mid z^{(n)}\} ].
\end{equation}

We start with the first term on the right hand side of \eqref{eq:post_var_breakdown}. Observe that $z_{1},\ldots,z_{n}$ are independent under $\widetilde{\Pi}$ and $\widetilde{E}(\pi_i) = 1/n$ for $i=1,\ldots,n$. Thus, we have
\begin{align*}
\widetilde{\mbox{var}}[\widetilde{E}\{f(x) \mid z^{(n)}\} ] 
& = \widetilde{\mbox{var}}\left(\frac{1}{n} \sum_{i=1}^{n} z_{i} \right)\\
& = \frac{1}{n^2} \sum_{i=1}^{n} \widetilde{\mbox{var}}(z_{i} )\\
& \leq \frac{1}{n^2} \sum_{i=1}^{n} \widetilde{E}(z_{i}^2 )\\
& = \frac{1}{n^2} \sum_{i=1}^{n} R_n |B_i|^{-1/2} t_{\gamma_n - p + 2}(x ; \mu_i, B_i), 
\end{align*} 
since for $i=1,\ldots,n$, we have
\begin{equation}\label{eq:z_sq}
\widetilde{E}(z_{i}^2 ) = R_n |B_i|^{-1/2} t_{\gamma_n - p + 2}(x ; \mu_i, B_i),
\end{equation}
where $$R_{n} = \dfrac{\Gamma\{(\gamma_n - p +2)/2 \}}{\Gamma\{(\gamma_n - p +1)/2 \}} \left[\dfrac{\nu_n+2}{4 \pi \nu_n (\gamma_n - p + 2)}\right]^{p/2}, \quad B_i = D_n \Lambda_i,$$ and $D_{n} = \{2(\gamma_n - p + 2) (\nu_n + 1)\}^{-1} (\gamma_n - p + 1) (\nu_n + 2)$. To obtain \eqref{eq:z_sq}, we integrate over the pseudo-posterior distribution of $(\eta_{i}, \Sigma_i)_{i=1}^{n}$, namely $\mbox{NIW}(\mu_{i}, \nu_n, \gamma_n, \Psi_i)$. 
For $i=1,\ldots,n$, since $|\Lambda_{i}| \geq |H_n|$, we have $|B_i| \geq D_n^p |H_n|$. Letting $\widehat{f}_{var}(x) = (1/n) \sum_{i=1}^{n} t_{\gamma_n - p + 2}(x; \mu_i, B_i)$, we have
\begin{equation}\label{eq:post_var_part1}
\widetilde{\mbox{var}}[\widetilde{E}\{f(x) \mid z^{(n)}\} ] \leq \frac{R_n D_n^{-p/2} \widehat{f}_{var}(x)}{n|H_n|^{1/2}}.
\end{equation}

We now analyze the second term on the right hand side of \eqref{eq:post_var_breakdown}. Recall that $\pi^{(n)}$ is independent of $z^{(n)}$ under $\widetilde{\Pi}$. Let $\Sigma_{\pi}$ denote the pseudo-posterior covariance matrix of $\pi^{(n)}$. Standard results yield $\Sigma_{\pi} = V_{n} \{(1-C_{n})\mathbbm{I}_{n} + C_{n} \mathbbm{1}_n \mathbbm{1}_n^{\T}\}$, where $V_{n} = (n-1)/[n^2 \{n(\alpha+1) + 1\}],$ and $C_{n} = -1/(n-1)$. Then, we have
\begin{equation}\label{eq:post_var_part2_lemma1}
\widetilde{E}[\widetilde{\mbox{var}}\{f(x) \mid z^{(n)}\} ] = \widetilde{E}[z^{(n)\T}\, \Sigma_{\pi}\, z^{(n)} ].
\end{equation}
Using the expression for $\Sigma_{\pi}$ along with \eqref{eq:post_var_part2_lemma1}, we obtain,
\begin{equation}\label{eq:post_var_part2_lemma2}
\widetilde{E}[\widetilde{\mbox{var}}\{f(x) \mid z^{(n)}\} ] = \frac{1}{n(\alpha+1)+1} \widetilde{E}\left\{\frac{1}{n} \sum_{i=1}^{n} (z_i - \bar{z})^2 \right\},
\end{equation}
where $\bar{z} = (1/n) \sum_{i=1}^{n} z_{i}$. We now have
\begin{align*}
\widetilde{E}[\widetilde{\mbox{var}}\{f(x) \mid z^{(n)}\} ] 
& = \frac{1}{n\{n(\alpha+1) + 1\}} \left\{\sum_{i=1}^{n} \widetilde{E}(z_{i}^2) - n \widetilde{E}(\bar{z}^2 ) \right\}\\
& \leq \frac{1}{n\{n(\alpha+1) + 1\}} \sum_{i=1}^{n} \widetilde{E}(z_{i}^2 )\\
& = \frac{1}{n\{n(\alpha+1) + 1\}} \sum_{i=1}^{n} R_n |B_i|^{-1/2} t_{\gamma_n - p + 2}(x ; \mu_i, B_i),
\end{align*} 
using \eqref{eq:z_sq}. Using $|B_i| \geq D_n^p |H|$ for $i=1,\ldots,n$ as before, we have
\begin{equation}\label{eq:post_var_part2}
\widetilde{E}[\widetilde{\mbox{var}}\{f(x) \mid z^{(n)}\} ] 
\leq \frac{R_n D_n^{-p/2} \widehat{f}_{var}(x)}{\{n(\alpha+1)+1\}|H_n|^{1/2}}.
\end{equation}
Combining \eqref{eq:post_var_part1} and \eqref{eq:post_var_part2} and putting the results back in  \eqref{eq:post_var_breakdown}, we have the desired result. If we let $n \to \infty$, we immediately obtain that $\widetilde{\mbox{var}}\{f(x) \} \to 0$ in $P_{f_0}$-probability.
\end{proof}

\section{Proof of Theorem \ref{theorem:asymptotic-normality}}
\label{app:surrogate-AN}
\begin{proof}
We have iid data $\mathcal{X}^{(n)} = (X_1, \ldots, X_n)$ such that $X_1, \ldots, X_n \overset{iid}{\sim} f_0$, with $f_0$ satisfying Assumptions \ref{assump:a1}-\ref{assump:a3} for $p=1$. Given the NN-DM estimator $f(x) = \sum_{i=1}^{n} \pi_i \phi(x; \eta_i, \sigma_i^2)$, we define the simplified NN-DM density estimator to be $$g(x) = \frac{1}{n}\sum_{i=1}^n \phi(x; \eta_i, \sigma^2_i),$$ 
The simplified estimator $g(x)$ can be interpreted as a version of $f(x)$ with the Dirichlet weights being replaced by their pseudo-posterior mean. That is, $g(x) = \widetilde{E}\{f(x) \mid (\eta_1, \sigma_1^2), \ldots, (\eta_n, \sigma_n^2)\}$. The pseudo-posterior distribution of $g(x)$ is induced through the pseudo-posterior distributions of $\{(\eta_i, \sigma_i^2)\}_{i=1}^{n}$. The pseudo-posterior mean is of the form $$\hat{f}_n(x) = \dfrac{1}{n} \sum_{i=1}^n \dfrac{1}{\lambda_i} t_{\gamma_n}\left(\dfrac{x - \mu_i}{\lambda_i}\right),$$ where $\lambda_{i} = \{(\nu_n+1)/\nu_n\}^{1/2} \delta_i$. Let $h_n = (\nu_n\gamma_n)^{-1/2}(\nu_n + 1)^{1/2}(\gamma_0\delta_0^2)^{1/2}$. Then
\begin{equation}
\label{eq:slutsky-PPM}
    {(nh_n)}^{1/2} \, E_{P_{f_0}}|\hat{f}_n(x) - f_K(x)| \to 0,
\end{equation}
for $k_n = o(n^{2/7})$ and $k_n \to \infty$ as $n \to \infty$ from Section \ref{app:multproof} of the Appendix, where $$f_K(x) = \frac{1}{nh_n} \sum_{i=1}^{n} t_{\gamma_n}\left(\frac{x - X_i}{h_n}\right).$$
We want to investigate the asymptotic distribution of $f(x)$ as $n \to \infty$. For that, we first investigate the asymptotic distribution of the simplified NN-DM estimator $g(x)$, and then show that $f(x)$ and $g(x)$ are asymptotically close in $P_{f_0}$-probability.

To derive the asymptotic distribution of $g(x)$, we begin with the asymptotic distribution of $f_{K}(x)$, which can be expressed as $f_{K}(x) = n^{-1}\sum_{i=1}^n u_{in}$, where $u_{in} = h_n^{-1}t_{\gamma_n}\{(x - X_i)/h_n\}$. Using Lyapunov's central limit theorem and denoting convergence in distribution under $f_0$ by $d_0$, we have $$\frac{f_{K}(x) - E_{P_{f_0}}\{f_{K}(x)\}}{[\text{var}_{P_{f_0}}\{f_K(x)\}]^{1/2}} \overset{d_0}{\to} \Gauss(0,1)$$ if
\begin{equation}
    \dfrac{(\sum_{i=1}^n \rho_{in})^{1/r}}{(\sum_{i=1}^n \tau^2_{in})^{1/2}} \to 0, \text{ as } n \to \infty,
\end{equation}
for some $r>2$, where $\rho_{in} = E|u_{in} - E(u_{in})|^r$ and $\tau^2_{in} = E\{u_{in} - E(u_{in})\}^2$ for $i=1,\ldots,n$. By standard calculations, we have $$\tau^2_{in} = \frac{f_0(x)}{h_n} \int t_{\gamma_n}^2(u)du + o\left(\frac{1}{h_n}\right).$$ For $r=3$, $$\rho_{in} \leq \frac{8 f_0(x)}{h_n^2} \int t_{\gamma_n}^3(u) du + o\left(\frac{1}{h_n^2}\right).$$ It is straightforward to see that $\int t_{\gamma_n}^r(u)du / \int t_{\gamma_n}(u) du = \mathcal{O}(1)$ for any $r\geq 1$. So,  Lyapunov's condition is satisfied as the ratio in this case satisfies $\mathcal{O}\{(nh_n)^{-1/6}\}$ and $nh_n \to \infty$. Additionally, $|\tau^2_{in} - \{f_0(x)/h_n\} \int \phi^2(u) \, du| \to 0$. So by a combination of Lyapunov's central limit theorem and Slutsky's theorem, we have 
\begin{equation}
    (nh_n)^{1/2}\left[f_K(x) - E_{P_{f_0}}\{f_K(x)\}\right] \overset{d_0}{\to} \Gauss\left(0, \frac{f_0(x)}{2\pi^{1/2}}\right),
\end{equation}
since $\int \phi^2(u) \, du = (2 \pi^{1/2})^{-1}$. From the calculations in Section \ref{app:kdeconsistency} of the Appendix, we can expand the Taylor series to two more terms to obtain $$E_{P_{f_0}}\left\{f_K(x) - f_0(x) - \frac{h_n^2 f_0^{''}(x)}{2}\right\} = \mathcal{O}(h_n^{4}),$$ since $|f_0^{(4)}(x)| \leq C_0$ for all $x \in [0,1]$. Thus,
\begin{equation}
    \label{eq:slutsky-kde}
    (nh_n)^{1/2}\left[f_K(x) - \left\{ f_0(x) + \frac{h_n^2 f_0^{(2)}(x)}{2}\right\}\right] \overset{d_0}{\to} \Gauss\left(0, \frac{f_0(x)}{2\pi^{1/2}}\right),
\end{equation}
provided $n^{-2/9}k_n \to \infty$ as $n \to \infty$, implying $(nh_n)^{1/2} h_n^4 \to 0$.

We now argue that $(nh_n)^{1/2}|g(x) - \hat{f}_n(x)| \to 0$ in $P_{f_0}$-probability. For this, we first look at 
\begin{align*}
    E_{P_{f_0}}\left[nh_n \{g(x) - \hat{f}_n(x)\}^2\right] & = nh_n \, E_{P_{f_0}}\left[\widetilde{E}\left\{(g(x) - \hat{f}_n(x))^2 \right\} \right]\\
    & = nh_n \, E_{P_{f_0}}\left[\widetilde{\text{var}}\{g(x) \}\right],
\end{align*}
since $\widetilde{E}\{g(x) \} = \hat{f}_n(x)$. The pseudo-posterior variance of $g(x)$ is given by $$\widetilde{\text{var}}\{g(x) \} = \frac{1}{n^2} \sum_{i=1}^{n} \widetilde{\text{var}}\{Z_i(x) \},$$
where $Z_i(x) = \phi(x; \eta_i, \sigma_i^2)$ for $i=1,\ldots,n$. It is straightforward to show that
\begin{equation}
\label{eq:varZ}
    \widetilde{\text{var}}\{Z_i(x) \} \sim |\Delta_n| \frac{1}{\widetilde{\lambda}_i^2} \, t_{2\gamma_n + 1} \left(\frac{x - \mu_i}{\widetilde{\lambda}_i}\right),
\end{equation}
as $n \to \infty$, where $\widetilde{\lambda}_i^2 = \lambda_i^2 / 2$ for $i=1,\ldots,n$ and $$\Delta_n = \frac{u_{\gamma_n}^2}{u_{2\gamma_n+1}} - \frac{1}{(2\pi)^{1/2}},$$ with $u_d = \Gamma\{(d+1)/2\}/\{(d\pi)^{1/2} \Gamma(d/2)\}$ being the normalizing constant of the Student's t-density with degrees of freedom $d > 0$. Using Stirling's approximation, $\Delta_n \to 0$ as $n \to \infty$. This immediately implies
$$nh_n \, \widetilde{\text{var}}\{g(x) \} \leq |\Delta_n| \, v_{g}(x),$$
where 
$$v_g(x) = \frac{1}{n} \sum_{i=1}^{n} \frac{1}{\widetilde{\lambda}_i} t_{2\gamma_n+1}\left(\frac{x - \mu_i}{\widetilde{\lambda}_i}\right).$$
Using the techniques of Section \ref{app:multproof} of the Appendix, it can be shown that $E_{P_{f_0}}\{v_g(x)\} \to f_0(x)$ as $n \to \infty$. Therefore, we have
\begin{align*}
    E_{P_{f_0}} \left[nh_n \{g(x) - \hat{f}_n(x)\}^2\right] & = nh_n \, E_{P_{f_0}} \left[\widetilde{E}\left\{(g(x) - \hat{f}_n(x))^2 \right\}\right]\\
    & = nh_n \,  E_{P_{f_0}}\left[\widetilde{\text{var}}\left\{g(x) \right\}\right]\\
    & \leq E_{P_{f_0}}\{|\Delta_n| v_g(x)\}\\
    & \to 0,
\end{align*}
as $n \to \infty$. A simple application of Chebychev's inequality implies  $(nh_n)^{1/2} |g(x) - \hat{f}_n(x)| \to 0$ in $P_{f_0}$-probability as $n \to \infty$. Combining this with \eqref{eq:slutsky-PPM} and \eqref{eq:slutsky-kde} and using Slutsky's theorem, we obtain the desired result for $g(x)$.

We now demonstrate that $f(x)$ and $g(x)$ are asymptotically close to derive the same result for $f(x)$. 
We start out with
\begin{equation}
\label{eq:fulluq-1}
    \begin{aligned}
\mbox{var}_{P_{f_0}}\left[(nh_n)^{1/2} \{f(x) - g(x)\} \right] & = nh_n \, E_{P_{f_0}} \left[\widetilde{\mbox{var}} \left\{f(x) - g(x) \right\}\right]\\
    & = nh_n \, E_{P_{f_0}} \left[\widetilde{\mbox{var}}\left\{\sum_{i=1}^{n} \left(\pi_i - \dfrac{1}{n}\right) Z_i(x) \right\} \right]
    \end{aligned}
\end{equation}
We now focus on the term inside $E_{P_{f_0}}$ in \eqref{eq:fulluq-1} and further decompose it as
\begin{equation}
\label{eq:fulluq-2}
\begin{aligned}
    \widetilde{\mbox{var}}\left\{\sum_{i=1}^{n} \left(\pi_i - \dfrac{1}{n}\right) Z_i(x) \right\} & = \widetilde{E}\left[\widetilde{\mbox{var}}\left\{\sum_{i=1}^{n} \left(\pi_i - \dfrac{1}{n}\right) Z_i(x) \mid \pi^{(n)} \right\} \right]\\
    & + \widetilde{\mbox{var}}\left[\widetilde{\mbox{E}}\left\{\sum_{i=1}^{n} \left(\pi_i - \dfrac{1}{n}\right) Z_i(x) \mid \pi^{(n)}\right\} \right],
    \end{aligned}
\end{equation}
where $\pi^{(n)} = (\pi_1, \ldots, \pi_n)^{\T}$. In \eqref{eq:fulluq-2},
let $\Xi_{1n}$ and $\Xi_{2n}$ be as follows:
$$\Xi_{1n} = \widetilde{E}\left[\widetilde{\mbox{var}}\left\{\sum_{i=1}^{n} \left(\pi_i - \dfrac{1}{n}\right) Z_i(x) \mid \pi^{(n)} \right\} \right],$$
$$\Xi_{2n} = \widetilde{\mbox{var}}\left[\widetilde{\mbox{E}}\left\{\sum_{i=1}^{n} \left(\pi_i - \dfrac{1}{n}\right) Z_i(x) \mid \pi^{(n)}\right\} \right].$$
Thus, we can write \eqref{eq:fulluq-1} as
\begin{equation}
    \label{eq:fulluq-breakup}
\mbox{var}_{P_{f_0}}\left[(nh_n)^{1/2} \{f(x) - g(x)\} \right] = nh_n E_{P_{f_0}}(\Xi_{1n}) + nh_n E_{P_{f_0}}(\Xi_{2n}).
\end{equation}
It is straightforward to see that $\Xi_{1n} = \sum_{i=1}^{n} \widetilde{\mbox{var}}(\pi_i ) \, \widetilde{\mbox{var}}\{Z_i(x)\}$. As $n \to \infty$, we use the fact that $\pi_i \sim \mbox{Beta}(\alpha_n + 1, (n-1)(\alpha_n + 1))$ under $\widetilde{\Pi}$ and \eqref{eq:varZ}, to get 
\begin{equation}
    \label{eq:fulluq-breakup-part1}
    nh_n E_{P_{f_0}}(\Xi_{1n}) \sim \dfrac{n \, \Delta_n}{n(\alpha_n+1)+1} \widetilde{\Xi}_{1n},
\end{equation}
for some $\widetilde{\Xi}_{1n}$ satisfying $\widetilde{\Xi}_{1n} \to f_0(x)$ and $\Delta_n \to 0$ as $n \to \infty$. Therefore, $nh_n E_{P_{f_0}}(\Xi_{1n}) \to 0$ as $n \to \infty$. For the second part, let $d_i(x) = (1/\lambda_i) \, t_{\gamma_n}\{(x-\mu_i) / \lambda_i\}$ so that the pseudo-posterior mean  $\hat{f}_n(x) = (1/n) \sum_{i=1}^{n} d_i(x)$. We first observe that
\begin{align*}
    \Xi_{2n} & = \widetilde{\mbox{var}}\left[\sum_{i=1}^{n} \left(\pi_i - \dfrac{1}{n}\right) \dfrac{1}{\lambda_i} t_{\gamma_n}\left(\dfrac{x-\mu_i}{\lambda_i}\right) \right]\\
    & = \widetilde{\mbox{var}}\left[\sum_{i=1}^{n} \pi_i  d_i(x) \right]\\
    & = \dfrac{1}{n(\alpha_n+1) + 1} \left[\dfrac{1}{n}\sum_{i=1}^{n}\{d_i(x) - \hat{f}_n(x)\}^2\right]\\
    & \leq \dfrac{1}{n(\alpha_n+1) + 1} \left[\dfrac{1}{n} \sum_{i=1}^{n} d_i^2(x)\right].
\end{align*}
It now follows from some algebra that
$$\dfrac{1}{n} \sum_{i=1}^{n} d_i^2(x) \leq \dfrac{d_0(n) \gamma_n}{2\gamma_n + 1} \left(\dfrac{2\gamma_n+1}{\gamma_n}\right)^{1/2} \dfrac{1}{h_n} \left[\dfrac{1}{n} \sum_{i=1}^{n} \dfrac{1}{\tilde{\lambda}_i}t_{2\gamma_n+1}\left(\dfrac{x-\mu_i}{\tilde{\lambda}_i}\right)\right],$$
where $d_0(n) \to (2\pi)^{-1/2}$ as $n \to \infty$. Therefore, we have
\begin{equation}
    \label{eq:fulluq-breakup-part2}
    nh_n E_{P_{f_0}}(\Xi_{2n}) \leq \mathcal{O}\left\{ \dfrac{1}{2 \pi^{1/2}} \dfrac{n}{n(\alpha_n+1)+1} \widetilde{\Xi}_{2n}\right\},
\end{equation}
for some $\widetilde{\Xi}_{2n}$ satisfying $\widetilde{\Xi}_{2n} \to f_0(x)$ as $n \to \infty$. By the conditions of the theorem, we have $nh_n E_{P_{f_0}}(\Xi_{2n}) \to 0$ as $n \to \infty$. This, along with \eqref{eq:fulluq-breakup-part1} substituted in \eqref{eq:fulluq-breakup} provides $(nh_n)^{1/2} \, |f(x) - g(x)| \to 0$ in $P_{f_0}$-probability. This implies the desired result for $f(x)$ using Slutsky's theorem. As a result, we can interpret pseudo-credible intervals to be  frequentist confidence intervals, on average, asymptotically.
\end{proof}
\section{Proof of Theorem \ref{theorem:alpha-plus-one}}
\label{app:alpha-plus-one-proof}
\subsection{A property of the $k_n$-nearest neighbor distance}
Suppose $X_1, \ldots, X_n \overset{iid}{\sim} f_0$ with $f_0$ a density on $\mathbb{R}^p$ satisfying Assumptions \ref{assump:a1}-\ref{assump:a3}. We denote the induced probability measure $P_{f_0}$ by $P_0$ in this section for the sake of convenience. We define the smoothed $k$-nearest neighborhood of $X_i$ as $\mathcal{B}_i = \{y \in \mathbb{R}^p : || X_i - y ||_2 \leq R_i \}$, where $R_i = || X_i - X_{i[k_n]} ||_2$ is the Euclidean distance between $X_i$ and the $k_n$-nearest neighbor of $X_i$ for $i = 1, \ldots, n$. By symmetry, $R_1, \ldots, R_n$ are identically distributed. Suppose $r_n = (k_n/n)^{1/p}$ and define the quasi-neighborhood $\tilde{\mathcal{B}}_i(r) = \{y \in \mathbb{R}^p : || X_i - y ||_2 \leq r \}$, where the random variables $R_i$ have been replaced by $r \geq 0$. Let $$\omega_x(r) = \underset{\{y \, : \, ||y - x||_2 \, \leq \, r\}}{\int} f_0(y) \, dy.$$ The positive density condition on $f_0$ obtained from Assumptions \ref{assump:a1} and \ref{assump:a3} \citep{evans2002asymptotic, evans2008law} ensures the existence of $A > 1$ and $\rho > 0$ such that for all $0 \leq r \leq \rho$ and for all $x \in [0,1]^p$,
\begin{equation}
    \label{eq:positive-density-condition}
    \dfrac{r^p}{A} \leq \omega_x(r) \leq A r^p.
\end{equation}
We first state a Lemma proving some important properties of $R_1$. We next use this Lemma to prove Theorem \ref{theorem:alpha-plus-one}. Recall that two non-negative sequences $(a_n)$ and $(b_n)$ are said to be asymptotically equivalent if $|a_n / b_n| \to c_0$ for some $c_0 > 0$, denoted by $a_n \sim b_n$.

\begin{lemma}
\label{lemma:radius-asymptotic}
Define $i_0 = \{2/(p^2 + p + 2)\} \wedge \{4/(p+2)^2\}$ as in Theorem \ref{theorem:post_mean}. Assume $k_n \sim n^{i_0 - \epsilon}$ for some $\epsilon \in (0, i_0)$. Suppose $\delta > 0$ satisfies $$\delta < \left(1 - i_0 + \epsilon\right)^{-1} - 1.$$ Define $$r_n = \left(\frac{k_n}{n}\right)^{1/p}, \quad c_n = \frac{1}{(Ae)^{1/p}} \, r_n^{1+\delta}, \quad \textrm{{\em and}} \quad n_0 = \left\lfloor \left\{\frac{1 - \frac{i_0 - \epsilon}{2}}{\delta(1 - i_0 + \epsilon)} + 1\right\}^{\frac{1}{i_0 - \epsilon}} \right\rfloor + 1.$$
Then, the following results hold:
\begin{enumerate}
    \item[(i)] $p_n = P_0(R_1 \leq c_n) = \mathcal{O}\left[k_n^{-1/2}\left(\dfrac{k_n}{n}\right)^{(k_n-1)\delta}\right]$ as $n \to \infty$. 
    \item[(ii)] $\displaystyle \sum_{n=n_0+1}^{\infty} p_n < \infty$. That is, $\{p_n\}_{n=1}^{\infty}$ is summable.
    \item[(iii)] $P\left[ \underset{n \to \infty}{ \mbox{\em limsup}} \, \{R_1 \leq c_n\} \right] = 0$.
    \item[(iv)] $n c_n^p \to \infty$ as $n \to \infty$.
\end{enumerate}
\end{lemma}
\begin{proof}
\begin{enumerate}
    \item[(i)] Note that $c_n \leq \rho$ for sufficiently large $n$. From Lemma 4.1 of \cite{evans2002asymptotic} we have 
    \begin{align*}
        P_0(R_1 \leq c_n \mid X_i = x) & = \overset{\omega_x(c_n)}{\underset{0}{\int}} (k_n-1) \binom{n-1}{k_n-1} y^{k_n-2} (1 - y)^{n-k_n} dy\\
        & \leq \binom{n-1}{k_n-1} \omega_x(c_n)^{k_n-1}\\
        & \sim k_n^{-1/2}\left(\frac{k_n}{n}\right)^{(k_n-1)\delta},
    \end{align*}
    for any $x \in [0,1]^p$, using \eqref{eq:positive-density-condition}. This immediately implies $$P_0(R_1 \leq c_n) = \int_{[0,1]^p} P_0(R_i \leq c_n \mid X_i = x) \, f_0(x) \, dx \leq \mathcal{O}\left[k_n^{-1/2}\left(\frac{k_n}{n}\right)^{(k_n-1)\delta}\right].$$
    \item[(ii)] For $n > n_0$, we have $p_n = \mathcal{O}\{n^{-(1+\Theta_n)}\}$ for a sequence $\Theta_n \to \infty$, $\Theta_n > 0$. This ensures that $\sum_{n=n_0+1}^{\infty} p_n < \infty$.
    \item[(iii)] Since $\sum_{n=1}^{\infty} p_n < \infty$, a direct application of the first Borel-Cantelli lemma proves the statement.
    \item[(iv)] We have, using the condition on $\delta$, \begin{align*}
        n c_n^p & = \frac{n}{Ae} \left(\frac{k_n}{n}\right)^{1+\delta}\\
        & = \frac{1}{Ae} \frac{k_n^{1+\delta}}{n^{\delta}}\\
        & \to \infty, \quad \text{as $n \to \infty$.}
    \end{align*}
\end{enumerate}
\end{proof}
We now use the above Lemma to prove Theorem \ref{theorem:alpha-plus-one}. The key idea is to leverage the fact that  $R_i > c_n$ for all $i=1,\ldots,n$ with probability $1$ for all but finite $n$.

\subsection{Number of effective member points in each neighborhood}
We now prove Theorem \ref{theorem:alpha-plus-one}.
\begin{proof}
Using (iii) from Lemma \ref{lemma:radius-asymptotic}, for $i=1,\ldots, n$, we have an integer $\widetilde{N}_i$ such that for all $n \geq \widetilde{N}_i$, $P_0(R_i > c_n) = 1$. However, since $R_1, \ldots, R_n$ are identically distributed, $\widetilde{N}_1 = \ldots = \widetilde{N}_n = \widetilde{N}$, say. Thus, for all $i=1,\ldots,n$, we have $P_0(R_i > c_n) = 1$ for all $n \geq \widetilde{N}$. This immediately implies that $P_0\left[\bigcap_{i=1}^{n}\{R_i > c_n\}\right] = 1 - P_0\left[\bigcup_{i=1}^{n}\{R_i \leq c_n\}\right] \geq 1 - \sum_{i=1}^{n} P_0[R_i \leq c_n] = 1$ for all $n \geq \widetilde{N}$, which shows $P_0\left[\bigcap_{i=1}^{n}\{R_i > c_n\}\right] = 1$. Therefore, we have 
\begin{align*}
    n Q_n & = n P_0\left[X_2 \in \mathcal{B}_1, \bigcap_{i=3}^{n} \{X_2 \notin \mathcal{B}_i\}, \bigcap_{i=1}^{n}\{R_i > c_n\}\right]\\
    & \leq n P_0\left[ \bigcap_{i=3}^{n} \{X_2 \notin \tilde{\mathcal{B}}_i(c_n)\}\right]\\
    & = n P_0\left[ \bigcap_{i=3}^{n} \{|| X_2 - X_i || > c_n\}\right]\\
    & = n \int \, \theta_n^{n-2}(x) \, f_0(x) \, dx,
\end{align*}
where $\theta_n(x) = 1 - \omega_x(c_n)$, since $X_1, \ldots, X_n \overset{iid}{\sim} f_0$. Using \eqref{eq:positive-density-condition}, we have $\theta_n(x) \leq 1 - (c_n^p/A)$ for all $x \in [0,1]^p$. Given the conditions on $k$, it follows that as $n \to \infty$, $$\log n + n \log\left(1 - \frac{c_n^p}{A}\right) \sim  \log n - \frac{n^{\xi}}{A^2 e} \to -\infty,$$ for all $x \in [0,1]^p$, where $\xi = 1 - (1 + \epsilon - i_0)(1 + \delta) > 0$. Therefore, we have
\begin{align*}
    nQ_n & = n \int  \theta_n^{n-2}(x) f_0(x) \, dx\\
    & \leq n  \left[1 - \frac{c_n^p}{A}\right]^{n-2} \int f_0(x) \, dx\\
    & = \mathcal{O}\left[\exp\left\{\log n + n \log\left(1 - \frac{c_n^p}{A}\right)\right\}\right]\\
    & \to 0,
\end{align*}
as $n \to \infty$. This proves the result.\end{proof}

\section{Proof of Consistency of $f_K(x)$}\label{app:kdeconsistency}
Define the standard multivariate t-density with $d > 0$ degrees of freedom to be $g_d(x) = t_{d}(x; 0_p, \mathbbm{I}_p)$. Since $H = H_n = h_n^2  \,\mathbb{I}_p$ as defined in Section \ref{sec:mean_and_variance} is diagonal, it immediately follows that  $t_{\gamma_n-p+1}(x;\mu, H) = h_n^{-p} \,g_{\gamma_n-p+1}\{h_n^{-1} (x-\mu)\}$. The following lemma proves the consistency of any such generic kernel density estimator with t kernel depending on $n$, say $$f_K(x) = \frac{1}{nw^p} \sum_{i=1}^{n} g_{\gamma_n-p+1}\left(\frac{x-X_{i}}{w}\right),$$ where the bandwidth $w = w_n$ satisfies $w_n \to 0$ and $nw_n^p \to \infty$ as $n \to \infty$, with independent and identically distributed data $X_1, \ldots, X_n \sim f_0$ satisfying Assumptions \ref{assump:a1}-\ref{assump:a3}.
\begin{lemma}
Suppose $w_{n}$ is a sequence satisfying $w_n \xrightarrow{} 0$ and $nw_n^p \xrightarrow{} \infty$ as $n \xrightarrow{} \infty$. Let $f_{K}(x) = (nw_n^p)^{-1} \sum_{i=1}^{n} g_{\gamma_n-p+1}\{w_n^{-1}(x-X_{i})\}$. Then $f_{K}(x) \to f_{0}(x)$ in $P_{f_0}$-probability for each $x \in [0,1]^p$.
\label{lemma:t_kde}
\end{lemma}
\begin{proof}
It is enough to show that $\E \{f_{K}(x)\} \xrightarrow{} f_{0}(x)$ and $\mbox{var}_{P_{f_0}}\{f_{K}(x)\} \xrightarrow{} 0$ as $n \xrightarrow{} \infty$. Let us start first with $\E\{f_{K}(x)\}$. We have
\begin{align*}
    \E\{ f_{K}(x)\} & = \E\left\{\frac{1}{w_n^p} g_{\gamma_{n}-p+1} \left(\frac{x-X_{1}}{w_n}\right)\right\}\\
    & = \int_{[0,1]^p} \dfrac{1}{w_n^p} g_{\gamma_{n}-p+1} \left(\frac{y-x}{w_n}\right) f_{0}(y) \, dy\\
    & = \int_{\left[-\frac{x}{w_n}, \frac{1-x}{w_n}\right]^p} g_{\gamma_n-p+1} (u)\, f_{0}(x + w_n u) \, du ,\\
    & = \int_{\left[-\frac{x}{w_n}, \frac{1-x}{w_n}\right]^p} g_{\gamma_n-p+1} (u)\, \{f_{0}(x) + w_n u^{\T} \nabla f_{0}(\xi)\} \, du \\
    & = f_{0}(x) \int_{\left[-\frac{x}{w_n}, \frac{1-x}{w_n}\right]^p} g_{\gamma_n-p+1}(u) \, du + w_n \int_{\left[-\frac{x}{w_n}, \frac{1-x}{w_n}\right]^p} g_{\gamma_n-p+1}(u) u^{\T} \nabla f_{0}(\xi) \, du,\\
    & = f_0(x)\{1-o_n(1)\} + w_n  \mathcal{O}_n(1),
\end{align*}
using the mean value theorem and Polya's theorem \citep{polya1920zentralen} along with Assumption \ref{assump:a2} to bound $\nabla f_0(\cdot)$. As $n \to \infty$, this implies that $E_{P_{f_0}}\{f_K(x)\} \to f_0(x)$ since $w_n \to 0$ as $n \to \infty$. 

The variance may be dealt with in a similar manner. Following the same steps as before we get
\begin{align*}
    \mbox{var}_{P_{f_0}}\{f_{K}(x)\} & = \frac{1}{n} \, \mbox{var}_{P_{f_0}}\left\{\frac{1}{w_n^p} g_{\gamma_n - p +1}\left(\frac{x-X_{1}}{w_n}\right)\right\} \leq \frac{1}{n} \E\left\{\frac{1}{w_n^{2p}} g^2_{\gamma_n-p+1}\left(\frac{x-X_{1}}{w_n}\right)\right\}\\
    & \leq \frac{1}{nw_n^{2p}} \int_{[0,1]^p} g^2_{\gamma_n-p+1}\left(\frac{y-x}{w_n}\right) f_{0}(y) \, dy\\
    &\leq \frac{1}{nw_n^p} \int_{\left[-\frac{x}{w_n}, \frac{1-x}{w_n}\right]^p} g^2_{\gamma_n-p+1}(u) \{f_{0}(x) + w_n u^{\T} \nabla f_{0}(\xi)\} \, du,\\
 &   \leq \frac{f_0(x) \mathcal{O}_n(1)}{nw_n^p} ,  
\end{align*}
which shows that the variance goes to $0$ as $n \to \infty$, since $nw_n^p \to \infty$ as $n \to \infty$.
\end{proof}
For the nearest neighbor-Dirichlet mixture, recall $f_K(x) = (1/n)\sum_{i=1}^n t_{\gamma_n - p + 1}(x ; X_i, H_n)$ from Section \ref{sec:mean_and_variance} of the main document, where $H_n = h_n^2 \mathbbm{I}_p$ and $h_n^2 = \{\nu_n(\gamma_n - p + 1)\}^{-1} \{(\nu_n+1)(\gamma_0 - p + 1)\} \delta_0^2$. Here, the bandwidth $h_n$ satisfies $h_n \to 0$ and $nh_n^p \to \infty$ as $n \to \infty$. Lemma \ref{lemma:t_kde} then shows that $f_K(x)$ converges to $f_0(x)$ in $P_{f_0}$-probability as $n \to \infty$. 

\section{Cross-validation}\label{app:CVdetails}
\subsection{Algorithm for leave-one-out cross-validation}
Consider independent and identically distributed data $X_1, \ldots, X_n \in \mathbb{R}^p \sim f$ with $f$ having the nearest neighbor-Dirichlet mixture formulation. 
The prior of the neighborhood parameters $(\eta_i, \Sigma_i)$ following Sections \ref{sec:multivariate_gaussian} and \ref{sec:CV} is $(\eta_i, \Sigma_i) \sim \mathrm{NIW}_p(\mu_0, \nu_0, \gamma_0, \Psi_0)$ where $\Psi_0 = (\gamma_*\delta_0^2)\, \mathbbm{I}_{p}$ with $\gamma_{*} = \gamma_0 - p + 1$. We use the pseudo-posterior mean in \eqref{eq:multivariate_posterior_mean} to compute leave-one-out log-likelihoods $\mathbf{L}(\delta_0^2)$ for different choices of the hyperparameter $\delta_0^2$, choosing $\delta_{0, \text{CV}}^2 = \mbox{arg\,sup}_{\delta_0^2} \mathbf{L}(\delta_0^2)$ to maximize this criteria. The details of the computation of $\mathbf{L}(\delta_0^2)$ for a fixed $\delta_0^2$ are provided in Algorithm \ref{algo:CV}.
\begin{algorithm}[h]
\begin{itemize}
    \item Consider data $\mathcal{X}^{(n)} = (X_1, \ldots, X_n)$ where $X_i \in \mathbb{R}^p, \, p \geq 1$. 
    
    Fix the number of neighbors $k$ and other hyperparameters $\mu_0, \nu_0, \gamma_{0}$.
    \item For $i \in \{1,\ldots,n\}$, consider the data set leaving out the $i$th data point, given by $\mathcal{X}^{-i} = (X_1, \ldots, X_{i-1}, X_{i+1}, \ldots, X_n)$. Compute the pseudo-posterior mean density estimate at $X_i$, namely $\hat{f}_{-i}(X_i)$, using $\mathcal{X}^{-i}$ and \eqref{eq:multivariate_posterior_mean}; let $\mathbf{L}_i(\delta_0^2) = \hat{f}_{-i}(X_i)$. Finally, compute the leave-one-out log-likelihood given by
\begin{equation*}
\mathbf{L}(\delta_0^2) = \frac{1}{n} \sum_{i=1}^n \log\{\mathbf{L}_i(\delta_0^2)\}.
\end{equation*}
	\item For $\delta_{0}^2 > 0$, obtain $\delta_{0, \text{CV}}^2 = \underset{\delta_0^2}{\mbox{arg\,sup}} \, \mathbf{L}(\delta_0^2)$.
\end{itemize}
\caption{Leave-one-out cross-validation for choosing the hyperparameter $\delta_0^2$ in nearest neighbor-Dirichlet mixture method.}
\label{algo:CV}
\end{algorithm}

\subsection{Fast Implementation of cross-validation}
In Algorithm \ref{algo:CV}, the nearest neighborhood specification for each $\mathcal{X}^{-i}$ is different for $i=1,\ldots,n$. However, we bypass this computation by initially forming a neighborhood of size $(k+1)$ for each data point using the entire data and storing the respective neighborhood means and covariance matrices. Suppose for $X_i$, the indices of the $(k+1)$-nearest neighbors are given by $\widetilde{\mathcal{N}}_i = \{j \in \{1,\ldots,n\} : ||X_i - X_j||_2 \leq ||X_i - X_{i[k+1]}||_2\}$, arranged in increasing order according to their distance from $X_i$ with $X_{i[1]} = X_i$. Define the neighborhood mean $m_i = \{1/(k+1)\}\sum_{j\in \widetilde{\mathcal{N}}_i} X_{j}$ and the neighborhood covariance matrix $S_i = (k+1)^{-1} \{\sum_{j \in \widetilde{\mathcal{N}}_i} (X_{j} - m_i)(X_{j} - m_i)^{\T}\}$. Then, to form a $k$-nearest neighborhood for the new data $\mathcal{X}^{-i}$, a single pass over the initial neighborhoods $\widetilde{\mathcal{N}}_i$ is sufficient to update the new neighborhood means and covariance matrices. Below, we describe the update for the neighborhood means $m_j^{(-i)}$ and covariance matrices $S_j^{(-i)}$ for $j = 1, \ldots, n$ and $j \neq i$, considering the data $\mathcal{X}^{-i}$. For $j=1,\ldots,n$ and $j \neq i$, we have,
\begin{align}
\nonumber
    m_j^{(-i)} & = 
    \begin{cases}
  (1/k)\{(k+1)m_j - X_{j[k+1]}\} & \quad \mbox{if} \,\,\, i \notin \widetilde{\mathcal{N}}_j,\\
    (1/k)\{(k+1)m_j - X_i\} & \quad \mbox{if} \,\,\, i\in \widetilde{\mathcal{N}}_j.
    \end{cases}\\
   S_j^{(-i)} & = 
   \begin{cases}
   S_j - \{(k+1)/k\} (m_j - X_{j[k+1]})(m_j - X_{j[k+1]})^{\T} & \quad \mbox{if} \,\,\, i \notin \widetilde{\mathcal{N}}_j,\\
  S_j - \{(k+1)/k\} (m_j - X_i)(m_j - X_i)^{\T} & \quad \mbox{if} \,\,\, i \in \widetilde{\mathcal{N}}_j.
   \end{cases}
    \label{eq:nbd_update_mult}
\end{align}

\section{Algorithm with Gaussian Kernels for Univariate Data}\label{app:univariate_gaussian}

For $p=1$, we have a univariate Gaussian density $\phi(x; \eta_{i}, \sigma^2_{i})$ in neighborhood $i$ and normal-inverse gamma priors 
 $(\eta_i,\sigma_i^2) \sim \mbox{NIG}(\mu_0, \nu_0, \gamma_0/2, \gamma_0\delta_0^2/2)$ independently for $i=1,\ldots,n$, with $\mu_0 \in \mathbb{R}$ and $\nu_0,\,\gamma_0, \delta_0^2 > 0$. That is, $$\eta_i \mid \sigma_i^2 \sim \Gauss\left(\mu_0, \dfrac{\sigma_i^{2}}{\nu_0} \right), \quad \sigma_{i}^2 \sim \mbox{IG}\left(\dfrac{\gamma_0}{2}, \dfrac{\gamma_0 \delta_0^2}{2}\right).$$ 
 Monte Carlo samples from the pseudo-posterior of the unknown density $f$ at any point $x$ can be generated following the steps of Algorithm \ref{algo:NNDP_univariate}.

\begin{algorithm}
\begin{itemize}
\item \textbf{Step 1: }Compute the $k$-nearest neighborhood $\mathcal{N}_i$ for data point $X_i$ with $X_{i[1]} = X_i$, using the distance $d(\cdot, \cdot)$.
\item \textbf{Step 2: }Update the parameters for neighborhood $\mathcal{N}_i$ to $(\mu_i, \nu_n, \gamma_n / 2, \gamma_n \delta_i^2 / 2)$ where 
$\nu_n = \nu_0 + k$, $\gamma_n = \gamma_0 + k$, 
$$\mu_i  = \dfrac{\nu_0 \mu_0 + k  \bar{X}_{i}}{\nu_n}, \quad \bar{X}_{i} = \dfrac{1}{k} \sum_{j \in \mathcal{N}_i} X_j,$$ and $\gamma_n \delta_i^2  =  \gamma_0 \delta_0^2 + \sum_{j \in \mathcal{N}_i}\left(X_j - \bar{X}_{i}\right)^2 + k \nu_0 \nu_n^{-1} \left(\mu_0 - \bar{X}_{i} \right)^2.$

\item \textbf{Step 3: }To compute the $t$-th Monte Carlo sample of $f(x)$, sample Dirichlet weights $\pi^{(t)} \sim \mbox{Dirichlet}(\alpha+1, \ldots, \alpha+1)$ and neighborhood-specific parameters $(\eta_{i}^{(t)}, \sigma_{i}^{(t)2}) \sim  \mbox{NIG}\left(\mu_i,\,\nu_n, \,\gamma_n/2,\gamma_n \delta_i^2/2\right)$ independently for $i = 1, \ldots, n$, and set 
\begin{equation}
         f^{(t)}(x) = \sum_{i=1}^n \pi_{i}^{(t)} \phi(x; \eta_{i}^{(t)}, \sigma_{i}^{(t)2}).
     \end{equation}

\end{itemize}
\caption{Nearest neighbor-Dirichlet mixture algorithm to obtain Monte Carlo samples from the pseudo-posterior of $f(x)$ with Gaussian kernel and normal-inverse gamma prior.}
\label{algo:NNDP_univariate}
\end{algorithm}

\section{Inverse Wishart Parametrization}\label{app:invwishart}
The parametrization of the inverse Wishart density defined on the set of all $p \times p$ matrices with real entries used in this article is given as follows. Suppose $\gamma>p-1$ and $\Psi$ is a $p\times p$ positive definite matrix. If $\Sigma \sim \mbox{IW}_{p}(\gamma, \Psi)$, then $\Sigma$ has the following density function:
\begin{align*}
g(\Sigma) & = 
\begin{cases}
\dfrac{|\Psi|^{\gamma/2}}{2^{\gamma p/2}\Gamma_{p}\left(\dfrac{\gamma}{2}\right)}|\Sigma|^{-(\gamma+p+1)/2}\etr\left(-\dfrac{1}{2}\Psi \Sigma^{-1}\right) & \quad \mbox{if } \Sigma \mbox{ is positive definite},\\ \\
0 & \quad \mbox{otherwise},
\end{cases}
\end{align*}
where $\Gamma_{p}(\cdot)$ is the multivariate gamma function given by $$\Gamma_{p}(a) = \pi^{p(p-1)/4} \prod_{j=1}^{p} \Gamma\left(a+\frac{1-j}{2}\right),$$ 
for $a\geq (p-1)/2$ and the function $\etr{(A)} = \exp{\{\mbox{tr}(A)\}}$ for a square matrix $A$. When $p=1$, the $\text{IW}_p(\gamma, \Psi)$ density is the same as the $\text{IG}(\gamma/2, \gamma \delta^2/2)$ density, where $\delta^2 = \Psi / \gamma$. The $\mbox{IW}_p(\gamma, \Psi)$ distribution has mean $\Psi / (\nu - p - 1)$ for $\nu > p + 1$ and mode $\Psi / (\nu + p + 1)$.
\section{$\mathcal{L}_1$ Error Tables in Sections \ref{sec:univ_simulations} and \ref{sec:mvt_simulations}} \label{app:L1tables}

\begin{center}
\begin{table}[H]
\Huge
\centering
\scalebox{0.45}{
\begin{tabular}{cccccccccccc}\toprule
Sample size          & Estimator & CA    & CW    & DE  & GS  & IE  & LN  & LO  & SB  & SP  & ST  \\
\cmidrule{1-12}\\								
\multirow{3}{*}{200} & NN-DM     & 0.20 & 0.31 & 0.19 & 0.12 & 0.36 & 0.20 & 0.13 & 0.16 & 0.30 & 0.31 \\
& NN-DM (default)    & 0.21 & 0.37 & 0.17 & 0.12 & 0.34 & 0.20 & 0.14 & 0.17 & 0.31 & 0.32 \\
                     & DP-MC     & 0.17 & 0.37 & 0.14 & 0.10 & 0.36 & 0.22 & 0.13 & 0.23 & 0.27 & 0.55 \\
                     & KDE       & -    & 0.37 & 0.16 & 0.12 & -    & 0.18 & 0.11 & 0.18 & -    & 0.52 \\
                     & KNN       & 5.99 & 0.58 & 0.59 & 0.28 & 3.46 & 0.54 & 0.48 & 0.39 & 6.02 & 0.46 \\
                     & DP-VB     & 0.20 & 0.35 & 0.15 & 0.08 & 0.53 & 0.25 & 0.11 & 0.11 & 0.44 & 0.57 \\
                     & RD        & -    & 0.35 & 0.13 & 0.12 & -    & 0.16 & 0.11 & 0.16 & -    & 0.53 \\
                     & PTM       & 0.29 & 0.27 & 0.18 & 0.13 & 0.38 & 0.22 & 0.13 & 0.20 & 0.40 & 0.39 \\
                     & LLDE       & 0.19 & 0.36 & 0.14 & 0.10 & - & 0.15 & 0.10 & 0.18 & - & 0.55 \\
                     & OPT       & 0.32 & 0.36 & 0.28 & 0.17 & 0.55 & 0.21 & 0.02 & 0.18 & 0.75 & 0.52 \\
       & A-KDE       & 0.22 & 0.35 & 0.15 & 0.14 & 0.46 & 0.16 & 0.11 & 0.17 & 0.38 & 0.53 \\
                     
\cmidrule{1-12}\\
\multirow{3}{*}{500} & NN-DM     & 0.16 & 0.17 & 0.13 & 0.08 & 0.30 & 0.16 & 0.10 & 0.10 & 0.24 & 0.20 \\
& NN-DM (default)    & 0.16 & 0.36 & 0.12 & 0.09 & 0.30 & 0.17 & 0.10 & 0.12 & 0.25 & 0.22 \\
                     & DP-MC     & 0.11 & 0.35 & 0.10 & 0.08 & 0.27 & 0.18 & 0.09 & 0.13 & 0.22 & 0.53 \\
                     & KDE       & -    & 0.32 & 0.11 & 0.08 & -    & 0.15 & 0.08 & 0.11 & -    & 0.51 \\
                     & KNN       & 3.62 & 0.47 & 0.48 & 0.27 & 3.39 & 0.40 & 0.30 & 0.31 & 5.64 & 0.35 \\
                     & DP-VB     & 0.14 & 0.33 & 0.11 & 0.05 & 0.48 & 0.19 & 0.08 & 0.08 & 0.45 & 0.55 \\
                     & RD        & -    & 0.32 & 0.10 & 0.09 & -    & 0.13 & 0.09 & 0.10 & -    & 0.50 \\
                     & PTM       & 0.24 & 0.19 & 0.14 & 0.10 & 0.32 & 0.19 & 0.11 & 0.14 & 0.32 & 0.30 \\
                     & LLDE       & 0.17 & 0.35 & 0.11 & 0.08 & - & 0.15 & 0.08 & 0.14 & - & 0.53 \\
                     & OPT       & 0.27 & 0.31 & 0.16 & 0.12 & 0.51 & 0.16 & 0.01 & 0.14 & 0.72 & 0.46 \\
                      & A-KDE       & 0.16 & 0.32 & 0.10 & 0.10 & 0.40 & 0.13 & 0.10 & 0.10 & 0.27 & 0.52 \\
\cmidrule{1-12}
\end{tabular}
}
\caption{Comparison of the methods in terms of $\mathcal{L}_1$ error in the univariate case. Number of test points and replications considered are $n_t = 500$ and $R = 20$, respectively.}
\label{tab:univariate_tab}
\end{table}
\end{center}
\begin{sidewaystable}
\centering
\Huge
\scalebox{0.38}{
\begin{tabular}[H]{cccccccccccccccccccccccccc}\toprule
& Density & \multicolumn{4}{c}{MG} & \multicolumn{4}{c}{MST} & \multicolumn{4}{c}{MVC} & \multicolumn{4}{c}{MVG} & \multicolumn{4}{c}{SN} & \multicolumn{4}{c}{T}\\
\cmidrule{1-26}\\
Sample size & Dimension    & 2 & 3  &  4  &    6 & 2 & 3 &    4 &  6 & 2 & 3 &   4 & 6 &     2 & 3 & 4 &  6 & 2 & 3 & 4 & 6 & 2 & 3 & 4 & 6\\
\cmidrule{1-26}\\
\multirow{3}{*}{200} & NN-DM     & 0.29 & 0.39 & 0.46 & 0.63  & 0.27 & 0.37 & 0.43 & 0.59  & 0.33 & 0.48 & 0.52  & 0.62  & 0.33 & 0.50 & 0.61 & 0.74 & 0.23                  & 0.32                  & 0.40                  & 0.56                  & 0.26                  & 0.33                  & 0.38                  & 0.52                  \\
& NN-DM (default)     & 0.29 & 0.40 & 0.47 & 0.65  & 0.28 & 0.38 & 0.45 & 0.61 & 0.37 & 0.49 & 0.61  & 0.80  & 0.38 & 0.50 & 0.62 & 0.75 & 0.24                  & 0.34                  & 0.42                  & 0.58                  & 0.26                  & 0.33                  & 0.40                  & 0.53                  \\
                      & DP-MC     & 0.20 & 0.36 & 0.62 & 0.67 & 0.22 & 0.42 & 0.55 & 0.68  & 0.31 & 0.46 & 0.51  & 0.60  & 0.34 & 0.52 & 0.61 & 0.73 & 0.16                  & 0.18                  & 0.25                  & 0.32                  & 0.21                  & 0.26                  & 0.34                  & 0.44                  \\
                      & KDE       & 0.28 & 0.52 & 0.79 & 1.29  & 0.27 & 0.55 & 0.72 & 1.21  & -    & -    & -     & -     & 0.43 & 0.77 & 0.90 & 1.09 & 0.22                  & 0.50                  & 0.78                  & 1.26                  & 0.26                  & 0.53                  & 0.74                  & 1.21                  \\
                      & KNN       & 1.93 & 3.82 & 4.80 & 18.83 & 7.28 & 8.22 & 9.25 & 11.05 & 4.92 & 7.31 & 15.24 & 20.5  & 3.30 & 3.65 & 4.75 & 5.54 & 2.21                  & 5.16                  & 8.37                  & 10.04                 & 2.97                  & 6.37                  & 10.22                 & 17.54                 \\
                      & DP-VB     & 0.29 & 0.38 & 0.41 & 0.50  & 0.24 & 0.36 & 0.44 & 0.59  & 0.48 & 0.85 & 1.28  & 1.69  & 0.45 & 0.58 & 0.71 & 0.86 & 0.17                  & 0.23                  & 0.31                  & 0.52                  & 0.19                  & 0.29                  & 0.34                  & 0.46                  \\
                      & LLDE     & 0.32 & 0.52 & 0.93 & 1.82 & 0.28 & 0.67 & 1.02 & 11.38 & 0.53 & - &  - & -  & 0.32 & 0.45 & 0.75 & 0.97 & 0.16                  & 0.22                  & 0.29                  & 0.65                  & 0.17                  & 0.26                  & 0.35           & 2.05                \\
                      & OPT     & 0.58 & 0.79 & 1.04 & -  & 0.57 & 1.08 & 1.31 & -  & 0.59 & 0.82 &  1.01 & - & 0.43 & 0.69 & 0.86 & - & 0.48                  & 0.68                  & 0.88                  & -                  & 0.57                  & 0.78                  & 1.10           & -                 \\
\cmidrule{1-26}\\
\multirow{3}{*}{1000} & NN-DM     & 0.18 & 0.26 & 0.34 & 0.46  & 0.19 & 0.27 & 0.32 & 0.44  & 0.28 & 0.33 & 0.46  & 0.50  & 0.29 & 0.44 & 0.49 & 0.62 & 0.15                  & 0.22                  & 0.29                  & 0.42                  & 0.17                  & 0.24                  & 0.30                  & 0.38                  \\
& NN-DM (default)     & 0.22 & 0.30 & 0.37 & 0.48  & 0.21 & 0.29 & 0.33 & 0.44  & 0.31 & 0.41 & 0.52  & 0.66  & 0.36 & 0.48 & 0.55 & 0.68 & 0.22                  & 0.29                  & 0.36                  & 0.45                  & 0.22                  & 0.28                  & 0.33                  & 0.39                  \\
                      & DP-MC     & 0.08 & 0.39 & 0.57 & 0.58  & 0.11 & 0.18 & 0.21 & 0.47  & 0.26 & 0.39 & 0.48  & 0.54  & 0.21 & 0.38 & 0.51 & 0.64 & 0.06                  & 0.07                  & 0.10                  & 0.15                  & 0.09                  & 0.15                  & 0.17                  & 0.28                  \\
                      & KDE       & 0.16 & 0.32 & 0.52 & 0.96  & 0.16 & 0.35 & 0.53 & 1.04  & -    & -    & -     & -     & 0.33 & 0.66 & 0.73 & 0.93 & 0.13                  & 0.32                  & 0.53                  & 1.05                  & 0.14                  & 0.32                  & 0.52                  & 0.90                  \\
                      & KNN       & 0.92 & 2.62 & 4.01 & 15.28 & 5.96 & 6.48 & 7.04 & 9.63  & 4.68 & 6.29 & 13.7  & 17.04 & 2.01 & 2.59 & 3.88 & 4.09 & 1.89                  & 4.39                  & 6.84                  & 8.15                  & 2.37                  & 5.30                  & 9.66                  & 13.28                 \\
                      & DP-VB     & 0.25 & 0.29 & 0.33 & 0.36  & 0.15 & 0.24 & 0.25 & 0.45  & 0.42 & 0.74 & 0.82  & 0.91  & 0.38 & 0.54 & 0.61 & 0.77 & 0.10                  & 0.12                  & 0.15                  & 0.20                  & 0.10                  & 0.15                  & 0.18                  & 0.31\\
                  & LLDE     & 0.31 & 0.37 & 0.52 & 1.02  & 0.22 & 0.40 & 0.46 & 1.18 & 0.47 & - & -  & -  & 0.29 & 0.39 & 0.51 & 0.94 & 0.07                  & 0.10                  & 0.14                  & 0.27                  & 0.10                  & 0.15                  & 0.23                  & 0.38\\
                  & OPT     & 0.39 & 0.72 & 1.00 & -  & 0.43 & 0.84 & 1.10 & -  & 0.51 & 0.72 & 0.89 & - & 0.30 & 0.60 & 0.85 & - & 0.31                  & 0.50                  & 0.76                  & -                  & 0.33                  & 0.53                  & 0.94                  & -\\
 \cmidrule{1-26}
\end{tabular}
}
\caption{Comparison of the methods in terms of $\mathcal{L}_1$ error in the multivariate case. Number of test points and replications considered are $n_t = 500$ and $R = 20$, respectively.}
\label{tab:multivariate_tab}
\end{sidewaystable}
\newpage
\bibliography{ref_library}
\bibliographystyle{apalike}

\end{document}